\documentclass{amsart}

\usepackage[utf8]{inputenc}
\usepackage[T1]{fontenc}
\usepackage{lmodern}
\usepackage[UKenglish]{babel}
\usepackage{comment}
\usepackage{hyperref}
\usepackage{amsmath,amsfonts,ifthen,amssymb,accents}
\usepackage{mathtools}

\usepackage{fdsymbol}
\usepackage{wrapfig}

 \usepackage{tikz}

\definecolor{gruen}{rgb}{0,0.8,0.2}
\definecolor{rot}{rgb}{0.7,0,0}

\usepackage{empheq}

\usepackage{microtype}
\colorlet{fillA}{gray!50}
\colorlet{fillB}{gray!15}

\usepackage[inline]{enumitem}
\newtheorem{proposition}{Proposition}[section]
\newtheorem{lemma}[proposition]{\textbf{Lemma}}
\newtheorem{remark}[proposition]{\textbf{Remark}}

\newtheorem{corollary}[proposition]{Corollary}
\newtheorem{example}[proposition]{Example}
\newtheorem{theorem}[proposition]{\textbf{Theorem}}

\newtheorem{conjecture}[proposition]{Conjecture}

\newcommand\drop[1]{}

\renewcommand\c{\gamma}

\newcommand{\R}{\ensuremath{\mathbb{R}}}
\newcommand{\N}{\ensuremath{\mathbb{N}}}

\def\qed{\hspace*{\fill}$\Box$}

\newtheorem{definition}[proposition]{\textbf{Definition}}
\newtheorem{notation}[proposition]{Notation}

\def\showlabel#1{}
\def\showfiglabel#1{}

\def\myclaim#1#2\par{{\medbreak\noindent\rlap{\rm(#1)}\ignorespaces
 \rightskip20pt
 \hangindent=20pt\hskip20pt{\ignorespaces\sl#2}\smallskip}}
\def\junk#1{}

\newcommand{\omitQ}{\textit{omit}}

\newcommand{\dtw}{\textit{dtw}}
\newcommand{\sth}{\mathrel : }

\newcommand{\AAA}{\mathcal{A}}
\newcommand{\BBB}{\mathcal{B}}
\newcommand{\CCC}{\mathcal{C}}

\newcommand{\FFF}{\mathcal{F}}
\newcommand{\GGG}{\mathcal{G}}
\newcommand{\HHH}{\mathcal{H}}
\newcommand{\JJJ}{\mathcal{J}}
\newcommand{\III}{\mathcal{I}}
\newcommand{\LLL}{\mathcal{L}}
\newcommand{\MMM}{\mathcal{M}}

\newcommand{\OOO}{\mathcal{O}}
\newcommand{\PPP}{\mathcal{P}}
\newcommand{\QQQ}{\mathcal{Q}}
\newcommand{\RRR}{\mathcal{R}}
\newcommand{\SSS}{\mathcal{S}}
\newcommand{\TTT}{\mathcal{T}}
\newcommand{\YYY}{\mathcal{Y}}
\newcommand{\UUU}{\mathcal{U}}
\newcommand{\VVV}{\mathcal{V}}
\newcommand{\WWW}{\mathcal{W}}

\newcommand{\ZZZ}{\mathcal{Z}}

\newcommand{\HH}{\mathfrak{H}}
\newcommand{\QQ}{\mathfrak{Q}}
\newcommand{\RR}{\mathfrak{R}}

\newcommand{\bn}{\textup{bn}}
\newcommand{\isom}{\cong}

\newcommand{\Bot}{\textit{bot}}
\newcommand{\Top}{\textit{top}}
\newcommand{\Ssplit}{\ensuremath\SSS_{\textit{split}}}
\newcommand{\Sseg}{\ensuremath\SSS_{\textit{seg}}}

\newenvironment{cenv}{\begin{list}{}{%
      \setlength{\labelwidth}{1.5em}%
      \setlength{\leftmargin}{\labelwidth}%
      \addtolength{\leftmargin}{\labelsep}%
      \setlength{\listparindent}{0em}%
      \setlength{\topsep}{10pt}%
      \setlength{\itemsep}{5pt}%
      \setlength{\parsep}{0pt}%
    }
  }{
  \end{list}
}

\newcounter{claimcounter}
\newcounter{conditioncounter}

\newenvironment{Claim}{
  
  \refstepcounter{claimcounter}
  \begin{cenv}
  \item[{Claim \arabic{claimcounter}.}]
  }{
  \end{cenv}
}

\newenvironment{ClaimProof}[1][]{\noindent{%
\ifthenelse{\equal{#1}{}}{{\sl Proof.\ }}{{\sl #1.\ }}%
}}{\hspace*{1em}\nobreak\hfill$\dashv$\endtrivlist\addvspace{2ex plus
0.5ex minus0.1ex}}

\newenvironment{recap}{\par\smallskip\noindent\textbf{Recap. }}%
{\hspace*{1em}\nobreak\hfill$\lrcorner$\endtrivlist\addvspace{2ex plus
0.5ex minus0.1ex}}

\newcommand{\Fpath}{f_1}

\newcommand{\Fcleanpath}{f_3}
\newcommand{\Fclique}{f_\textit{clique}}

\renewcommand\drop[1]{}

\newcommand{\row}{\textup{row}}

\renewenvironment{proof}[1][]%
{\setcounter{claimcounter}{0}\ifthenelse{\equal{#1}{}}{\noindent\textit{Proof.
    }}{\noindent\textit{#1. }}}%
{\qed\par\bigskip}

\def\showlabel#1{}
\let\oldlabel\label
\renewcommand{\label}[1]{\oldlabel{#1}\marginpar{lab: #1}}
\makeindex

\excludecomment{longdoneproof}

\begin{document}

\title{The Directed Grid Theorem}

\author{Ken-ichi Kawarabayashi}
\author{Stephan Kreutzer}

\address{National Institute of Informatics, 2-1-2 Hitotsubashi, Chiyoda-ku, Tokyo, Japan}
\thanks{Ken-ichi Kawarabayashi's research is partly supported  by JST ERATO Kawarabayashi Large Graph Project and by JSPS KAKENHI Grant Number JP18H05291 and JP20A402.}
\email{k\_keniti@nii.ac.jp}
\address{Chair for Logic and Semantics, Technical University Berlin, Sekr TEL 7-3, Ernst-Reuter Platz 7, 10587 Berlin, Germany}
\thanks{Stephan Kreutzer's research is partly supported by DFG Emmy-Noether
Grant \emph{Games} and by the
European Research Council (ERC) under the European Union’s Horizon
2020 research and innovation programme (grant agreement No 648527).
}
\email{stephan.kreutzer@tu-berlin.de}

\thanks{An extended abstract appeared in the Proceedings of the 47th
  ACM Symposium on Theory of Computing (STOC 2015). This paper also
	combines some proofs from \cite{KawarabayashiKK14,KawarabayashiK14}.}

\begin{abstract}
  The  grid theorem, originally proved by Robertson and
  Seymour in Graph Minors V in 1986, is one of the most central results in the
  study of graph minors. It has found numerous applications in
  algorithmic graph structure theory, for instance in
  bidimensionality theory, and it is the basis for several other
  structure theorems developed in the graph minors project.

  In the mid-90s, Reed and  Johnson, Robertson, Seymour and
  Thomas (see \cite{Reed97,JohnsonRobSeyTho01}), independently, conjectured an analogous theorem
  for directed graphs, i.e.~the existence of a function $f\sth\N
  \rightarrow \N$ such that every digraph of directed tree-width at
  least $f(k)$ contains a directed grid of order $k$. In an
  unpublished manuscript from 2001, Johnson, Robertson, Seymour and
  Thomas gave a proof of this conjecture for planar digraphs. But for
  over 15 years, this was the most general case proved for the
  Reed, Johnson, Robertson, Seymour and Thomas conjecture.

In this paper, nearly two decades after
  the conjecture was made, we are finally able to confirm the
  Reed, Johnson, Robertson, Seymour and Thomas conjecture in full
  generality and to prove the directed grid theorem.

  As consequence of our results we have several 
  results using our directed grid theorem. 
  For example, we are able to improve results in Reed
  et al.\,in  1996 \cite{ReedRST96} (see also \cite{opg}) on disjoint
  cycles of length at least $l$. We expect many more algorithmic
  results to follow from the grid theorem.
\end{abstract}

\newif\ifproofmode

\maketitle

\clearpage
\setcounter{page}{1}

\section{Introduction}

Structural graph theory has proved to be a powerful tool for coping
with computational intractability. It provides a wealth of concepts
and results that can be used to design efficient algorithms for hard
computational problems on specific classes of graphs occurring
naturally in applications. Of particular importance is the concept of
\emph{tree-width}, introduced by Robertson and Seymour as part of
their seminal graph minor series \cite{GM-series}\footnote{Strictly speaking, Halin \cite{halin} came up with the same notion in 1976, but
it went unnoticed until it was rediscovered by Robertson and Seymour \cite{GMV} in 1984.}. Graphs of small
tree-width can recursively be
decomposed into subgraphs of constant size which can be combined in a
tree like way to yield the original graph. This property allows to use
algorithmic techniques such as dynamic programming, divide and conquer etc.~to solve
many hard computational problems efficiently on graphs of small
tree-width. In this way, a huge number of problems has been shown to
become tractable, e.g. solvable in linear or polynomial time, on graph
classes of bounded tree-width. See
e.g.~\cite{Bodlaender96a,Bodlaender97,Bodlaender05,DowneyF13} and
references
therein.
But methods from structural graph theory, especially graph minor
theory, also provide a powerful and vast toolkit of concepts and ideas
to handle graphs of large tree-width and to understand their
structure.

One of the most fundamental theorems in this context is the \emph{grid theorem}, proved by Robertson and Seymour in \cite{GMV}. It
states that there is a function $f \sth \N \rightarrow \N$ such
that every graph of tree-width at least $f(k)$ contains a $k\times
k$-grid as a minor. The known upper bounds on this function $f(k)$, initially
being enormous, have subsequently been improved and
are now polynomial \cite{ChekuriC14,chuzoy}.
The grid theorem is important both for structural graph
theory as well as for algorithmic applications. For instance, algorithmically it is the
basis of an algorithm design principle called \emph{bidimensionality
  theory}, which has been used to obtain many approximation
algorithms, PTASs, subexponential algorithms and fixed-parameter
algorithms on graph classes excluding a fixed minor. See
\cite{DemaineHaj08,DemaineHaj08b,DemaineH04,DemaineH05,FominLST10,FominLRS11}
and references therein.

Furthermore, the grid theorem also
plays a key role in Robertson and Seymour's graph minor algorithm and
their solution to the disjoint paths problem \cite{GMXIII} (also see \cite{kkr})
in  a technique known
as the \emph{irrelevant vertex technique}. Here, a problem is solved by
showing that it can be solved efficiently on graphs of small tree-width
and otherwise, i.e.\,if the tree-width is large and therefore the graph
contains a large grid, that a vertex deep in the
middle of the grid is irrelevant for the problem solution and can
therefore be deleted.
This yields a natural recursion that eventually
leads to the case of small tree-width. Such applications also appear in some other problems, see \cite{topo,kr,kkr,kleinberg}. 

Furthermore, with respect to graph structural aspects, the excluded grid theorem is the
basis of the seminal structure and decomposition theorems in graph minor theory
such as in \cite{GMXVI}.

The structural parameters and techniques discussed above all relate to undirected
graphs. However, in various applications in computer science, the most
natural model are directed graphs. Given the enormous success width
parameters had for problems defined on undirected graphs, it is
natural to ask whether they can also be used to analyse the structure
of digraphs and the complexity of
NP-hard problems on digraphs. In principle it is possible to
apply the structure theory for undirected graphs to directed graphs by
ignoring the direction of edges. However, this implies an
information loss and may fail to properly distinguish between simple
and hard input instances (for example, the disjoint paths problem is
NP-complete for directed graphs even with only two source/terminal pairs \cite{FortuneHW80}, yet it is solvable in polynomial time for
any fixed number of terminals for undirected graphs \cite{kkr,GMXIII}). Hence, for computational problems whose
instances
are digraphs, methods based on undirected graph structure theory may be less useful.

As a first step towards a structure theory specifically for directed
graphs, Reed \cite{Reed99} and Johnson, Robertson, Seymour and Thomas
\cite{JohnsonRobSeyTho01} proposed a concept of
\emph{directed tree-width}  and showed that the $k$-disjoint paths
problem is solvable in polynomial time for any fixed $k$ on any class
of graphs of bounded directed tree-width.
Reed \cite{Reed97} and Johnson et al. \cite{JohnsonRobSeyTho01} also conjectured a directed analogue of the
grid theorem.

\begin{conjecture} \textbf{\upshape(Reed and Johnson, Robertson, Seymour,  Thomas)}
  There is a function $f\sth \N\rightarrow \N$
 such that every digraph of directed tree-width at least $f(k)$
 contains a cylindrical grid of order $k$ as a butterfly minor
\end{conjecture}
Actually, according to \cite{JohnsonRobSeyTho01}, this conjecture was
formulated by Ro\-bertson, Seymour and Thomas, together with
Alon and Reed at
a conference in Annecy, France in 1995.
Here, a \emph{cylindrical grid} consists of $k$
concentric directed cycles and $2k$ paths connecting the cycles in
alternating directions.
See Figure~\ref{fig:grid} for an illustration and
Definition~\ref{def:cyl-grid} for details. A \emph{butterfly minor} of a
digraph $G$ is a digraph obtained from a subgraph of $G$ by
contracting edges which are either the only outgoing edge of their
tail or the only incoming edge of their head. See
Definition~\ref{def:butterfly}  for details.

In an unpublished manuscript, Johnson et
al. \cite{JohnsonRobSeyTho01b} proved the conjecture for planar
digraphs. Very recently, we started working on
this conjecture and made some progress: in
\cite{KawarabayashiK14},
this result was generalised to all
classes of directed graphs excluding a fixed undirected graph as an
undirected minor. This includes classes of digraphs of
bounded genus. Another related result was established in
\cite{KawarabayashiKK14}, where a half-integral directed grid
theorem was proved. More precisely, it was shown that there is a
function $f\sth \N\rightarrow \N$ such that every digraph $G$ of
directed tree-width at least $f(k)$ contains a half-integral grid of
order $k$. Here, essentially, \emph{a half-integral grid} in a digraph $G$ is
a cylindrical grid in the digraph obtained from $G$ by duplicating
every vertex, i.e.~adding for each vertex an isomorphic copy with the
same in- and out-neighbours.
 However, despite the
conjecture being open for nearly 20 years now, no progress beyond the
results mentioned before has been obtained.
The main result of this paper, building on
\cite{ReedRST96,KawarabayashiK14,KawarabayashiKK14}, is to finally
solve this long standing open problem.

\begin{theorem}\label{thm:main}
  There is a function $f\sth \N\rightarrow \N$
  such that every digraph of directed tree-width at least $f(k)$
  contains a cylindrical grid of order $k$ as a butterfly minor.
\end{theorem}

\begin{figure}
  \begin{center}
    \includegraphics{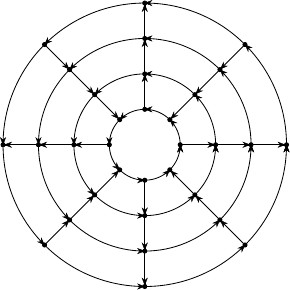}
    \includegraphics{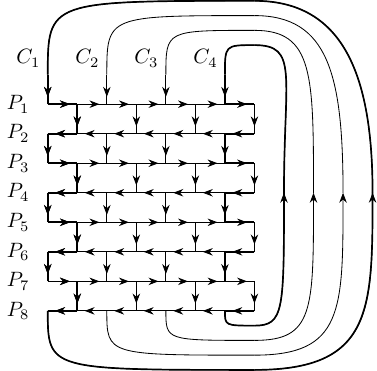}
    \addtolength{\textfloatsep}{-200pt}
    \caption{Cylindrical grid $G_4$ and the cylindrical wall of order
      $4$. The perimeters of the wall are depicted using thick edges.}\label{fig:grid}\label{fig:wall}%
  \end{center}
\end{figure}

 We believe that this grid theorem for digraphs is a first but
important step
towards a more general structure theory for directed graphs based on
directed tree-width, similar to the grid theorem for
undirected graphs being the basis of more general structure
theorems. Furthermore, it is likely that the duality of directed
tree-width and directed grids will make it possible to
develop algorithm design techniques such as bidimensionality theory or
the irrelevant vertex technique for directed graphs. We are
particularly optimistic that this approach will prove useful for
algorithmic versions of Erd\H os-P\'osa type results and in
the study of the directed disjoint paths problem.
The half-integral directed grid theorem
in \cite{KawarabayashiKK14} has been used to show that a variant of
the quarter-integral directed disjoint paths problem can be solved in
polynomial time. It is conceivable that our grid theorem here will
allow us to show that the half-integral directed disjoint paths
problem can be solved in polynomial time. Here, the half-integral
directed disjoint paths problem is the problem to decide for a given
digraph $G$ and $k$ pairs $(s_1, t_1), \dots, (s_k, t_k)$ of vertices
whether there are directed paths $P_1, \dots, P_k$ such that
$P_i$ links $s_i$ to $t_i$ and such that no vertex of $G$ is contained
in more than two paths from $\{P_1, \dots, P_k\}$. While we are optimistic that the
directed grid theorem will provide the key for proving that
the problem is solvable in polynomial
time, this requires much more work and significant new ideas and we
leave this for future work. 
Note that in a sense, half-integral
disjoint paths are the best we can hope for, as the directed disjoint
paths problem is NP-complete even for only $k=2$
source/target pairs \cite{FortuneHW80}.

However, the directed grid theorem may also prove relevant
for the integral directed disjoint paths problem. In a recent
breakthrough, Cygan et al.\,\cite{CyganMPP13} showed that the planar
directed disjoint paths problem is fixed-parameter tractable using an
irrelevant vertex technique (but based on a different type of directed
grid). They show that if a planar digraph contains a
grid-like subgraph of sufficient size, then one can delete a vertex in
this grid without changing the solution. The bulk of the paper then
analyses what happens if such a grid is not present. If one could
prove a similar irrelevant vertex rule for the directed grids used in
our paper, then the grid
theorem would immediately yield the dual notion in terms of
directed tree-width for free. The directed disjoint paths problem
beyond planar graphs therefore is another prime algorithmic
application we
envisage for directed grids.

Since the first version of this paper appeared in STOC'15, building on our directed grid theorem, much progress has been made with respect to  directed structure results and with respect to  algorithmic applications (including the half disjoint paths problem).  
In the conclusion section at the end of this paper, we give several results that build on our directed grid theorem, including the directed flat wall theorem and the tangle-tree theorem, using the directed grid theorem. 
Moreover, we also made progress towards the half disjoint paths problem.

Another obvious application of our result is to
Erd\H os-P\'osa type results such as Younger's conjecture proved by
Reed et al.\,in  1996 \cite{ReedRST96}. In fact, in their proof of
Younger's conjecture, Reed et al. construct a version of a directed
grid. This technique was indeed a primary motivation for  considering
directed tree-width and a directed grid minor as a proof of the
directed grid conjecture would yield a simple proof for
Younger's conjecture.
In fact 
our grid theorem implies the following stronger result than
Reed et al.\,in  1996 \cite{ReedRST96} (see also \cite{opg}):
for every $\ell$ and every integer $n\geq 0$, there exists an integer $t_n=t_n(\ell)$
such that for every digraph $G$, either $G$ has  $n$
pairwise vertex disjoint directed cycles of length at least $\ell$ or
there exists a set $T$ of at most $t_n$ vertices such that $G-T$
has no directed cycle of length at least $\ell$. Namely, we can also
impose the condition on the length of directed cycles, while
the proof of Reed et al.\ does not imply this statement.

The undirected version was proved by Birmel\'e, Bondy and Reed~\cite{BBR}, and
very recently, Havet and Maia~\cite{HM} proved the case $ \ell=3 $ for directed graphs.

\medskip

\noindent\textbf{Organisation and high level overview of the proof
  structure. }
In Section~\ref{sec:dtw},
we state our main result and present relevant definitions. In
Sections~\ref{sec:bramble} to~\ref{sec:cylindrical-grids}, then, we
present the proof
of our main result.

At a very high level, the proof works as
follows. It was already shown in \cite{Reed99} that if a
digraph $G$ has high directed tree-width, it contains a \emph{directed
bramble} of very high order (see Section~\ref{sec:dtw}). From this
bramble one either gets a subdivision of a suitable form of a directed
clique, which contains the cylindrical grid as butterfly minor, or one
can construct a structure that we call a \emph{web} (see
Definition~\ref{def:web-linkedness}).

Our main technical contributions of this paper are  in Sections~\ref{sec:web-to-grid} and~\ref{sec:cylindrical-grids}.
In
Section~\ref{sec:web-to-grid} we show that this web can be ordered and
rerouted to obtain a nicer version of a web called a
\emph{fence}. Actually, we need a much stronger property for this fence.
Let us observe that a
fence is essentially a cylindrical grid with one edge of each cycle
deleted.  In Section~\ref{sec:web-to-grid},
we also prove that there is a linkage from the bottom of the fence back to its top (in addition, we require some other properties that are
too technical to state here).

Hence, in order to obtain a cylindrical grid, all that is needed is
to find such a linkage that is
disjoint from (a sub-fence of) the fence.  The biggest problem here is that
the linkage from the bottom of the fence back to its top can
go anywhere in the fence. Therefore, we cannot get a sub-fence that is disjoint from this linkage.
This means that we have to create a cylindrical
grid from this linkage, together  with some portion of the fence.
This, however, is by far the most difficult
part of the proof, which we present in
Section~\ref{sec:cylindrical-grids}.

Let us mention that our proof is constructive in the sense that we can
obtain the following theorem, which may be of independent interest.
\begin{theorem}\label{thm:cor}
  There is a function $g\sth \N\rightarrow \N$
  such that given any directed graph and any fixed constant $k$, in polynomial time, we can obtain either
  \begin{enumerate}
  \item
  a cylindrical grid of order $k$ as a butterfly minor, or
  \item
  a directed tree decomposition of width at most $g(k)$.
  \end{enumerate}
\end{theorem}
 Note that the second conclusion follows from the result in
 \cite{JohnsonRobSeyTho01b}, which says that for fixed $l$, there is
 a polynomial time algorithm\footnote{actually, the time complexity is $f(k)n^c$ for some absolute constant $c$ that does not depend on $k$}
 to construct a directed tree
 decomposition of a given directed graph $G$ of width $3l$, if $G$ has
 directed tree-width at most $l$. So for Theorem~\ref{thm:cor}, if we
 set $g(k) = 3\cdot f(k)$, where $f(k)$ is the function of
 Theorem~\ref{thm:main}, then if the directed tree-width of a given directed graph is at least $f(k)$,
 we obtain the first conclusion from the constructive proof of
 Theorem~\ref{thm:main}. Otherwise, we obtain the second conclusion by
 the result in \cite{JohnsonRobSeyTho01b}.

\medskip

\medskip

\noindent\textbf{Acknowledgement. }
We would like to thank Julia Chuzhoy as well as an anonymous STOC
referee for reading an earlier full version of this paper,
and suggesting useful improvements for the presentation. Moreover, we would like to thank the referee who pointed out many small mistakes. The referee's patience leads to much better shape of this paper.

\section{Preliminaries}
\label{sec:prelims}

In this section we fix our notation and briefly review relevant
concepts from graph theory. We refer to, e.g.,~\cite{Diestel05} for
background. For any $n\in \N$ we define $[n] := \{ 1, \dots, n \}$.
For any set $U$ and $k\in \N$ we define $[U]^{\leq k} := \{ X\subseteq
U \sth |X|\leq k \}$. We define $[U]^{=k}$ etc. analogously. We write
$2^U$ for the power set of $U$.

\subsection{Background from graph theory.}

Let $G$ be a digraph. We refer to its vertex set by $V(G)$ and its
edge set by $E(G)$. If $(u,v)\in E(G)$ is an edge then $u$ is its
\emph{tail} and $v$ its \emph{head}.
Unless stated explicitly otherwise, all paths in this paper are
directed. We therefore simply write \emph{path} for \emph{directed path}.

\begin{notation}
  The following non-standard notation will be used frequently
  throughout the paper. If $Q_1$ and $Q_2$ are paths and $e$ is an
  edge whose tail is the last vertex of $Q_1$ and whose head is the
  first vertex of $Q_2$ then $Q_1eQ_2$ is the path $Q = Q_1+e+Q_2$
  obtained from concatenating $e$ and $Q_2$ to $Q_1$. We will usually
  use this notation in reverse direction and, given a path $Q$ and an
  edge $e\in E(Q)$, write ``Let $Q_1$ and $Q_2$ be subpaths of $Q$ such
  that $Q=Q_1 e Q_2$.'' Hereby we define the subpath $Q_1$ to be the
  initial subpath of $Q$ up to the tail of $e$ and $Q_2$ to be the
  suffix of $Q$ starting at the head of $e$.
\end{notation}
In this paper we will work with a version of directed minors known as
\emph{butterfly minors} (see
\cite{JohnsonRobSeyTho01}).

\begin{definition}[butterfly minor]\label{def:butterfly}
  Let $G$ be a digraph. An edge $e = (u,v)\in E(G)$ is
  \emph{butterfly-contractible} if $e$ is the only outgoing edge of
  $u$ or the only incoming edge of $v$. In this case the graph $G'$
  obtained from $G$ by butterfly-contracting $e$ is the graph with
  vertex set $(V(G) - \{u,v\}) \cup \{x_{u,v}\}$, where $x_{u,v}$ is a
  fresh vertex. The edges of $G'$ are the same as the edges of $G$ except for the edges
  incident with $u$ or $v$. Instead, the new vertex $x_{u,v}$ has the
  same neighbours as $u$ and $v$, eliminating parallel edges. A
  digraph $H$ is a \emph{butterfly-minor} of $G$ if it can be obtained
  from a subgraph of $G$ by butterfly contraction.
\end{definition}

\begin{figure}
  \includegraphics{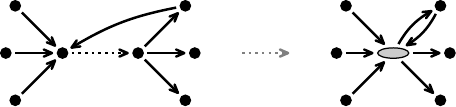}
	\caption{Butterfly contracting the dotted edge in the digraph on the left.}
	\label{fig:butterfly}
\end{figure}
See Figure~\ref{fig:butterfly} for an illustration of butterfly contractions.
We illustrate butterfly-contractions by the following example, which
will be used frequently in the paper.

\begin{example}\label{ex:butterfly} Let $G$ be a digraph.
  Let $P = P_1eP_2$ be a directed path in $G$ consisting of two
    subpaths $P_1, P_2$ joined by an edge $e$ with tail in $P_1$ and
    head in $P_2$. If every edge in $E(G)\setminus E(P)$ incident to a
    vertex $v\in V(P_1)$ has $v$ as its head and every edge in
    $E(G)\setminus E(P)$ incident to a vertex $u\in (P_2)$ has $u$ as
    its tail then $P$ can be butterfly-contracted into a single
    vertex, as then every edge in $E(P_1)$ is the only outgoing edge
    of its tail and every edge in $E(P_2)\cup \{e \}$ is the only
    incoming edge of its head.
\end{example}

We will also use the well-known concept of subdivisions.

\begin{definition}[subdivision]
  Let $G$ be a digraph. A digraph $H$ is a \emph{subdivision} of $G$
  if $H$ can be obtained from $G$ by replacing a set $\{e_1, \dots,
  e_k\} \subseteq E(G)$ of edges by directed paths $P_1, \dots, P_k$
  such that if $e_i = (u, v)$ then $P_i$ links $u$ to $v$ and $P_i$ is
  internally vertex disjoint from $G\cup \bigcup \{ P_j \sth i\not= j \}$.
\end{definition}

\begin{definition}[intersection graph]\label{def:intersection-graph}
  Let $\PPP$ and $\QQQ$ be sets of pairwise disjoint paths in a digraph
  $G$. The \emph{intersection graph $\III(\PPP, \QQQ)$} of $\PPP$ and
  $\QQQ$ is the bipartite (undirected) graph with vertex set $\PPP\cup\QQQ$ and an edge
  between $P\in \PPP$ and $Q\in \QQQ$ if $P\cap Q\not=\emptyset$.
\end{definition}

We will also frequently use Ramsey's
theorem (see e.g.\cite{Diestel05}).

\begin{theorem}[Ramsey's Theorem]\label{thm:ramsey}
  For all integers $q,l,r\geq 1$, there exists a (minimum) integer
  $R_{l}(r,q)\geq 0$ so that
  if $Z$ is a set with $|Z|\geq R_{l}(r,q)$, $Q$ is a set of
  $q$ colours and
  $h \sth [Z]^l \rightarrow Q$ is a function assigning a colour from
  $Q$ to every $l$-element subset of $Z$ then
  there exist
  $T\subseteq Z$ with $|T|=r$ and $x\in Q$ so that $h(X)=x$ for all
  $X \in [T]^l$.
\end{theorem}

We will also need the following lemmas adapted from
\cite{KreutzerT12}. By $K_s$, for some $s\geq 1$, we denote the (up to isomorphism)
complete graph on $s$ vertices.

\begin{lemma}\label{lem:clique-vertex}
  For all integers $n, k\geq 0$, if
  $G := K_{n\cdot (2k+1)}$ and $\gamma\sth V(G) \rightarrow [V(G)]^{\leq k}$
  such that $v\not \in \gamma(v)$ for all $v\in V(G)$ then
  there is $H \cong K_n \subseteq G$ such that $\gamma(v) \cap V(H) =
  \emptyset$ for all $v\in V(H)$.
\end{lemma}
\begin{proof}
  We construct the following auxilliary graph $A$: $V(A) = V(G)$ and
  for all $v\in V(A)$ we add an edge $\{v, u\}$ for every
  $u\in\gamma(v)$.
  By construction, for every $S\subseteq V(A)$, the subgraph $A[S]$
  contains at most $|S|\cdot k$ edges and
  therefore contains a vertex of degree at most $2k$. In other words, $A$ is $2k$-degenerate.

  As $|V(A)| = n(2k+1)$, $A$ contains an independent set $I\subseteq
  V(A)$ of size $n$ and $H := G[I]$ satisfies the condition of the lemma.
\end{proof}

\begin{lemma}\label{lem:clique}
  There is a computable function $\Fclique \sth \N\times \N\rightarrow \N$
  such that for all $n, k\geq 0$, if
  $G := K_{\Fclique(n,k)}$ and $\gamma\sth E(G) \rightarrow [V(G)]^{\leq k}$
  such that $\gamma(e) \cap e = \emptyset$ for all $e\in E(G)$ then
  there is $H \cong K_n \subseteq G$ such that $\gamma(e) \cap V(H) =
  \emptyset$ for all $e \in E(H)$.
\end{lemma}
\begin{proof}
  Let $R(n) := R_2(n, 2)$ denote the $n$-th Ramsey number as defined above.
  We define the function $\Fclique$ inductively as follows. For all
  $n\geq 0$ let $\Fclique(n, 0) := n$ and for $k>0$ let
  \[
     \Fclique(n,k) := (k+1)\cdot R\big(\max \{ \Fclique(n, k-1),
     \Fclique(n-1, k)  \}\big) + 1.
  \]

  We prove the lemma by induction on $k$.
  For $k=0$ there is nothing to show.
  So let $k>0$. 
  Choose a  vertex $v \in V(G)$. For each $u \in V(G) \setminus \{ v \}$ let $\eta(u)
  := \gamma(v, u)$. By assumption, $u \not\in \eta(u)$ and thus, as $|V(G)| \geq
  (k+1)\cdot R\big(\max \{ \Fclique(n, k-1), \Fclique(n-1, k)  \}\big)
  +1$, we can
  apply Lemma~\ref{lem:clique-vertex} to $G-v$ and $\eta$ to obtain a set $U
  \subseteq V(G) \setminus \{ v \}$ of order $R\big(\max \{ \Fclique(n, k-1),
  \Fclique(n-1, k)  \}\big)$ such that $U\cap \gamma(\{v, u\}) =
  \emptyset$ for all $u \in U$.

  Let $G_v := G[U]$, the subgraph of $G$ induced by $U$. We colour an
  edge $e$ in $G_v$ by $v$ if $v \in \gamma(e)$ and by $\bar v$
  otherwise. Let $l := \max \{ \Fclique(n, k-1), \Fclique(n-1, k) \}$. By
  Ramsey's theorem, as $|G_v| = R(l)$
  there is a set $X\subseteq V(G_v)$ of size $l$ such that
  all edges between elements of $X$ are coloured $v$ or there is such
  a set where all edges are coloured $\bar v$.

  In the first case, let $G'$ be the subgraph of $G_v$ induced by
  $X$ and let $\gamma'(e) := \gamma(e) \setminus \{ v\}$, for all
  edges $e\in E(G')$. Then, $|\gamma'(e)| \leq k-1$ for all $e\in
  E(G')$ and as $|X|\geq \Fclique(n, k-1)$, we can apply the induction
  hypothesis to find the desired clique $H\cong K_n$  in $G'$.

  So suppose $X$ induces a subgraph where all edges are labelled
  by $\bar v$.   Let $G' := G[X]$ and $\gamma'(e) := \gamma(e)$ for all
  $e\in E[G']$. As $|G'|
  \geq \Fclique(n-1,k)$, by the induction hypothesis, $G'$ contains a
  subgraph  $H'\isom K_{n-1}$ such that $\gamma(e) \cap ( V(H') \cup \{v\}) = \emptyset$ for all
  $e\in E(H')$. Furthermore, $(V(H') \cup \{v\}) \cap \gamma(\{v, u\}) = \emptyset$
  for all $u\in V(H')$. Hence, $H := G[V(H')\cup \{ v\}]$ is the required
  subgraph of $G$ isomorphic to $K_{n}$ with $\gamma(e)\cap V(H) =
  \emptyset$ for all $e\in E(H)$.
\end{proof}

We also need the next result by Erd\H{o}s and Szekeres\,\cite{ErdosS35}.

\begin{theorem}[Erd\H{o}s and Szekeres' Theorem]
\label{thm:szekeres}
Let $s,t$ be integers and let $n=(s-1)(t-1)+1$. Let $a_1,\dots, a_n$ be distinct integers.
Then there exist $1\leq i_1<\dots< i_s\leq n$ such that $a_{i_1}<\dots
  <a_{i_s}$ or
there exist $1\leq i_1<\dots< i_t\leq n$ such that $a_{i_1}>\dots >a_{i_t}$.
\end{theorem}

\subsection{Linkages, Separations, Half-Integral and Minimal Linkages. }

A \emph{linkage} $\mathcal P$ is a set of mutually vertex-disjoint
directed paths in a digraph.
For two vertex sets $Z_1$ and $Z_2$,
$\mathcal P$ is a \emph{$Z_1${-}$Z_2$ linkage} if each member of
$\PPP$ is a directed path from a vertex in $Z_1$ to a vertex in $Z_2$.
The \emph{order} of the linkage, denoted by $|\mathcal P|$, is the number of paths.
We write $\bigcup\PPP$ for the subgraph consisting of the
paths in $\mathcal P$.
Furthermore, we define $V(\PPP) := \bigcup \{ V(P) \sth P\in \PPP \}$
and $E(\PPP) := \bigcup \{ E(P) \sth P\in \PPP \}$.

\begin{definition}[well-linked sets]
  Let $G$ be a di\-graph and $A\subseteq V(G)$. $A$ is
  \emph{well-linked}, if for all $X, Y\subseteq A$ with $|X|=|Y|=
  r$ there is an $X-Y$-linkage of order $r$.
\end{definition}

 A separation $(A, B)$ in an undirected graph is a pair $A, B\subseteq
 G$ such that $G = A \cup B$. The order is $|V(A\cap B)|$.  A \emph{separation} in a directed graph $G$ is an ordered pair $(X,Y)$ of subsets of
 $V(G)$ with $X\cup Y=V(G)$ so that no edge has its tail in $X\setminus
 Y$ and its head in $Y\setminus X$.
The \emph{order} of $(X, Y)$ is $|X\cap Y|$.
 We shall frequently need the following version of Menger's theorem.

 \begin{theorem}[Menger's Theorem]\label{thm:menger}
   Let $G=(V, E)$ be a digraph with $A, B\subseteq V$ and let $k\geq 0$ be an integer.
   Then exactly one of the following holds:
   \begin{itemize}%
   \item there is a linkage from $A$ to $B$ of order $k$ or
   \item there is a separation $(X, Y)$ of $G$ of order less than $k$ with $A\subseteq X$ and $B\subseteq Y$.
   \end{itemize}
 \end{theorem}

 Let $A, B\subseteq V(G)$. A \emph{half-integral $A$-$B$ linkage of
   order $k$} in a digraph $G$ is a set $\PPP$ of $k$ $A$-$B$-paths
 such that no vertex of $G$
 is contained in more than two paths in $\PPP$. The next lemma
 collects simple facts about half-integral linkages which are needed
 below.

 \begin{lemma}\label{lem:half-integral}
   Let $G$ be a digraph and $A, B, C\subseteq V(G)$.
   \begin{enumerate}%
   \item If $G$ contains a half-integral $A$-$B$ linkage of order $k$
     then $G$ contains an $A$-$B$-linkage of order $\frac k2$.
   \item If $|B| = k$ and $G$ contains an $A$-$B$-linkage $\LLL$ of order $k$ and a
     $B$-$C$-linkage $\LLL'$ of order $k$
     then $G$ contains an $A$-$C$-linkage of order $\frac k2$.
   \end{enumerate}
 \end{lemma}
 \begin{proof}
   Part (2) follows immediately from Part (1) as $\LLL$ and $\LLL'$
   can be combined to a half-integral $A${-}$C$-linkage (this follows as
   $|B| = k$ and therefore the endpoints of $\LLL$ and $\LLL'$ in $B$
   coincide).

   For Part (1), suppose towards a contradiction that $G$ does not
   contain an $A$-$B$-linkage of order $\frac k2$. Hence, by Menger's theorem, there is a
   separation $(X, Y)$ of $G$ of order $<\frac k2$ such that
   $A\subseteq X$ and $B\subseteq Y$. But then there cannot be a
   half-integral $A$-$B$-linkage of order $k$ as every vertex in
   $X\cap Y$ can only be used twice.
 \end{proof}

We now define \emph{minimal linkages}, which play an important
role in our proof.

\begin{definition}[minimal linkages]
  Let $G$ be a digraph and let $H\subseteq G$ be
  a subgraph. Let $\LLL$ be a linkage of order $k$, for some $k\geq 1$,
  and let $C$ be the set of start vertices of $\LLL$ and $D$ be the
  set of endpoints. $\LLL$ is
  \emph{minimal with respect to $H$}, or \emph{$H$-minimal}, if for all edges $e \in
  \bigcup_{P\in \LLL} E(P)\setminus E(H)$ there is no
  $C$-$D$-linkage of order $k$ in
   the graph $(\LLL \cup H) - e$.
\end{definition}

If $\PPP$ and $\LLL$ are linkages
then we simply say that $\LLL$ is $\PPP$-minimal, instead of $\LLL$ being $\bigcup\PPP$-minimal.
The following lemma will be used later on.

\begin{lemma}\label{lem:min-linkage-closure}
  Let $G$ be a digraph.
  Let $\PPP$ be a linkage and let $\LLL$ be a linkage such that $\LLL$
  is $\PPP$-minimal. Then $\LLL$ is
  $\PPP'$-minimal for every $\PPP'\subseteq \PPP$.
\end{lemma}
\begin{proof}
  It suffices to show the lemma for the case where $\PPP' = \PPP
  \setminus \{ P\}$ for some path $P$. The general case then follows
  by induction.

  Suppose $\LLL$ is not $\PPP'$-minimal. Let $A$ and $B$ be the set of
  start and end vertices of the paths in $\LLL$, respectively. Hence, there is an edge
  $e\in E(\LLL) \setminus E(\PPP')$ such that there is an
  $A$-$B$-linkage $\LLL'$ of order $k$ in $(\PPP' \cup
  \LLL)-e$. Clearly, this edge has to be in $E(P)\cap E(\LLL)$ as it would
  otherwise violate the minimality of $\LLL$ with respect to $\PPP$.

  Furthermore, $\LLL'$ must use every edge in $E(\LLL) \setminus
  E(\PPP)$ as again it would otherwise violate the $\PPP$-minimality
  of $\LLL$. Let $Q\subseteq P$ be the maximum
  directed subpath of $P \cap \LLL$
  containing $e$ and let $s,t\in V(G)$ be its first and last vertex,
  respectively.

  If $s$ and $t$ are both end vertices of paths in $\LLL$ then this
  implies that $Q\in \LLL$ and no vertex of $Q$ is adjacent in
  $\LLL\cup \PPP'$ to any
  vertex of $\PPP'$. Hence in $(\LLL\cup \PPP') - e$ there is no path
  from $s$ to $t$, contradicting the choice of $\LLL'$.

  It follows that at least one of $s, t$ is not an endpoint of a path
  in $\LLL$. We assume that $s$ is this vertex. The case for $t$ is analogous.
  So suppose $s$ is not an end vertex of any path in $\LLL$. Let $e_s$ be the edge in
  $E(\LLL)$ with head $s$.
  As
  any two paths in $\PPP$ are pairwise vertex-disjoint, the edge $e_s$ cannot be in $E(\PPP)$.

  By construction of $Q$, no vertex in $V(Q)\setminus \{s,t\}$ is incident to any
  edge in $E(\LLL)\cup E(\PPP)$ other than the edges in $Q$.
  Furthermore, as explained above, $e_s$ must be in
  $E(\LLL')$ as it is not in  $E(\PPP)$.
  As $s$ is not an end vertex of
  a path in $\LLL$, and
  hence not an end vertex of a path in $\LLL'$, this implies that there must be an outgoing edge
  of $s$ in $\LLL'$. But this must be on the path $Q$. Hence, $\LLL'$
  must include $Q$, and
  thus the edge $e$, a contradiction.
\end{proof}

Note that deleting paths from $\LLL$ can destroy minimality. I.e.~if $\LLL$
  is $\PPP$-minimal and $\LLL'\subset \LLL$ then $\LLL'$ may no longer
  be $\PPP$-minimal. In the rest of the paper we will mainly use the
  following property of minimal linkages.

\begin{lemma}\label{lem:no-forward-paths}
  Let $G$ be a digraph and $H\subseteq G$ be a subgraph. Let $\LLL$ be
  an $H$-minimal linkage between two sets $A$ and $B$. Let
  $P \in \LLL$ be a path and let $e \in E(P)\setminus E(H)$. Let $P_1,
  P_2$ be the two components of $P - e$ such that the tail of $e$ lies
  in $P_1$. Then there are at most $r := |\LLL|$ pairwise vertex-disjoint paths from $P_1$ to
  $P_2$ in $H \cup \LLL$.
\end{lemma}
\begin{proof}
  As $\LLL$ is $H$-minimal, there are no $r$ pairwise vertex
  disjoint $A{-}B$ paths in $(H \cup \LLL)- e$. Let $S$ be a minimal
  $A{-}B$ separator in $(H \cup \LLL)-e$. Hence, $|S|=r-1$ and $S$
  contains exactly one vertex from every $P' \in \LLL\setminus
  \{P\}$.

  Towards a contradiction, suppose there were $r$ pairwise vertex-disjoint paths
  from $P_1$ to $P_2$ in $(H\cup\LLL)-e$. At most $r-1$ of these contain a vertex from
  $S$ and hence there is a $P_1{-}P_2$ path $P'$ in $(H\cup\LLL) -e-S$. But then
  $P_1 \cup P' \cup P_2$ contains an $A{-}B$ path in $(H\cup\LLL) -e-S$,
  contradicting the fact that $S$ is an $A{-}B$ separator in $(H\cup\LLL) -e$.
  Hence, there are at most $r-1$ disjoint paths from
  $P_1$ to $P_2$ in $(H\cup\LLL)-e$ and therefore at most $r$
  pairwise vertex-disjoint $P_1{-}P_2$ paths in $(H\cup\LLL)$.
\end{proof}

\section{Directed Tree-Width}
\label{sec:dtw}

The main result of this paper is the grid theorem for directed tree-width.
Directed tree-width was introduced by Johnson, Robertson, Seymour and
Thomas in 2001 \cite{JohnsonRobSeyTho01}, in the same paper the
directed grid conjecture was made. See also \cite{KreutzerO14} for a thorough introduction to
directed tree-width and its obstructions.

Unfortunately, Adler \cite{Adler07} proved that there are digraphs $G$
and butterfly minors $H$ of $G$ such that the directed tree-width of
$H$ is larger than the directed tree-width of $G$. That is, taking
butterfly minors may increase the directed tree-width.

In \cite{JohnsonRobSeyTho01b}, Johnson et al. proposed a slightly more
general variant of directed tree-decompositions which allows for bags
to be empty and proved that any directed tree-decomposition in
the original sense of \cite{JohnsonRobSeyTho01} can be converted into
a directed tree-decomposition without increasing the width which allows empty bags and with the
property that for every edge of the tree-decomposition the union of
all bags below the edges induced a strong component of the digraph
without the guard of the edge. Furthermore, the number of nodes in the
tree-decomposition only grows at most quadratic. 
A detailed comparison of various proposed definitions of directed tree-width can be found in \cite{Kim22}.

In this paper we essentially use the modified version of directed
tree-width proposed in \cite{JohnsonRobSeyTho01b} with only a
insignificant variation in our guarding condition which has no impact
on the width of the decomposition but makes the guarding condition
symmetric. 

Unlike the original definition of directed tree-width, we will show
below that the variant of directed tree-width with empty bags is
closed under taking butterfly minors.

By an \emph{arborescence} we mean a rooted tree in which every edge is
oriented away from the root $r_0$, i.e.~an acyclic directed graph $T$ such that $T$
has a vertex $r_0$, called the
\emph{root} of $T$, with the property that for every vertex $r \in V (T)$ there is a
unique directed path from $r_0$ to $r$.

For $r \in V(T)$ we denote the sub-arborescence of $T$ induced by the
set of vertices in $T$ reachable from $R$ by
$T_t$. In particular, $r$ is the root of $T_r$.

\begin{definition}\label{def:dtw}
  A \emph{directed tree decomposition} of a digraph $G$ is a triple
  $(T,\beta,\gamma)$, where $T$ is an arborescence,  $\beta \sth V(T)
  \rightarrow 2^{V(G)}$ and $\gamma \sth E(T) \rightarrow 2^{V(G)}$ are
  functions such that
  \begin{enumerate}
  \item $\{ \beta(t) \sth t \in V(T) \}$ is a partition of $V(G)$ into (possibly
    empty) sets and
  \item if $e = (s, t) \in E(T)$ and $A = \bigcup \{ \beta(t') \sth  t' \in V(T_t)\}$ and
    $B = V(G) \setminus A$ then there is no closed directed walk in $G - \gamma(e)$
    containing a vertex in $A$ and a vertex in $B$.
  \end{enumerate}
  For $t \in V(T)$ we define $\Gamma(t) := \beta(t) \cup \bigcup \{
  \gamma(e) \sth e \sim t\}$, where $e\sim t$ if $e$ is incident with
  $t$, and we define $\beta(T_t) :=  \bigcup \{ \beta(t') \sth  t' \in V(T_t)\}$. 

  The \emph{width} of $(T, \beta, \gamma)$ is the least integer $w$
  such that $|\Gamma(t)| \leq w + 1$  for all $t
  \in V(T)$.  The \emph{directed tree-width} of $G$ is the least
  integer $w$ such that $G$ has a directed tree decomposition of width $w$.
\end{definition}

The sets $\beta(t)$ are called the \emph{bags} and the sets
$\gamma(e)$ are called the \emph{guards} of the directed tree decomposition.
As proved in \cite{Kim22}, the variant of directed tree-width we use here is closed under taking butterfly minors. For convenience, we give a proof of this fact below. 

\begin{lemma}\label{lem:dtw-closed-butterfly-minors}
  Let $G$ be a digraph and let $H$ be a butterfly minor of $G$.
  Then $\dtw(H) \leq \dtw(G)$.
\end{lemma}
\begin{proof}
  It is easily seen that if $H$ is a subdigraph of $G$ then $\dtw(H) \leq
  \dtw(G)$. Thus it suffices to show that if $e = (u, v) \in E(G)$ is
  butterfly contractible and $H$ is obtained from $G$ by contracting
  $e$, then $\dtw(H) \leq \dtw(G)$.  

  Let $\TTT := (T, \beta, \gamma)$ be a directed tree-decomposition of $G$ of minimal
  width $k$. Let $t_u, t_v \in V(T)$ be such that $u \in \beta(t_u)$ and $v \in
  \beta(t_v)$.

  Suppose first that $\delta^+(u) = 1$. To simplify notation, we may assume that $e$ is
  contracted onto $v$, that is $V(H) = V(G) \setminus \{ u \}$ and $E(H) := E(G)
  \setminus \{ e'  \sth e' = (s, u) \in E(G)$ or $e' = e \} \cup \{ (s, v) \sth (s, u) \in
  E(G)$ and $(s, v) \not\in E(G) \}$.

  Let $\TTT' := (T, \beta', \gamma')$ be obtained from $\TTT$ as follows. We set
  $\beta'(t_u) := \beta(t_u) \setminus \{ u \}$ and for all $t \in V(T) \setminus \{ t_u \}$ we set
  $\beta'(t) := \beta(t)$. If $e \in E(T)$ is an edge with $u \in \gamma(e)$ then we
  define $\gamma'(e) := \gamma(e) \setminus \{ u \} \cup \{ v \}$, otherwise we define $\gamma'(e) :=
  \gamma(e)$. 

  We claim that $\TTT'$ is a directed tree-decomposition of $H$ of width
  $\leq k$. It is easy to see that the width of $\TTT'$ is still bounded by
  $k$. 
  
  What remains to be shown is the guarding condition. Towards this aim,
  let $e = (s, t) \in E(T)$ be an edge and let 
  $W' \subseteq H$ be a closed directed walk containing a vertex in $\bigcup \{ \beta'(t') \sth t' \in
  V(T_{t'}) \}$  and a vertex in $V(H) \setminus \bigcup \{ \beta'(t') \sth t' \in  V(T_{t'})
  \}$. We need to show that $W'$ contains a vertex of  $\gamma'(e)$.

  If  $W' \subseteq H$ then $V(W') \cap \gamma(e) \neq \emptyset$, as $\TTT$ is a directed
  tree-decomposition. Let $x \in V(W') \cap \gamma(e)$. By construction, $x \in \gamma'(e)$ and thus $V(W') \cap
  \gamma'(e) \neq \emptyset$. Otherwise, $W'$ is obtained from  a walk $W \subseteq G$
  by replacing a subwalk $s \rightarrow u \rightarrow v$ by the new edge $s  \rightarrow v$, as $v$
  is the only out-neighbour of $u$ in $G$. Let $a \in V(W') \cap \bigcup \{ \beta'(t')
  \sth t' \in V(T'_{t'}) \}$  and  $b \in V(H) \setminus \bigcup \{ \beta'(t') \sth t' \in  V(T'_{t'})  \}$. It
  follows immediately from the construction that $a \in V(W) \cap \bigcup \{ \beta(t') \sth t' \in
  V(T_{t'}) \}$  and  $b \in V(G) \setminus \bigcup \{ \beta(t') \sth t' \in  V(T_{t'})  \}$. 
  By assumption on $\TTT$, $V(W) \cap \gamma(e) \neq \emptyset$. If there is
  $x \in V(W) \cap \gamma(e) \setminus \{ u \}$  then $x \in \gamma'(e)$ and thus $V(W') \cap \gamma'(e) \neq
  \emptyset$. Thus we may assume that $V(W) \cap \gamma(e) = \{ u \}$. But then $v \in
  \gamma'(e)$ and again $V(W') \cap \gamma'(e) \neq \emptyset$.

  This completes the case where $\delta^+(u) =1 $. If $\delta^+(u) > 1$ then
  $\delta^-(v) = 1$ and we proceed in the same way as above but with the r\^ole
  of $u$ and $v$ exchanged. We set
  $\beta'(t_v) := \beta(t_v) \setminus \{ v \}$ and for all $t \in V(T) \setminus \{ t_v \}$ we set
  $\beta'(t) := \beta(t)$. If $e \in E(T)$ is an edge with $v \in \gamma(e)$ then we
  define $\gamma'(e) := \gamma(e) \setminus \{ v \} \cup \{ u \}$, otherwise we define $\gamma'(e) :=
  \gamma(e)$. The proof that $(T, \beta', \gamma')$ is a directed tree-decomposition
  of $H$ is analogous to the case above.
\end{proof}

We now recall the concept of cylindrical
grids as defined in \cite{Reed97,JohnsonRobSeyTho01}.

\begin{definition}[cylindrical grid]\label{def:cyl-grid}
  A \emph{cylindrical grid} of order $k$, for some $k\geq 1$, is a
  digraph $G_k$ consisting of $k$ directed cycles $C_1, \dots, C_k$
  of length $2k$,
  pairwise vertex disjoint, together with a set of $2k$ pairwise
  vertex disjoint paths $P_1, \dots, P_{2k}$ of length $k-1$ such that \parsep-10pt
  \begin{itemize}%
  \item each path $P_i$ has exactly one vertex in common with each
    cycle $C_j$ and both endpoints of $P_i$ are in $V(C_1)\cup
    V(C_k)$%
  \item the paths $P_1, \dots, P_{2k}$ appear on each $C_i$ in this
    order and
  \item for odd $i$ the cycles $C_1, \dots, C_k$ occur on all $P_i$
    in this order and for even $i$ they occur in reverse order $C_k,
    \dots, C_1$.
  \end{itemize}
\end{definition}

See Figure~\ref{fig:grid} for an illustration of $G_4$.

\begin{definition}\label{rem:grid-in-wall}
  Let us define an \emph{elementary cylindrical wall} $W_k$ of order $k$ to be the
  digraph obtained from the cylindrical grid $G_k$ by replacing every
  vertex $v$
  of degree $4$ in $G_k$ by two new vertices $v_s, v_t$ connected by
  an edge $(v_s, v_t)$ such that $v_s$ has the same in-neighbours as
  $v$ and $v_t$ has the same out-neighbours as $v$.

  A \emph{cylindrical wall} of order $k$ is a subdivision of $W_k$.
  Clearly, every cylindrical wall of order $k$ contains a cylindrical grid of
  order $k$ as a butterfly minor. Conversely, a cylindrical grid of
  order $k$ contains a cylindrical wall of order $\frac k2$ as
  subgraph.
\end{definition}

Again, see Figure~\ref{fig:wall} for an illustration. 
What we actually show in this paper is that every digraph of large
directed tree-width contains a cylindrical wall of high order as
subgraph.

Directed tree-width has a natural duality, or obstruction, in terms of
directed brambles (see \cite{Reed97,Reed99}).

\begin{definition}
  Let $G$ be a digraph. A \emph{bramble} in $G$ is a set $\BBB$ of
  strongly connected subgraphs $B \subseteq G$ such that if $B,
  B' \in \BBB$ then $B \cap B' \neq \emptyset$.%

  A \emph{cover of $\BBB$} is a set $X \subseteq V(G)$ of vertices such
  that $V(B) \cap X \neq \emptyset$ for all $B \in \BBB$. Finally,
  the \emph{order} of a bramble is the minimum size of a cover of
  $\BBB$. The \emph{bramble number} $\bn(G)$ of $G$ is the maximum
  order of a bramble in $G$.
\end{definition}

Our definition here differs slightly from the definition in
\cite{Reed99} where it is allowed for two bramble elements $B$ and
$B'$ to be disjoint as long as there are edges linking $B$ and $B'$
both ways. Again this definition has the problem of not being closed
under taking butterfly minors whereas our definition above is easily
seen to be closed under taking butterfly minors. It is not hard to see
that if $G$ has a bramble of order $2k+1$ with respect to the
definition in \cite{Reed99} then it has a bramble of order $k$ in our
definition (and clearly any bramble with respect to our definition
is a bramble as defined in \cite{Reed99}).

The next lemma is mentioned in \cite{Reed99} and can be
proved by converting brambles into havens and back using
\cite[(3.2)]{JohnsonRobSeyTho01}. See also \cite{KreutzerO14} for a proof.

\begin{lemma}\label{lem:bramble=dtw}
  There are constants $c, c'$ such that
  for all digraphs $G$, $\bn(G) \leq c\cdot \dtw(G) \leq c'\cdot \bn(G)$.
\end{lemma}

Using this lemma we can state our main theorem equivalently as
follows, which is the result we prove in this paper.

\begin{theorem}\label{thm:main-bramble}
  There is a computable function
  $f\sth \N \rightarrow \N$ such that for all digraphs $G$ and all
  $k\in\N$, if $G$ contains a bramble of order at least $f(k)$ then
  $G$ contains a cylindrical grid of order $k$ as a butterfly minor.
\end{theorem}

\section{Getting a web}
\label{sec:bramble}

The main objective of this section is to show that every digraph
containing a bramble of high order either contains a cylindrical wall
of order $k$ or contains a structure that we call a \emph{web}.

\begin{definition}[{\bf $(p,q)$-web}]
  Let $p,q, d > 0$ be integers. A \emph{$(p,q)$-web $(\PPP, \QQQ)$
    with avoidance $d$} in a
  digraph $G$ consists of two linkages $\PPP = \{P_1,\dots,P_p\}$
  and $\QQQ = \{Q_1,\dots,Q_q\}$ such that
  \begin{enumerate}%
  \item $\mathcal{P}$ is an $A{-}B$ linkage for two distinct vertex
    sets $A,B \subseteq V(G)$ and $\mathcal{Q}$ is a $C{-}D$ linkage for
    two distinct vertex sets $C,D \subseteq V(G)$,
  \item for $1 \leq i \leq q$, $Q_i$ intersects all but at most $\frac
    1d\cdot p$
    paths in $\mathcal{P}$ and
  \item $\PPP$ is $\QQQ$-minimal.
  \end{enumerate}
  We say that $(\PPP, \QQQ)$ has avoidance $d=0$ if $Q_i$ intersects every
  path in $\mathcal{P}$, for all $1 \leq i \leq q$.

  The set $C\cap V(\QQQ)$ is called the \emph{top} of the web, denoted
  $\Top\big( (\PPP,
  \QQQ) \big)$, and $D\cap V(\QQQ)$ is the \emph{bottom} $\Bot\big( (\PPP,
  \QQQ) \big)$.  The web $(\PPP, \QQQ)$ is \emph{well-linked} if $C\cup D$ is
  well-linked in $G$.
\end{definition}

The notion of top and bottom refers to the intuition, used in the rest
of the paper, that the paths in
$\QQQ$ are thought of as \emph{vertical} paths and the paths in $\PPP$
as \emph{horizontal}. In this section we will prove the following theorem.

\begin{theorem}\label{thm:bramble-to-web}
  For every $k, p, l, c\geq 1$ there is an integer $l'$ such that the
  following holds. Let $G$ be a digraph of bramble number at least
  $l'$. Then $G$ contains a cylindrical grid of order $k$ as a
  butterfly minor or a
  $(p',l\cdot p')$-web with avoidance $c$, for some $p'\geq p$, such
  that the top and the bottom of the web are subsets of a well-linked
  set $A\subseteq V(G)$.
\end{theorem}

The starting point for proving the theorem are brambles of high order
in directed graphs. In the
first step we adapt an approach developed in \cite{KawarabayashiK14},
based on \cite{ReedW12}, to
our setting.

\begin{lemma}\label{lem:long-path}
  Let $G$ be a digraph and $\BBB$ be a bramble in $G$. Then there is a
  path $P := P(\BBB)$ intersecting every $B\in \BBB$.
\end{lemma}
\begin{proof}
  We inductively construct the path $P$ as follows. Choose a  vertex
  $v_1 \in V(G)$ such that $v_1 \in V(B_1)$ for some $B_1\in \BBB$ and set $P
  := (v_1)$. During the construction we will maintain the property
  that there is a bramble element $B\in\BBB$ such that the last vertex $v$ (i.e., the endvertex)
  of $P$ is the only element of $P$ contained in $B$.
  Clearly this property is true for the path $P = (v_1)$ constructed
  so far.

  As long
  as there still is an element $B\in \BBB$ such that $V(P) \cap V(B) =
  \emptyset$, let $v$ be the last vertex of $P$ and $B' \in \BBB$ be such
  that $P\cap B' = \{ v\}$. By definition of a directed bramble, there
  is a path in $G[V(B\cup B')]$ from $v$ to a vertex in $B$. Choose $P'$ to be such a path
  so that only its endpoint is contained in $B$ and all other vertices
  of $P'$ are contained in $B'$. Hence, $P'$ only shares $v$ with $P$
  and we can therefore combine $P$ and $P'$ to a path ending in $B$ to obtain the desired path.
\end{proof}

\begin{lemma}\label{lem:well-linked-path}
  Let $G$ be a digraph, $\BBB$ be a bramble of order $2k\cdot (k+1)$ and $P =
  P(\BBB)$ be a path intersecting every $B\in\BBB$. 
  Then there is a set $A\subseteq V(P)$ of order $k$ which is well-linked.
\end{lemma}
\begin{proof}
  We first construct a sequence of subpaths $P_1, \dots$, $P_{2k}$ of $P$
  and brambles $\BBB_1, \dots, \BBB_{2k} \subseteq \BBB$ of order
  $k+1$ as follows. Let $P_1$ be the minimal initial subpath of $P$ such
  that $\BBB_1 := \{ B \in \BBB \sth B \cap P_1 \not= \emptyset \}$ is
  a bramble of order $k+1$.
  Now suppose $P_1, \dots, P_i$ and $\BBB_1, \dots, \BBB_i$ have already
  been constructed. Let $v$ be the last vertex of $P_i$ and let
  $s$ be the successor of $v$ on $P$. Let $P_{i+1}$ be the minimal subpath of $P$
  starting at $s$ such that
  \[
     \BBB_{i+1} := \{ B \in \BBB \sth B\cap \bigcup_{l \leq i} V(P_l) =
     \emptyset \text{ and } B \cap P_{i+1} \not= \emptyset\}
  \]
  has order $k+1$. As long as $i < 2k$ this is always possible as
  $\BBB$ has order $2k\cdot(k+1)$.

  Now let $P_1, \dots, P_{2k}$ and $\BBB_1, \dots,
  \BBB_{2k}$ be constructed in this way. For $1\leq i \leq k$ let $a_i$ be the first
  vertex of $P_{2i}$ and let $Q_i$ be the minimal subpath of $P$
  containing $P_{2i-1}$ and $P_{2i}$. We define $A := \{ a_1, \dots, a_k\}$.

  We show next that $A$ is well-linked.
  Let $X, Y \subseteq A$ be such that $|X| = |Y| = r$.

  Suppose there is no
  linkage of order $r$ from $X$ to $Y$. By Menger's theorem
  (Theorem~\ref{thm:menger}) there is a set $S\subseteq V(G)$ of order
  $|S| < r$ such
  that there is no path from $X$ to $Y$ in $G\setminus S$. Clearly,
  $X\cap Y \subseteq S$.

  As $|S| < r$, by the pigeon hole principle, there are indices $i$ and $j$ such that
  $a_i\in X\setminus S$,
  $a_j\in Y\setminus S$ and $S\cap V(Q_i) = S\cap V(Q_j) =
  \emptyset$.
  W.l.o.g.~assume $i<j$.

  By construction, $\BBB_{2i}$ and $\BBB_{2j-1}$ are both brambles
  of order $k+1$ and hence $S$ is not a hitting set of $\BBB_{2i}$
  or $\BBB_{2j-1}$. Hence, there are bramble elements $B \in
  \BBB_{2i}$ and $B'\in\BBB_{2j-1}$ such that $S\cap B
  = S\cap B'=\emptyset$. But as any two
  bramble elements of $\BBB$
  intersect this implies that there is a path in $B \cup B'$ from
  a vertex of $v \in V(P_{2i})$ to a vertex $w \in
  V(P_{2j-1})$ and hence, as $a_i$ is the start vertex of
  $P_{2i}$ and $a_j$ is the start vertex of $P_{2j}$
  a path
  linking  $a_i$ to $a_j$ avoiding $S$, which is a contradiction.
\end{proof}

 We will now define the first of various sub-structures we are
guaranteed to find in digraphs of large directed tree-width.

\begin{figure}[t]
  \centering
    \includegraphics{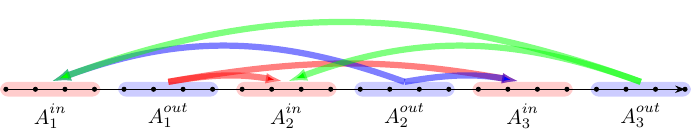}
  \caption{A $4$-linked path system of order $3$.}
  \label{fig:path-system}
\end{figure}

\begin{definition}[path system]\label{def:path-system}
  Let $G$ be a digraph and let $l, p \geq 1$. An \emph{$l$-linked path
    system of order $p$} is a sequence $\SSS := (\PPP, \LLL, \AAA)$,
  where
  \begin{itemize}
  \item  $\AAA := \big(A_i^{in}, A_i^{out}\big)_{1\leq i
    \leq p}$ such that $A := \bigcup_{1\leq i \leq p} A_i^{in}
  \cup A_i^{out} \subseteq V(G)$ is a
    well-linked set of order $2\cdot l \cdot p$ and $|A_i^{in}| = |A_i^{out}| = l$, for all $1\leq
    i\leq p$,
  \item $\PPP := (P_1, \dots, P_p)$ is a sequence of pairwise vertex
    disjoint paths and for all
    $1\leq i \leq p$, $A_i^{in}, A_i^{out} \subseteq V(P_i)$ and all $v\in A^{in}_i$ occur on
    $P_i$ before any $v'\in A_i^{out}$ and the first vertex of $P_i$
    is in $A_i^{in}$ and the last vertex of $P_i$ is in $A_i^{out}$ and
   \item $\LLL := (L_{i,j})_{1\leq i\not= j \leq p}$ is a sequence of
     linkages such that
    for all $1\leq i\not=j\leq p$, $L_{i,j}$ is a linkage of
       order $l$ from $A_i^{out}$ to $A_j^{in}$.
  \end{itemize}

  The system $\SSS$ is \emph{clean} if for all $1\leq i \not=  j\leq
  p$ and all $Q\in L_{i,j}$, $Q\cap P_s=\emptyset$
  for all $1\leq s \leq p$ with $s\not\in \{i,j\}$.
\end{definition}

See Figure~\ref{fig:path-system} for an illustration of path
systems. The next lemma
follows easily from Lemma~\ref{lem:well-linked-path}.

\begin{lemma}\label{lem:path-system}
  Let $G$ be a digraph and $l, p\geq 1$. There is a computable function
  $\Fpath \sth \N\times \N\rightarrow \N$ such that if $G$ contains a
  bramble of
  order $\Fpath(l, p)$ then $G$ contains an $l$-linked path system
  $\SSS$ of order $p$.
\end{lemma}
\begin{proof}  %
  We set $\Fpath(l, p) := (4\cdot l \cdot p) \cdot (2\cdot l \cdot p+1)$.
  Let  $\BBB$ be a bramble  of
  order at least $\Fpath(l, p)$.
  Let $P$ be a path hitting every $B\in \BBB$ and let
  $A\subseteq V(P)$ be a well-linked subset of $V(P)$ with $|A| = 2\cdot l
  \cdot p$ as in Lemma~\ref{lem:well-linked-path}.
  Let $a_1, \dots, a_{2lp}$ be the elements of $A$ in the order in which
  they occur on $P$ when traversing $P$ from beginning to end.

  For $1 \leq i \leq p$ let $A_i := \{ a_{2l(i-1)+1}, \dots, a_{2li} \}$,
  $A_i^{in} := \{a_{2l(i-1)+1}, \dots, a_{2l(i-1)+l}\}$, $A_i^{out} :=
  \{a_{2l(i-1)+l+1}, \dots, a_{2li}\}$, and let $P_i$ be the minimal subpath of $P$ containing
  the elements  of $A_i$. The well-linkedness of $A$implies that for
  all $1 \leq i \neq j \leq p$ there is a linkage $L_{i,j}$ of order $l$ from
  $A_i^{out}$ to $A_j^{in}$. Thus, if we set $\PPP := (P_1, \dots, P_p), \AAA :=
  \big(A_i^{in}, A_i^{out}\big)_{1 \leq i \leq p}$ and $\LLL := \big(L_{i,j}\big)_{1 \leq
    i \neq j \leq p}$ then the triple $(\PPP, \LLL, \AAA)$ satisfies the requirements
  of Definition~\ref{def:path-system} and is an $l$-linked path system
  of order $p$ as required.

\end{proof}

We now show how to construct a clean path system from a general
path system.

\begin{lemma}\label{lem:clean-path-system}
  Let $G$ be a digraph. There is a computable function
  $\Fcleanpath \sth \N^4 \rightarrow \N$ such that for all integers $l,
  p, k, c\geq 1$, if $G$ contains a bramble of
  order $\Fcleanpath(l, p, k, c)$ then $G$ contains a clean $l$-linked
  path system $\SSS$ of order $p$ or a well-linked $(p', k\cdot p')$-web
  of avoidance   $c$, for some $p' \geq p$.
\end{lemma}
\begin{proof}
   Let $l, p, k, c\geq 1$. For all $0\leq i \leq p$ let $l_i := c^i\cdot l$. Let $n :=
  p$. Furthermore, let $p_0 := p$ and for $0 < i \leq p$ let $p_i :=
  \Fclique((p_{i-1}+1) \cdot (1+4k\cdot l_i)^{n-i}, 2k\cdot l_i)$,
  where $\Fclique$ is the function defined in Lemma~\ref{lem:clique}. Finally, define $\Fcleanpath(l, p,
  k, c) := \Fpath(l_n, p_n)$, where $\Fpath$ is the function defined
  in Lemma~\ref{lem:path-system}.

  Let $G$ be a digraph containing a bramble of
  order $\Fcleanpath(l, p, k, c)$.
  By Lemma~\ref{lem:path-system}, $G$ contains an $l_n$-linked path
  system $\SSS := (\PPP, \LLL, \AAA)$ of order $p_n$.
  By backwards induction on $i$, we define for all $0\leq i\leq n$
  \begin{itemize}
  \item  disjoint sets $\YYY_i, \PPP_i \subseteq \PPP$ with $|\YYY_i| = n-i$ and
    $|\PPP_i|  = p_i$ and
  \item for
    all $P_s, P_t \in \YYY_i \cup \PPP_i$ with $s\not= t$ a
    $\PPP_i$-minimal $A^{out}_s$-$A^{in}_t$-linkage $L^i_{s,t}$ of order $l_i$
    such that no path in $L^i_{s,t}$ hits any path in $\YYY_i
    \setminus \{ P_s, P_t\}$.
  \end{itemize}
  Clearly,
  $\YYY_0$ contains $n=p$ paths and induces a clean $l=l_0$-linked
  path system of order $p$.

  Initially, we set $\PPP_n := \PPP$ and $\YYY_n :=
  \emptyset$. Furthermore, for all $1\leq i \not=j \leq
  p_n$, we choose a $\PPP$-minimal
  $A^{out}_i$-$A_j^{in}$-linkage $L^n_{i,j}$.
  Clearly this satisfies the conditions above.

  Now suppose $\YYY_i, \PPP_i$ and the $L^i_{s,t}$ satisfying the
  conditions above have already been defined.

  \medskip\noindent\textit{Step 1. } We label each pair $s\not= t$
  with $P_s, P_t \in \PPP_i$ by
  \[
    \gamma(s,t) := \big\{ P \in \PPP_i \sth P\not= P_s,P_t \text{ and }
    \begin{array}{l}
      P  \text{ hits at least } (1-\frac 1c)l_i \text{ paths in }
      L^i_{s,t} \text{  or}\\
      P \text{ hits at least } (1-\frac 1c)l_i \text{ paths in } L^i_{t, s}
     \end{array}
       \big \}.
  \]
  If there is a pair $s,t$ such that $|\gamma(s,t)| \geq 2\cdot k\cdot
  l_i$ then for at least one of $L^i_{s,t}$ or $L^i_{t, s}$, say
  $L^i_{s,t}$, there is a set $\gamma \subseteq \gamma(s,t)$ of order
  $k\cdot l_i$ such that all $P\in \gamma$ hit at least $(1-\frac
  1c)l_i$ paths in $L^i_{s,t}$.
 As by Lemma~\ref{lem:min-linkage-closure} the
  $\PPP_i$-minimal linkage $L^i_{s,t}$ is also $\gamma$-minimal,
  $(L^i_{s,t}, \gamma)$ form a $(l_i, k\cdot l_i)$-web with avoidance
  $c$. This yields the second outcome
  of the lemma. Note that the web is well-linked as the vertical paths $\gamma$
  are formed by paths from $\PPP$ and, by definition of path systems,
  these start and
  end in elements of the well-linked set $A$.

  So suppose that all sets $\gamma(s,t)$ contain at most $2k\cdot l_i$
  paths. Note that, by construction, $\gamma$ is symmetric, i.e.~$\gamma(s,t) = \gamma(t,s)$ for all $s, t$, and therefore $\gamma$ defines a labelling of the clique $K_{|\PPP_i|}$ so that we can apply Lemma~\ref{lem:clique}. As $|\PPP_i| \geq \Fclique((p_{i-1}+1)(1+4k\cdot l_i)^{n-i}, 2k\cdot l_i)$, by
  Lemma~\ref{lem:clique}, there is a set $\PPP'_i \subseteq \PPP_i$ of
  order $(p_{i-1}+1)(1+4k\cdot l_i)^{n-i}$
  such that for any pair $s \not=t$ with $P_s, P_t \in \PPP'_i$ there is no
  path $P \in \PPP'_i\setminus \{ P_s, P_t\}$ hitting $(1-\frac 1c)l_i$
  paths in $L^i_{s,t}$ or $L^i_{t, s}$.

  \smallskip

  So far we have found a set $\PPP'_i$ such if $P_s \neq P_t\in \PPP'_i$ then
  there 
  is no
  path $P\not= P_s, P_t$ hitting at least $(1-\frac 1c)l_i$ paths in
  $L^i_{s,t}$. However, such a path $P$ could still exist for a linkage
  $L_{s,t}$ or $L_{t,s}$ between a $P_s\in \YYY_i$ and a $P_t\in \PPP'_i$. We address this
  problem in
  the second step.

  \medskip

  \noindent\textit{Step 2. }
  Let $(Y_1, \dots, Y_{n-i}) = \YYY_i$ be an enumeration of all paths
  in $\YYY_i$.   For all $1\leq s \leq n-i$ and $t$ such that $P_t\in
  \PPP'_i$ let $L^i_{s,t}$ and $L^i_{t,s}$ be the linkages between
  $Y_s$ and $P_t$ as constructed above. Inductively, we will construct sets
  $\QQQ^i_j\subseteq \PPP'_i$, for
  $0\leq j\leq n-i$, such that $|\QQQ^i_j|=(p_{i-1}+1)(1+4kl_i)^{n-i-j}$
  and
  for any $1\leq s\leq j$ (with $j \geq 1$) and $P_t\in \QQQ^i_j$, there is no
  path $P\in \QQQ^i_j\setminus \{P_t\}$ such that $P$ hits at least
  $(1-\frac 1c)l_i$ paths in $L^i_{s,t}$ or $L^i_{t,s}$.

  Let $\QQQ^i_0 := \PPP_i'$ which clearly satisfies the conditions.
  So suppose that $\QQQ^i_j$, for some $j<n-i$, has already been
  defined. Set $s := j+1$.
  For every $P_t\in \QQQ^i_{j}$ define
  \[
     \gamma(t) := \big\{ P  \in \QQQ^i_{j} \setminus \{ P_t\} \sth \begin{array}{l}
       P\text{ hits at least $(1-\frac 1c)l_i$ paths
         in $L^i_{s,t}$ or}\\
       P\text{ hits at least $(1-\frac 1c)l_i$ paths
         in $L^i_{t,s}$}
     \end{array}
     \big\}.
  \]
  Again, if there is a $P_t\in \QQQ^i_{j}$ such that $|\gamma(t)| \geq
  2\cdot k\cdot l_i$ then choose $\gamma\subseteq \gamma(t)$ of size
  $|\gamma|= k\cdot l_i$ such that for one of $L^i_{s,t}$ or
  $L^i_{t,s}$, say $L^i_{s,t}$, every $P\in \gamma$ hits at least
  $(1-\frac 1c)l_i$ paths
  in $L^i_{s,t}$. Then $(\gamma, L^i_{s,t})$ is a well-linked web as
  requested.

  Otherwise, as $|\QQQ^i_{j}| = (p_{i-1}+1)(1+4k\cdot l_i)^{n-i-j}$, by
  Lemma~\ref{lem:clique-vertex},
 there is a
  subset $\QQQ^i_{j+1}$ of order $(p_{i-1}+1)(1+4k\cdot l_i)^{n-i-(j+1)}$
  such that for no $P_t\in
  \QQQ^i_{j+1}$ there is a path $P\in
  \QQQ^i_{j+1}\cup \YYY_i$ hitting at least $(1-\frac 1c)l_i$ paths in
  $L^i_{s,t}$ or $L^i_{t, s}$.

  \smallskip

  Now suppose that $\QQQ^i_{n-i}$ has been defined. We choose a path $P_n
  \in \QQQ^i_{n-i}$ and set $\YYY_{i-1} := \YYY_i \cup \{ P_n\}$, and
  define $\PPP_{i-1} := \QQQ^i_{n-i}\setminus \{ P_n \}$.
  By construction, $|\YYY_{i-1}| = n-(i-1)$ and $|\PPP_{i-1}| = p_{i-1}$.
  Furthermore, for every pair $P_s, P_t\in
  \YYY_{i-1}\cup\PPP_{i-1}$ there is a linkage $L$ from $A^{out}_s$ to
  $A^{in}_t$ of order $\frac 1c\cdot l_i = l_{i-1}$ such that $P_n$
  does not hit any path in $L$ and, by induction hypothesis, neither
  does any $P'\in \YYY_i$. Hence, for any such pair $P_s, P_t$ we
  can choose a $\PPP_{i-1}$-minimal $A^{out}_s{-}A^{in}_t$-linkage
  avoiding every path in $\YYY_{i-1}$.

 Hence, $\YYY_{i-1}$ and $\PPP_{i-1}$ satisfy the conditions above and
 we can continue the induction.

  \smallskip

  If we do not get a web as the second outcome of the lemma then after
  $p$ iterations we have constructed $\YYY_0$ and $\PPP_0$ and the
  linkages $L^0_{s,t}$ for every $P_s, P_t\in \YYY_0$ with $s\not=t$. Clearly, $\YYY_0$ induces a clean $l =
  l_0$-linked path system of order $p_0 = p$ which is the first
  outcome of the lemma.
\end{proof}

The following lemma completes the proof of Theorem~\ref{thm:bramble-to-web}.

\begin{lemma}\label{lem:tournament-clique}
  For every $k, p, l, c\geq 1$ there is an integer $l'$ such that the
  following holds. Let $\SSS$ be a clean $l'$-linked path system of
  order $3k$. Then either $G$ contains a bidirected clique of order $k$
  as a butterfly minor or a well-linked
  $(p',l\cdot p')$-web with avoidance $c$, for some $p'\geq p$.
\end{lemma}
\begin{proof}
  Let $K := 2 \cdot {3k \choose 2}$.
  We define a function $f \sth [K] \rightarrow \N$ with $f(t) :=
  (c \cdot K \cdot l)^{({K}-t+1)} p$ and set $l' :=
  f(1)$. For all $1 \leq t \leq K$ we define $g(t) := \frac{f(t)}{K \cdot l}$.

  Let $\SSS := (\PPP, \LLL, \AAA)$ be a clean $l'$-linked path system
  of order $3k$, where $\PPP := (P_1, \dots, P_{3k})$, $\LLL :=
  (L_{i,j})_{1\leq i\not=j \leq 3k}$ and $\AAA :=
  (A^{in}_i, A^{out}_i)_{1\leq i \leq 3k}$.
  We fix an ordering of the pairs $\{ (i,j) \sth 1\leq i \not= j \leq 3k\}$.
  Let $\sigma \sth [K] \rightarrow \{ (i,j) \sth 1\leq i \not= j \leq
  3k\}$ be the bijection between $[K]$ and $\{ (i,j) \sth 1\leq i \not= j \leq
  3k\}$ induced by
  this ordering.
  We will inductively construct linkages $L_{i,j}^r$, where $r \leq K$,
  such that
  \begin{enumerate}
  \item for all $1 \leq s < r$, $L^r_{\sigma(s)}$ consists of
    a single path $P$ from $A^{out}_i$ to $A^{in}_j$, where $(i,j) :=
    \sigma(s)$, and $P$ does not share an
    internal vertex with any path in any $L^r_q$ with $q \not=  s$,
  \item $|L_{\sigma(r)}^r| = f(r)$
  \item for all $q > r$ we have $|L_{\sigma(q)}^r| = g(r) = \frac{f(r)}{K\cdot l}$,
    and $L_{\sigma(q)}^r$ is $L ^r_{\sigma(r)}$-minimal, and
  \item for all $1 \leq i \neq j \leq 3k$ and all $P \in L^r_{\sigma^{-1}(i,j)}$ the
    path $P$ has no vertex in common with any $P_t$ for $t \neq i,j$.
  \end{enumerate}

  For $r=1$ we choose a linkage $L_{\sigma(1)}^1$ satisfying
  Condition $2$ and $4$ and for $q > 1$ we choose the other linkages as in
  Condition $3$ and $4$.
  Such linkages exist as $\SSS$ is a clean path system.
  
  Now suppose the linkages have already been defined
  for $r \geq 1$. Let $(i,j) := \sigma(r)$.

  If there is a path
  $P \in L_{i,j}^r$ which, for all $q > r$, is internally
  disjoint to at least $\frac 1c \cdot g(r)$ paths in $L_{\sigma(q)}^r$,
  define $L_{i,j}^{r+1} = \{ P \}$. Let $(s,t) :=
  \sigma(r+1)$ and let $L^{r+1}_{s,t}$ be
  an $A^{out}_s{-}A^{in}_t$-linkage of order $\frac{g(r)}{c}=f(r+1)$
  such that no path in $L^{r+1}_{s,t}$
  has an internal vertex in $V(P) \cup \bigcup_{r'\leq r}
  V(L^r_{\sigma(r')})$ and, furthermore, every path in $L^{r+1}_{s,t}$ is
  disjoint from all $P \in \PPP \setminus \{ P_s, P_t \}$.
  Such a linkage exists by the
  choice of $P$.

  For each $q > r+1$ and $(s',t') =
  \sigma(q)$ choose an $A^{out}_{s'}{-}A^{in}_{t'}$-linkage
  $L_q^{r+1}$ of order $g(r+1) = \frac{g(r)}{(c\cdot K\cdot
    l)}$ satisfying Condition $(4)$ such that every path in it
  has no inner vertex in $V(P) \cup \bigcup_{r'\leq r}
  V(L^r_{\sigma(r')})$ and which is
  $L_{s,t}^{r+1}$-minimal. So in this case, we can construct linkages $L_{i,j}^{r+1}$ as desired.
  
  \smallskip
  
  Otherwise, for all paths $P \in L_{i,j}^r$ there
  are $i',j'$ with $\sigma^{-1}(i',j') > r$ such that $P$ hits more than
  $(1-\frac{1}{c})g(r)$ paths in $L_{i',j'}^r$. As $|L_{i,j}^r|=f(r)
  = g(r)\cdot K\cdot l$, by the pigeon hole principle, there is a $q > r$ such
  that at least $\frac{f(r)}{K}=g(r) \cdot l$ paths in $L_{i,j}^r$ hit all
  but at most $\frac 1c \cdot g(r)$ paths in $L_{\sigma(q)}^r$. Let
  $\QQQ \subseteq L_{i,j}^r$ be the set of such paths. As a
  result, $(L_{\sigma(q)}^r, \QQQ)$ forms a $(g(r), \frac{f(r)}{K})$-web with avoidance $c$. As
  $\frac{f(r)}{K} = g(r) \cdot l$ and the endpoints of the paths in $\QQQ$ are
  in the well-linked set $A_i^{out}\cup A_j^{in} \subseteq A$, this case gives the second possible
  output of the lemma.

  \smallskip
  
  Hence, we may assume that the previous case never happens and eventually
  $r = K$. We now have paths $P_1, \dots, P_{3k}$ and between any pair
  $P_i, P_j$ with $i<j$ a path $L'_{i,j}$ from $A^{out}_i$ to
  $A^{in}_j$ and a path $L'_{j,i}$ from $A^{out}_j$ to
  $A^{in}_i$. Furthermore, for all $(i,j)\not=(i',j')$ these
  paths are pairwise vertex disjoint except possibly at their
  endpoints in case they begin or end in the same path in $\{P_1,
  \dots, P_{3k}\}$. Furthermore, if $s \neq i,j$ then $L'_{i,j}$ is
  disjoint from $P_s$.

  By definition of path systems, $A_i^{in}$
  occurs on $P_i$ before $A_i^{out}$.
  We would now like to use the construction in
  Example~\ref{ex:butterfly} to show that $\bigcup_{1\leq i\leq {3k}}
  P_i \cup \bigcup_{1\leq i \neq j\leq 3k} L'_{i,j}$ contains a
  bidirected clique grid of order $k$ as a butterfly minor. However, the
  path $L'_{i, j}$, which goes from $A^{out}_i$ to $A^{in}_j$ might
  have internal vertices on the subpath of $P_i$ containing
  $A^{in}_i$ or on the subpath of $P_j$ containing $A^{out}_j$. Hence the example is not readily applicable.

  \begin{figure}[t]
    \centering
    \includegraphics[width=12cm]{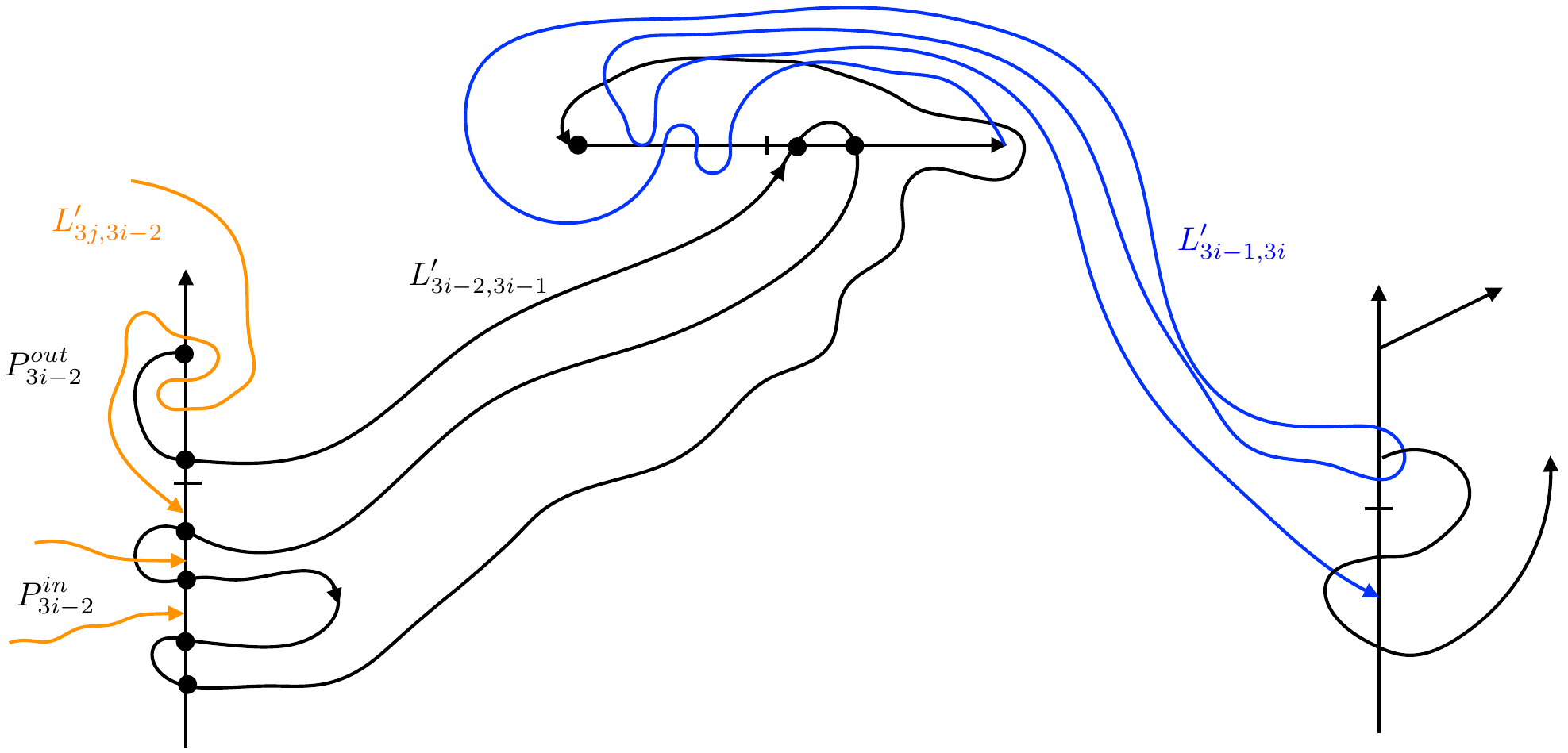}
    \caption{Illustration of the proof of Lemma~\ref{lem:tournament-clique}.}
    \label{fig:tournament-clique}
  \end{figure}
  
  Instead, we construct the
  clique as follows. Essentially, we will use three paths $P_{3i-2},
  P_{3i-1}$, and $P_{3i}$ to construct a single vertex of the
  clique minor. See
  Figure~\ref{fig:tournament-clique} for an illustration of the following
  construction.

  For all $1 \leq i \leq k$, let $P_{3i-2}^{in}$ be the maximal initial subpath of
  $P_{3i-2}$ not including the start vertex of $L'_{3i-2, 3i-1}$, let
  $P_{3i-1}^{in}$ be the minimal initial subpath of $P_{3i-1}$
  containing $A_i^{in}$ and $P_{3i-1}^{out}$ be the minimal final
  subpath of $P_{3i-1}$ containing $A_i^{out}$. Finally, let
  $P_{3i}^{out}$ be the path $P_{3i}$ without its start vertex, which
  is the last vertex of $L'_{3i-1, 3i}$.
  
  W.l.o.g.~we may assume that for all $1 \leq i \neq j \leq 3k$ the path
  $L'_{i,j}$ has exactly one vertex in common with $P_i^{out}$ and
  exactly one vertex in common with $P_j^{in}$.

  For all $1 \leq i \leq k$ and $1 \leq j \neq i \leq k$ let $v_i(j)$ be the first
  vertex of $L'_{3j, 3i-2}$ on $P_{3i-2}$ when traversing $L'_{3j,
    3i-2}$ from beginning to end. We let $I_i := \{ v_i(j) \sth 1 \leq j \neq i
  \leq k \} \subseteq V(P_{3i-2}^{in})$ and  $O_i := \{ v_j(i) \sth 1 \leq j
  \neq i \leq k \} \subseteq V(P_{3i}^{out})$.

  Our goal is to use the paths $P_{3i-2}, P_{3i-1}, P_{3i}, L'_{3i-2,
    3i-1}$,  and $L'_{3i-1, 3i}$ to construct a butterfly minor
  with vertex set $O_i \cup I_i \cup \{ r_i \}$, for some fresh vertex $r_i$,
  and an edge from every $v \in I_i$ to $r_i$ and from $r_i$ to every $v \in
  O_i$. Doing this for all $1 \leq i \leq k$ simultaneously then yields a
  subdivided bidirected clique on $k$ vertices which contains the
  bidirected clique on $k$ vertices as butterfly minor.

  Towards this aim, let $r_i$ be the first vertex of $L'_{3i-1,3i}$ on $P_{3i-1}$, when
  traversing $P_{3i-1}$ from beginning to end.
  Let $L$ be the final subpath of $L'_{3i-1,3i}$ starting at $r_i$. Let $P_{3i-1}'$ be the initial
  subpath of $P_{3i-1}$ ending in $r_i$ and let $P_{3i-2}'$ be the
  initial subpath of $P_{3i-2}$ ending in the start vertex of
  $L'_{3i-2, 3i-1}$. Let $H_i^1 := P'_{3i-2} \cup L'_{3i-2, 3i-1} \cup
  P_{3i-1}'$. Finally, let $H_i^2$ be the digraph obtained from the
  union of $L \cup P_{3i}$ and for each $j \neq i$ the initial subpath of
  $L_{3i, 3j-2}$ up to $v_j(i)$, i.e.~the first vertex $L_{3i, 3j-2}$ has in common
  with $P_{3j-2}$.

  By construction, for all $1 \leq i \leq k$, $H_i^1 \cap H_i^2 = r_i$ and
  $H_i^1$ contains a path from every $v \in I_i$ to $r_i$ and $H_i^2$
  contains a path from $r_i$ to every $v \in O_i$. 
  Furthermore, if
  $j \neq i$ then $H_i^1 \cap H_j^1 = H_i^2 \cap H_j^2 = \emptyset$ and $H_i^1 \cap H_j^2$
  is the unique vertex in $I_i \cap O_j$. Thus, for all $1 \leq i \leq k$,
  $H_i^1$ contains an inbranching $B_i^1 \subseteq H_i^1$ with root $r_i$
  which spans all vertices of $I_i$ and $H_i^2$ contains an
  outbranching $B_i^2$ with root $r_i$ which spans all vertices in
  $O_i$ and contracting, for all $1 \leq i \leq k$,  all edges of $B_i^1, B_i^2$ not incident with
  any vertex in $I_i \cup O_i$ yields a subdivided bidirected clique on
  $k$ vertices as butterfly minor. Contracting the paths of length $2$
  between vertices $r_i, r_j$, $i \neq j$, yields the bidirected clique
  on $k$ vertices, which is the first outcome of the lemma.

\end{proof}

The previous lemma has the following simple corollary.

\begin{corollary}\label{cor:tournament}
  For every $k, p, l, c \geq 1$ there is an integer $l'$ such that the
  following holds. Let $\SSS$ be a clean $l'$-linked path system of
  order $6k^2$. Then either $G$ contains a cylindrical grid of order $k$
  as a butterfly minor or a well-linked
  $(p',l\cdot p')$-web with avoidance $c$, for some $p' \geq p$.
  
\end{corollary}

We are now ready to prove Theorem~\ref{thm:bramble-to-web}.

\smallskip

\begin{proof}[Proof of Theorem~\ref{thm:bramble-to-web}]
  Let $k, p, l, c\geq$ be given as in the statement of the
  theorem. Let $l_1$ be the number $l'$ defined for $k,p,l,c$ in
  Corollary~\ref{cor:tournament}. Let $f$ be the function as defined in
  Lemma~\ref{lem:clean-path-system}. Let $m := \max \{ k, p\}$. We define
  $l' := f(l_1, m, l, c)$.  By Lemma~\ref{lem:clean-path-system}, if $G$ contains a bramble of order
  $l'$, then $G$ contains a clean $l_1$-linked path system $(\PPP,
  \LLL, \AAA)$ of order $m$ or a well-linked $(p', l\cdot p')$-web
  with avoidance $c$, for some $p' \geq m$. As $m\geq p$, the latter
  yields the second outcome of the theorem.

  If instead we obtain the path system, then
  Corollary~\ref{cor:tournament} implies that $G$ contains a cylindrical
  grid of order $m\geq k$, which is the first outcome of the theorem,
  or a well-linked $(p', l\cdot p')$-web with avoidance $c$, for some
  $p'\geq m\geq p$. This yields the second outcome of the theorem.
\end{proof}

We close the section with a simple lemma allowing us to reduce every
web to a web with avoidance $0$.

\begin{lemma}\label{lem:web-avoidance}
  Let $p', q', d$ be integers and let $p \geq \frac{d}{d-1}p'$ and $q \geq
  q'\cdot {p\choose \lceil\frac{1}{d}p\rceil}$ . If a digraph $G$ contains a $(p,
  q)$-web $(\PPP, \QQQ)$ with avoidance $d$ then it contains a $(p', q')$-web with
  avoidance $0$.
\end{lemma}
\begin{proof}
  For all $Q\in \QQQ$ let $\AAA(Q)
  \subseteq \PPP$ be the paths $P\in \PPP$ with $P\cap
  Q=\emptyset$. By definition of avoidance in webs, $|\AAA(Q)|\leq
  \frac 1dp$. Hence, by the pigeon hole principle, there is a set
  $\AAA \subseteq \PPP$ and a set
  $\QQQ'\subseteq \QQQ$ of at least $q'$ paths such that $\AAA(Q) =
  \AAA$ for all $Q\in \QQQ'$. Let $\PPP':=
  \PPP\setminus \AAA$. Hence, $P\cap
  Q\not=\emptyset$ for all $P\in \PPP'$ and $Q\in \QQQ'$.  As $p\geq \frac{d}{d-1}\cdot p'$ we have $p -
  \frac1d p\geq p'$ and hence $|\PPP'|\geq p'$. Furthermore, $\PPP'$
  is $\QQQ'$-minimal. For, Lemma~\ref{lem:min-linkage-closure}
  implies that $\PPP$ is $\QQQ'$-minimal. But $\PPP'$ is obtained from
  $\PPP$ by deleting only those paths $P$ which have an empty
  intersection with every $Q\in \QQQ'$. Clearly, deleting these does
  not destroy minimality. Therefore, $(\PPP',
  \QQQ')$ contains a $(p', q')$-web with avoidance $0$.
\end{proof}

\section{From Webs to Fences}
\label{sec:web-to-grid}

The objective of this section is to show that if a digraph contains a
large well-linked web, then it also contains a big fence whose bottom and top come
from a well-linked set.
We give a precise definition of a fence and
then state the main theorem of this section. The results obtained in
this section are inspired by results in \cite{ReedRST96}, which we
generalise and extend. Some of the results we prove in this section are
stated and proved in greater generality than actually needed in this section.
We need them in full generality in Section~\ref{sec:cylindrical-grids} below.

\begin{definition}[fence]\label{def:fence}
  Let $p,q$ be integers. A $(p, q)$-fence in a digraph $G$ is a
  sequence $\FFF := (P_1,\dots$, $P_{2p}$, $Q_1,\dots,Q_q)$ with the following
  properties:
  \begin{enumerate}%
  \item $P_1,\dots,P_{2p}$ are pairwise vertex disjoint paths of $G$ and
    $\{ Q_1, \dots, Q_q\}$ is an $A$-$B$-linkage for two distinct sets
    $A, B\subseteq V(G)$, called the \emph{top} and \emph{bottom},
    respectively. We denote the top $A$ by $\Top\big(\FFF\big)$ and
    the bottom $B$ by $\Bot\big(\FFF\big)$.
  \item For $1 \leq i \leq 2p$ and $1 \leq j \leq q$, $P_i \cap Q_j$
    is a path (and therefore non-empty).
  \item For $1 \leq j \leq q$, the paths $P_1 \cap Q_j,\dots, P_{2p}
    \cap Q_j$ appear in this order on $Q_j$.
  \item For $1 \leq i \leq 2p$, if $i$ is odd then $P_i \cap
    Q_1,\dots, P_i \cap Q_q$ are in order in $P_i$ and if $i$ is even
    then $P_i \cap Q_q, \dots, P_i \cap Q_1$ are in order in $P_i$.
  \end{enumerate}
  The fence $\FFF$ is \emph{well-linked}
  if $A\cup B$ is well-linked.
\end{definition}

The main theorem of this section is to show that any digraph with a
large web where bottom and top come from a well-linked set contains a
large well-linked fence.

\begin{theorem}\label{theo:main-fence}
  For every $p, q, d\geq 1$ there  are $p', q'$ such that any digraph $G$
  containing a well-linked $(p', q')$-web with avoidance $d$ contains
  a well-linked $(p, q)$-fence.
\end{theorem}

To prove the previous theorem we first establish a weaker version
where instead of a fence we obtain an acyclic grid. We give the
definition first.

\begin{definition}[acyclic grid]\label{def:grid}
  An \emph{acyclic $(p,q)$-grid}  is a $(p,q)$-web $(\PPP, \QQQ)$ with avoidance $0$,
  where $\mathcal{P}=\{P_1,$ $\dots,P_p\}$ and
  $\mathcal{Q}=\{Q_1,\dots,Q_q\}$, such that
  \begin{enumerate}%
  \item for $1 \leq i \leq p$ and $1 \leq j \leq q$, $P_i \cap Q_j$ is
    a path $R_{ij}$,
  \item for $1 \leq i \leq p$, the paths $R_{i1},\dots,R_{iq}$ are in
    order in $P_i$ and
  \item for $1 \leq j \leq q$, the paths $R_{1j},\dots,R_{pj}$ are in
    order in $Q_j$.
  \end{enumerate}
  The definition of \emph{top} and \emph{bottom} as well as well-linkedness is taken over from
  the underlying web.
\end{definition}

\begin{theorem}\label{thm:grid1} \showlabel{thm:grid1}
  For all integers $t,d \geq 1$, there are integers $p, q$ such that
  every digraph $G$ containing a well-linked $(p, q)$-web $(\PPP,
  \QQQ)$ with
  avoidance $d$ contains a well-linked acyclic $(t, t)$-grid.
\end{theorem}

Theorem~\ref{theo:main-fence} is now easily obtained from Theorem~\ref{thm:grid1}
using the following lemma, which is  (4.7) in \cite{ReedRST96}. It is easily
seen that in the construction in \cite{ReedRST96} the top and bottom of the fence are subsets of the top and
bottom of the acyclic grid it is constructed from.

\begin{figure}
  \centering
\includegraphics{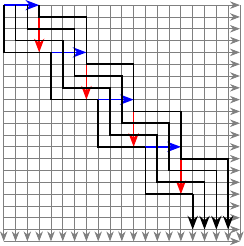}  
  \caption{Constructing a fence in a grid.}
  \label{fig:grid-fence}
\end{figure}

\begin{lemma}\label{grid}
  For every integer $p \geq 1$, there is an integer $p'' \geq 1$ such that every digraph
  with a $(p'',p'')$-grid has a $(p,p)$-fence such that the top and
  bottom of the fence are subsets of the top and bottom of the grid, respectively.
\end{lemma}

As we will be using this result frequently, we demonstrate the
construction in Figure~\ref{fig:grid-fence}. Essentially, we construct
the fence inside the grid by starting at the top left corner and then
taking alternatingly vertical and horizontal parts of the acyclic grid. This
yields the vertical paths of the fence (marked as dotted black lines
from the top left to the bottom right in the
figure). To get the alternating horizontal paths we use the short
paths
marked by solid black arrows in the figure.
Note that by this construction, the horizontal and vertical paths of
the new fence each contain subpaths of the horizontal as well as
vertical paths of the original grid.

We now turn towards the proof of Theorem~\ref{thm:grid1}.
We first need some definitions.

\begin{definition}\label{def:split-edge-path}
  Let $\QQQ^*$ be a linkage. Let $\QQQ\subseteq \QQQ^*$ be a sublinkage of order $q$ and let $P$ be a path intersecting every
    path in $\QQQ$.
  \begin{enumerate}
  \item   Let $r\geq 0$. An edge $e\in E(P)-E(\QQQ^*)$ is \emph{$r$-splittable
      with respect to $\QQQ$}\index{$r$-splittable}\index{splittable} (and $\QQQ^*$) if there is a set $\QQQ'\subseteq \QQQ$
    of order $r$ such that $Q\cap P_1 \not=\emptyset$ and $Q\cap P_2
    \not=\emptyset$ for all $Q\in \QQQ'$, where $P_1, P_2$ are the two subpaths of $P-e$
    such that $P = P_1eP_2$.
  \item Let $P$ be a path. $\QQQ$ is a
    \emph{$q$-segmentation}\index{segmentation} of $P$ if $P$ can be
    divided into subpaths $P = P_1e_1 \dots
    P_{q-1}e_{q-1}P_{q}$, for suitable edges $e_1, \dots,  e_{q-1}\in
    E(P)$, such that $\QQQ$ can be ordered as $(Q_1, \dots, Q_{q})$ and
    $\emptyset \neq V(Q_i)\cap V(P)\subseteq V(P_i)$.
    We say that $\QQQ$ is a
    segmentation of $P$ with respect to $\QQQ^*$ if $e_i \in
    E(P)\setminus E(\QQQ^*)$ for all $1\leq i
    \leq q-1$.
  \end{enumerate}
\end{definition}

We next lift the previous definition to pairs $(\PPP, \QQQ)$ of
linkages.

\begin{definition}\label{def:split-edge}
  Let $\PPP$ and $\QQQ^*$ be linkages and let $\QQQ\subseteq \QQQ^*$
  be a sublinkage of order $q$. Let $r\geq 0$.
  \begin{enumerate}
  \item An \emph{$(r, q')$-split of $(\PPP, \QQQ)$ (with respect to
      $\QQQ^*$)} is a pair $(\PPP', \QQQ')$ of linkages of order
    $r = |\PPP'|$ and $q'=|\QQQ'|$ with $\QQQ'\subseteq \QQQ$ such that
   there is a path $P \in \PPP$ and edges $e_1, \dots, e_{r-1}\in E(P)\setminus
   E(\QQQ^*)$ such that $P = P_1e_1P_2\dots e_{r-1} P_r$ and $\PPP' := (P_1,
   \dots, P_r)$ and every $Q\in \QQQ'$ can be divided
    into subpaths $Q_1, \dots, Q_r$ such that
    $Q = Q_1e'_1\dots e'_{r-1}Q_r$, for suitable edges
    $e'_1, \dots, e'_{r-1} \in E(Q)$, and
    $\emptyset \neq V(Q)\cap V(P_i) \subseteq V(Q_{r+1-i})$, for all $1\leq i \leq r$.
  \item An $(r, q')$-segmentation is a pair $(\PPP', \QQQ')$ where
    $\PPP'$ is a linkage of order $r$ and $\QQQ'$ is a linkage of order
    $q'$ such that $\QQQ'$ is a $q'$-segmentation of every path $P_i$ into
    segments $P^i_1e_1P^i_2\dots e_{q'-1}P^i_{q'}$. %
  \item 
    A segmentation $(\PPP', \QQQ')$ is \emph{strong} if for every $Q \in \QQQ'$
    and $P_i \in \PPP'$, if $Q$ intersects $P_i$ in segment
    $P^i_l$, for some $l$, then $Q$ intersects every $P_j\in \PPP'$ in segment
    $P^j_l$.
    
    We say that $(\PPP', \QQQ')$ is a (strong) $(r, q')$-segmentation of
    $(\PPP, \QQQ)$ if $\QQQ' \subseteq \QQQ$ and every path in $\PPP'$ is a subpath of
    a path in $\PPP$.
  \end{enumerate}
  An $(r, q)$-split $(\PPP, \QQQ)$ and an $(r, q)$-segmentation $(\PPP, \QQQ)$
  is well-linked if the set of start and end vertices of the paths in
  $\QQQ$ is a well-linked set.
\end{definition}

We also need a slightly weaker version of a split.

\begin{definition}\label{def:weak-split}
  Let $\PPP$ and $\QQQ^*$ be linkages and let $\QQQ \subseteq \QQQ^*$ be a sublinkage of
  order $q$. Let $r \geq 0$. A \emph{weak $(r, q')$-split of $(\PPP, \QQQ)$} is
  defined in the same way as an $(r, q')$-split but without the
  requirement that all paths in $\PPP'$ are subpaths of the same path in
  $\PPP$.

  Formally, a \emph{weak $(r, q')$-split of $(\PPP, \QQQ)$ (with
    respect to $\QQQ^*$)} is a pair $(\PPP', \QQQ')$ of linkages of order
  $r = |\PPP'|$ and $q' = |\QQQ'|$ with $\QQQ' \subseteq \QQQ$ such that
  $\PPP'$ can be ordered $\PPP' := (P_1, \dots, P_r)$ in such a way that
  every $P_i$ is a subpath of a path $P \in \PPP$ and every
  $Q \in \QQQ'$ can be divided into subpaths $Q_1, \dots, Q_r$ such that
  $Q = Q_1e'_1\dots e'_{r-1}Q_r$, for suitable edges
  $e'_1, \dots, e'_{r-1} \in E(Q)$, and
  $\emptyset \neq V(Q)\cap V(P_i) \subseteq V(Q_{r+1-i})$, for all $1\leq i \leq r$.
\end{definition}

\begin{remark}
  To simplify the presentation we agree on the following notation when
  working with $r$-splits as in the previous definition. If $\PPP$ only
  contains a single path $P$ we usually simply write an $(r, q)$-split of
  $(P, \QQQ)$ instead of $(\{P\}, \QQQ)$. Furthermore, as the order in an
  $r$-split is important, we will often write $r$-splits as
  $\big( (P_1, \dots, P_r), \QQQ'\big)$.
\end{remark}

Note that in an $(r, q')$-split $(\PPP', \QQQ')$ of $(\PPP, \QQQ)$, all paths in
$\PPP'$ are obtained from a single path in $\PPP$. The previous concepts are
illustrated in Figure~\ref{fig:split-segment}. Part a) of the figure
illustrates a (strong) segmentation, where the vertical (black) paths
are the paths in $\QQQ$ which segment the two paths in $\PPP$ (red and blue
dotted). An alternative way of looking at segmentations is shown in
Part c). Here, a $(4, 3)$-segmentation is shown where the
paths in $\QQQ$ are again the vertical (black) paths and the paths in $\PPP$ are the
straight (red) horizontal paths. This way of viewing segmentations
will be 
useful in Section \ref{sec:cylindrical-grids} below. In Part b) of Figure~\ref{fig:split-segment}, a
single path $P$ (marked in red) is split at $5$ edges, marked by the arrows on $P$.
The paths in $\QQQ$ involved in the split are marked by solid black
vertical paths whereas the paths in $\QQQ$ which do not split $P$ are
displayed in light grey.

\begin{figure}
  \centering
  \begin{tabular}{cc}
  \includegraphics[height=4cm]{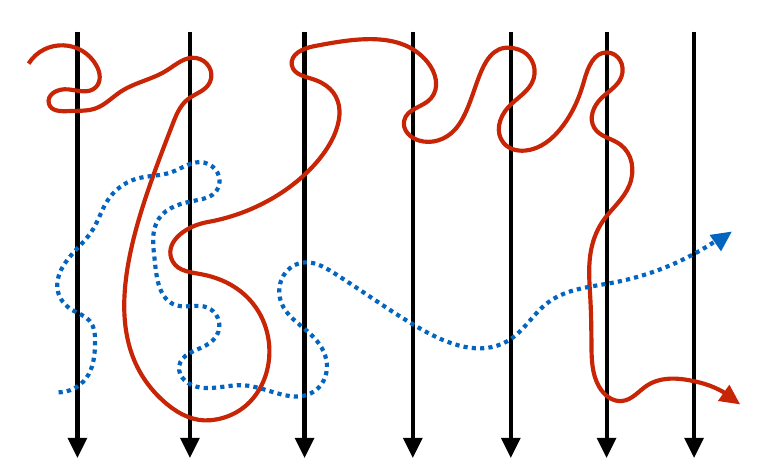}&
  \includegraphics[height=4cm]{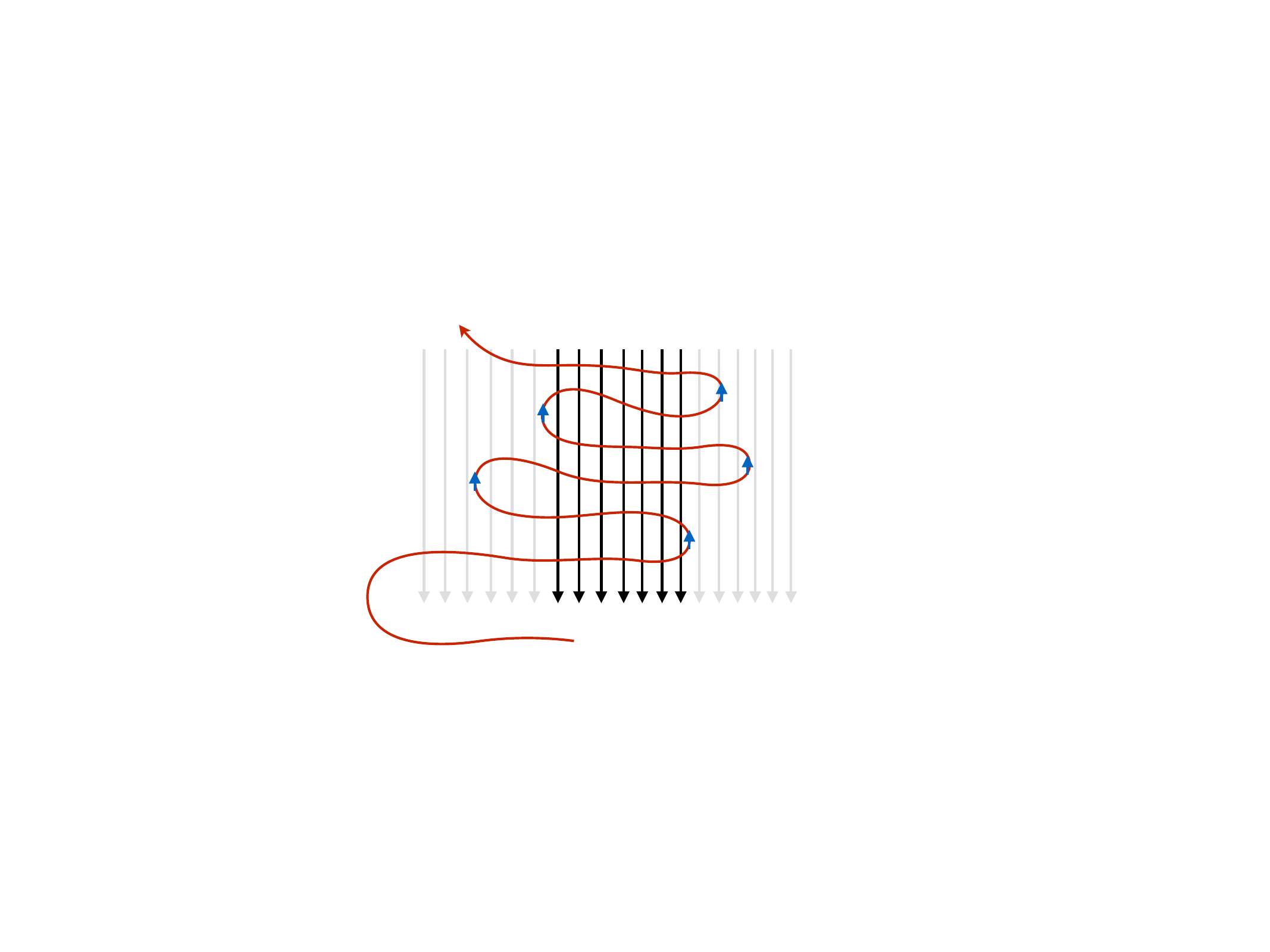}\\
  a)  & b)\\
  \includegraphics[height=4cm]{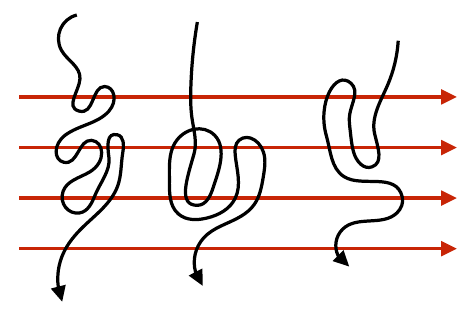}\\c)
\end{tabular}
  \caption{a) A $(2,7)$-segmentation and b) a $(6, 7)$-split. c) An
    alternative view of segmentations, here a $(4, 3)$-segmentation.}
  \label{fig:split-segment}\label{fig:splitting}
\end{figure}

Note that we can make any segmentation $(\PPP, \QQQ)$ strong by
sacrificing either some paths in $\PPP$ or some paths in $\QQQ$. We
will need both ways in the sequel.

\begin{lemma}\label{lem:strong-segmentation}
  \begin{enumerate}
  \item For every $p, q\geq 1$, if
    $(\PPP, \QQQ)$ is a $(p', q)$-segmentation for some
    $p'\geq q!\cdot p$, then there is a
    $\PPP'\subseteq \PPP$ of order $p$ such that $(\PPP', \QQQ)$ is a
    strong $(p, q)$-segmentation.
  \item For every $p, q\geq 1$ there is a $q'\geq 1$ such that if
    $(\PPP, \QQQ)$ is a $(2p, q')$-segmentation, then there is 
    $\QQQ'\subseteq \QQQ$ of order $q$ and $\PPP' \subseteq \PPP$ of order $p$ such that
    $(\PPP, \QQQ')$ is a
    strong $(p, q)$-segmentation.
  \end{enumerate}
\end{lemma}
\begin{proof}
  We first prove Part (1).
  For every $P\in \PPP$ let $P(\QQQ) := (Q_{i_1}, \dots, Q_{i_q})$ be
  the order in which the paths in $\QQQ$ occur on $P$ when traversing $P$
  from beginning to end. As $(\PPP, \QQQ)$ is a segmentation, this is
  well defined.
  As $p' \geq q!\cdot p$ there is a subset $\PPP'\subseteq \PPP$ of
  order $p$ such
  that $P(\QQQ) = P'(\QQQ)$ for all $P, P'\in \PPP'$. Hence, $(\PPP',
  \QQQ')$ is a strong $(p, q)$-segmentation as required.

  We now prove the second part.
  Let $x_{2p} := q$ and for $0\leq j < 2p$ let $x_j := (x_{j+1}-1)^2+1$.
  We set $q' := x_{0}$.

  Fix an ordering $(P_1, \dots, P_{2p})$ of
  $\PPP$ and
  for every $P_i\in \PPP$ let $P^i_1, \dots, P^i_{q'}$ be the segments of
  $P_i$ such that $P = P^i_1 e_1 \dots e_{q'-1} P^i_{q'}$ and every $Q\in
  \QQQ$ intersects $P_i$ in exactly one segment.
  We will also fix an ordering $(Q_1, \dots, Q_{q'})$ of $\QQQ$.

  For every subset
  $\QQQ^* := \{ Q_{i_1}, \dots, Q_{i_k}\}\subseteq \QQQ$, defined by
  some indices
  $i_1 < \dots < i_k$, and every $P\in \PPP$
  let $\pi(P, \QQQ^*) := (j_1, \dots, j_k)$ be such that the paths in
  $\QQQ^*$  occur on $P$ (when traversing $P$ from beginning to end)
  in the order $Q_{j_1},  \dots, Q_{j_k}$.

  For $0\leq i \leq 2p$ we will construct a
  sequence $\QQQ_i\subseteq \QQQ$ of order $x_i$ such that there is an ordering
  $(Q_1, \dots, Q_{x_i})$ of the paths in $\QQQ_i$ and for all $1\leq j \leq
  i$, the paths in $\QQQ_i$ occur on $P_j$ in this order (when traversing from
  beginning to end) or in reverse order $(Q_{x_i}, \dots, Q_1)$.

  Choose $\QQQ_0
  = \QQQ$ which has order $x_0$. Now suppose $\QQQ_i$
  has already been constructed. Let $\pi(P_{i+1}, \QQQ_i) := (j_1, \dots,
  j_{x_i})$. By the choice of $x_i$ and Theorem~\ref{thm:szekeres}, there is a
  subsequence $j_{a_1}, \dots, j_{a_{x_{i+1}}}$ such that $j_{a_1} < \dots <
  j_{a_{x_{i+1}}}$ or $j_{a_1} > \dots > j_{a_{x_{i+1}}}$. We set $\QQQ_{i+1}
  :=
  \{
  Q_{j_{a_l}} \sth 1\leq l \leq x_{i+1} \}$.

  Now let $\QQQ_{2p}$ be as constructed. We set $\QQQ' := \QQQ_{2p}$. Let $(Q_1,
  \dots, Q_{2p})$ be the order in which the paths in $\QQQ'$ occur on $P_1$. Then
  for every $1\leq i \leq 2p$, the paths in $\QQQ'$ occur on $P_i$ in
  this order or in reverse order. Hence, there is a subset $\PPP'\subseteq
  \PPP$
  of order $p$ such that all paths in $\QQQ'$ occur on every $P \in \PPP'$ in
  the same order.
\end{proof}

We next show the following lemma, which is essentially shown in \cite{ReedRST96}.

\begin{lemma}\label{lem:split-or-segment-path}
  Let $r, s \geq 0$. Let $\QQQ^*$ be a linkage and let $\QQQ\subseteq
  \QQQ^*$ be a sublinkage of order $q$. Let $P$ be
  a path intersecting every path in $\QQQ$. If $q\geq r\cdot s$ then $P$
  contains an $r$-splittable edge with respect to $\QQQ$ and $\QQQ^*$
  or there is an
  $s$-segmentation $\QQQ'\subseteq \QQQ$ of $P$ with respect to $\QQQ^*$.
\end{lemma}
\begin{proof}
  Let $(Q_1,\dots,Q_{r\cdot s}) \subseteq \QQQ$ be a
  linkage of order $r \cdot s$.  For $1 \leq j \leq
  r\cdot s$, let $F_j$ be the minimal subpath of $P$ that includes
  $V(P \cap Q_j)$.
  Note that as $\QQQ^*$ is a linkage, if
  $Q \not= Q'\in \QQQ^*$, $v\in V(P)\cap V(Q)$ and $v'\in V(P)\cap
  V(Q')$ then every subpath of $P$ containing $v$ and $v'$ must
  contain an edge $e\not\in E(\QQQ^*)$.

  Suppose first that for some edge $e\in E(P) \setminus E(\QQQ^*)$
  there are at least $r$ distinct  values 
  $j_1,\dots, j_r$ such that $e$
  belongs to $F_{j_1}, \dots, F_{j_r}$. Then $e$ is $r$-spittable as witnessed by $\QQQ' := \{ Q_{j_1}, \dots, Q_{j_r}\}$.

    Thus we may assume that every edge of $E(P) \setminus
    E(\QQQ^*)$ occurs in $F_j$ for fewer than $r$ values of $j$.
    Consequently there are $s$ distinct numbers $j_1,\dots,j_s$ so
    that $F_{j_1},\dots,F_{j_s}$ are pairwise vertex-disjoint.
    Thus $\QQQ'  :=  (Q_{j_1},\dots,Q_{j_s})$ is an $s$-segmen\-tation of $P$.
\end{proof}

The lemma has the following consequence which we will use frequently
below.

\begin{corollary}\label{cor:split-minimal}
  Let $H$ be a digraph and let $\QQQ^*$ be a linkage in $H$ and let
  $\QQQ\subseteq \QQQ^*$ be a linkage of order $q$. Let $P\subseteq H$ be a
  path intersecting every path in $\QQQ$. Let $c\geq 0$ be such that
  for every edge $e\in E(P)\setminus E(\QQQ^*)$ there are no $c$ pairwise vertex
  disjoint paths in $H-e$ from $P_1$ to $P_2$, where $P = P_1eP_2$. For all $s,
  r\geq 0$, if $q\geq (r+c)\cdot s$ then
  \begin{enumerate}[label=\emph{\alph*)}]
  \item there is an $s$-segmentation $\QQQ'\subseteq \QQQ$ of $P$ with respect to $\QQQ^*$ or
  \item a $(2, r)$-split $\big( (P_1, P_2), \QQQ''\big)$ of $(P, \QQQ)$ with
    respect to $\QQQ^*$.
  \end{enumerate}
\end{corollary}
\begin{proof}
  By Lemma~\ref{lem:split-or-segment-path}, there is an $s$-segmentation of
  $P$ or an $(r+c)$-splittable edge $e\in E(P)-E(\QQQ^*)$. In the second
  case, let $P = P_1eP_2$ and let $\QQQ' \subseteq \QQQ$ be of order $(r+c)$ witnessing that
  $e$ is $(r+c)$-splittable. Thus, every path in $\QQQ'$ intersects
  $P_1$ and $P_2$. As there are at most $c$ disjoint paths from $P_1$
  to $P_2$ in $H-e$, at most $c$ of the paths in $\QQQ'$ hit $P_1$ before
  they hit $P_2$. Hence, there is a subset $\QQQ''\subseteq \QQQ'$ of
  order $r$ such that for all $Q\in \QQQ''$ the last vertex of $V(P_2\cap Q)$
  occurs before the first vertex of $V(P_1\cap Q)$. Hence, $\big((P_1,
  P_2), \QQQ''\big)$ is a $(2, r)$-split.
\end{proof}

We will mostly apply the previous lemma in a case where $H\subseteq G$
is a subgraph induced by two linkages $\PPP, \QQQ$ and
$P\in \PPP$.

We now present one of our main constructions showing that for every
$x, y, q$ every web of sufficiently high order either contains an
$(x, q)$-segmentation or a $(y, q)$-split.
This construction will again be used in
Section~\ref{sec:cylindrical-grids} below.
We first refine the definition of webs from
Section~\ref{sec:bramble}. The difference between the webs (with
avoidance $0$) used in Section~\ref{sec:bramble} and the webs with
linkedness $c$ defined here is that we no longer require that in a web
$(\PPP, \QQQ)$, $\PPP$ is $\QQQ$-minimal. Instead we require that in
every path $P$, if we split $P$ at an edge $e$, i.e. $P = P_1eP_2$, then there are at most $c$
paths from $P_1$ to $P_2$ in $\PPP\cup \QQQ$. This is necessary as in the various
constructions below, minimality will not be preserved but this forward
path property is preserved. We give a formal definition now.

\begin{definition}[$(p,q)$-web with linkedness $c$]\label{def:web-linkedness}
  Let $p,q, c \geq 0$ be integers and let $\QQQ^*$ be a linkage. A
  \emph{$(p,q)$-web with linkedness $c$ with respect to $\QQQ^*$} in a
  digraph $G$ consists of two linkages $\mathcal{P}=\{P_1,\dots,P_p\}$
  and $\mathcal{Q}=\{Q_1,\dots,Q_q\} \subseteq \QQQ^*$ such that
  \begin{enumerate}%
  \item $\mathcal{Q}$ is a $C{-}D$ linkage for
    two distinct vertex sets $C,D \subseteq V(G)$ and $\mathcal{P}$ is an $A{-}B$ linkage for two distinct vertex
    sets $A,B \subseteq V(G)$,
  \item for $1 \leq i \leq q$, $Q_i$ intersects every path $\mathcal{P}\in\PPP$,
  \item for every $P\in \PPP$ and every edge $e\in E(P)\setminus
    E(\QQQ^*)$ there are at most $c$ disjoint paths from $P_1$ to
    $P_2$ in $\PPP\cup \QQQ$ where $P_1, P_2$ are the subpaths of $P$ such that $P=P_1eP_2$.
  \end{enumerate}
  The set $C\cap V(\QQQ)$ is called the \emph{top} of the web, denoted
  $\Top\big( (\PPP,
  \QQQ) \big)$, and $D\cap V(\QQQ)$ is the \emph{bottom} $\Bot\big( (\PPP,
  \QQQ) \big)$.  The web $(\PPP, \QQQ)$ is \emph{well-linked} if $C\cup D$ is
  well-linked.
\end{definition}

\begin{remark}
  Every $(p,q)$-web with avoidance $0$ is a $(p,q)$-web with
  linkedness $p$. The linkedness $p$ follows from Lemma~\ref{lem:no-forward-paths}.
\end{remark}

\begin{lemma}\label{lem:split-grid-segmentation-refined}\label{lem:split-or-segment}
	For all $c, x, y, q\geq 0$ and $p\geq x$ there is a  number
	$q'$ such that if $G$ contains a
	$(p, q')$-web
	$\WWW :=  (\PPP, \QQQ)$ with linkedness $c$,
	then
        \begin{enumerate}
        \item $G$ contains a $(y, q)$-split $(\PPP', \QQQ')$ of
          $(\PPP, \QQQ)$ or
        \item an $(x, q)$-segmentation
          $(\PPP', \QQQ')$ of $(\PPP, \QQQ)$  with the
          following additional properties: at most $y-1$ paths in
          $\PPP'$ are subpaths of the same path in $\PPP$. In
          addition, for every path $P\in \PPP$, either
          $V(P)\subseteq V(\PPP')$ or
          $V(P) \cap V(\PPP') = \emptyset$. Finally, if
          $P_1, \dots, P_l \in \PPP'$ are the subpaths of
          the same path $P \in \PPP$, for some $l\leq y-1$, so that $P_1, \dots, P_l$ occur
          on $P$ in this order, then for all $l_1 < l_2 \leq l$, no
          path $Q\in \QQQ'$ intersects any $P_{l_2}$ after the first
          vertex $Q$ has in common with $P_{l_1}$.
        \end{enumerate}
	Furthermore, if $\WWW$ is well-linked then so is $(\PPP',
	\QQQ')$.
\end{lemma}
\begin{proof}
   Fix $\QQQ^*	:= \QQQ$ for the rest of the proof. All applications of
   Corollary~\ref{cor:split-minimal} will be with respect to this
	original linkage $\QQQ^*$.
	
	For all $0\leq i\leq xy$ we define values $q_i$ inductively as
	follows. We set $q_{xy} := q$ and $q_{i-1} := q_i\cdot
	(q_i+c)$. We define $q' := q_0$.

	Let $(\PPP,\QQQ)$ be a $(p, q_0)$-web. For $0\leq i\leq
	xy$ we construct a
	sequence
	$\MMM_i := (\PPP^i, \QQQ^i, \Sseg^i,
	\Ssplit^i)$, where $\QQQ^i,
	\Sseg^i$ and $\Ssplit^i$ are
	linkages of order $q_i, x_i$ and $y_i$, respectively, and $\PPP^i$ is a linkage of order at least $p-i$ such that $\QQQ^i\subseteq\QQQ^*$ and $(\Sseg^i, \QQQ^i)$ is an
	$(x_i, q_i)$-segmentation  and $(\Ssplit^i, \QQQ^i)$ is a
	$(y_i, q_i)$-split of $(\PPP, \QQQ)$. Furthermore, $(\PPP^i, \QQQ^i)$
        has linkedness $c$ and $\PPP^i\cap \Sseg^i = \emptyset$, for all $i$. Recall that, in
	particular, this means
	that the paths in $\Ssplit^i$ are the subpaths of a single path in
	$\PPP$ that is split by edges $e\in E(P)\setminus E(\QQQ^*)$.
	
	We first set  $\MMM_0 := (\PPP^0 := \PPP, \QQQ^0:=\QQQ^*, \Sseg^0 :=
	\emptyset, \Ssplit^0 := \emptyset)$. Clearly,
	this satisfies the conditions on $\MMM_0$
	defined above.
	
	Now suppose that $\MMM_i$ has already been
	defined for some $i\geq 0$.
	If $\Ssplit^i\setminus \Sseg^i = \emptyset$, we first choose a path $P \in
	\PPP^i$ and set $\Ssplit =
	\{ P \}$ and $\PPP^{i+1} := \PPP^i\setminus
	\{P\}$.

	Otherwise, if $\Ssplit^i\setminus
	\Sseg^i \not= \emptyset$, we set $\Ssplit :=
	\Ssplit^i$ and $\PPP^{i+1} := \PPP^i$. In both cases, we set $\Sseg :=
	\Sseg^i$.
	
	Now, let $P\in \Ssplit \setminus  \Sseg$. We apply
	Corollary~\ref{cor:split-minimal} to
	$P, \QQQ^i$ and  $\QQQ^*$ setting $r =s = q_{i+1}$ in the
        statement of the lemma. If we
	get a $q_{i+1}$-segmentation $\QQQ_1\subseteq \QQQ^i$ of $P$ with respect to
	$\QQQ^*$ we set
	\[
	\QQQ^{i+1} := \QQQ_1, \qquad \Sseg^{i+1} :=
	\Sseg^i \cup
	\{ P \}\quad \text{ and }\quad \Ssplit^{i+1} :=
	\Ssplit.
	\]
	Otherwise, we get a $(2, q_{i+1})$-split $\big( (P_1, P_2), \QQQ_2)$  where $\QQQ_2\subseteq \QQQ^i$.
	Then we set
	\begin{eqnarray*}
		\QQQ^{i+1} &:=&
		\QQQ_2,\\
		\Sseg^{i+1} &:=& \Sseg^i\quad\text{ and }\\
		\Ssplit^{i+1} &:=& (\Ssplit^i
		\setminus \{ P \}) \cup \{ P_1, P_2 \}.
	\end{eqnarray*}
	It is easily verified that the conditions for $\MMM^{i+1} := (\PPP^{i+1},
	\QQQ^{i+1},
	\Sseg^{i+1}, \Ssplit^{i+1})$ are maintained. In particular, the
	linkedness $c$ of $(\PPP^{i+1}, \QQQ^{i+1})$ is preserved as deleting or
	splitting
	paths cannot increase forward connectivity (in contrast to the
	minimality property). This concludes the construction of
	$\MMM_{i+1}$.
	
	We stop this process as soon as for some $i$
	\begin{enumerate}
		\item $|\Ssplit^i| \geq y$ or
		\item $|\Sseg^i|\geq x$ and $|\Ssplit^i\setminus \Sseg^i|=0$.
	\end{enumerate}
	
	Note that in the construction, after every $y$ steps, either we have
	found a set $\Ssplit^i$ of size $y$ or $\Ssplit^i\setminus \Sseg^i$
	has become empty at some point. More precisely, we start with a path
	$P\in \PPP$ to put into $\Ssplit$. Then in every step we try to
	split a path in $\Ssplit$. If this works and we find a splittable
	edge, we add both subpaths  to $\Ssplit$. Otherwise, the path will
	be added to $\Sseg$ and then we do not try to split it again later
	on. Hence, continuing in this way, for the path $P$ we started with,
	either it will be split $y$ times and we stop the construction or at
	some point all its subpaths generated by splitting will also be
	contained in $\Sseg$. We then stop working on $P$ and choose a new
	path $P'\in \PPP$ for which we repeat the process.
	
	Hence, in the construction above, in each step we either increase
	$x_i$ and decrease $|\Ssplit^i \setminus \Sseg^i|$  or we increase $y_i$.
	After at most
	$i \leq xy$ steps, either we have constructed a set
	$\Sseg^i$ of order $x$ and $\Ssplit^i\setminus\Sseg^i=\emptyset$ or  a set
	$\Ssplit^i$ of order $y$.
	
	If we found a set $\Ssplit^i$ of order $y$ then we can choose any
	set $\QQQ'\subseteq \QQQ^i$ of order $q$ and $(\Ssplit^i,
	\QQQ')$ is the first outcome of the lemma.
	
	So suppose that instead we get a set $\Sseg := \Sseg^i$ of order
	$y'\geq y$ such that $\Sseg\setminus \Ssplit^i=\emptyset$. This
	implies that we get a segmentation $(\Sseg, \QQQ^i)$ such that for
	every path $P\in \PPP$, either $V(P)\cap V(\Sseg) = \emptyset$ or
	$V(P)\subseteq V(\Sseg)$. Note further that as we are in the second
	case, no path $P$ was split $y$ or more times. Hence at most $y-1$
	paths in $\Sseg$ belong to the same path $P\in
        \PPP$. Furthermore, if $P_1, \dots, P_l$ belong to the same
        path $P\in\PPP$, then they are obtained from $P$ by
        splits. Hence, if they occur on $P$ in this order, then the paths
        in $\QQQ^i$ have to go through $P_1, \dots, P_l$ in the
         reverse order $P_l, ..., P_1$. 
	Hence, $(\Sseg, \QQQ^i)$ satisfy the conditions for the second
	outcome of the lemma.

        Finally, it is easily seen that if $\WWW$ is well-linked then
        so is $(\Ssplit^i, \QQQ')$ (in case of the first outcome) or
        $(\Sseg, \QQQ^i)$ (in case of the second outcome). This
        concludes the proof of the lemma.
\end{proof}

Lemma~\ref{lem:split-grid-segmentation-refined}
and the second part of Lemma~\ref{lem:strong-segmentation} together
imply the following corollary needed in
Section~\ref{sec:cylindrical-grids} below.

\begin{corollary}\label{cor:split-strong-segmentation}
	For all $c, x, y, q\geq 0$ and $p\geq 2x$ there is a  number
	$q'$ such that if $G$ contains a
	$(p, q')$-web
	$\WWW :=  (\PPP, \QQQ)$ with linkedness $c$,
	then
	$G$ contains
        \begin{enumerate}
        \item a $(y, q)$-split $(\PPP', \QQQ')$ of $(\PPP, \QQQ)$  or 
        \item a strong $(x, q)$-segmentation $(\PPP', \QQQ')$ of
          $(\PPP, \QQQ)$ with the following additional
          properties: 
          at most $y-1$ paths in $\PPP'$ are subpaths of the same path
          in $\PPP$. In addition, if
          $P_1, \dots, P_l \in \PPP'$ are the subpaths of
          the same path $P \in \PPP$, for some $l\leq y-1$, so that $P_1, \dots, P_l$ occur
          on $P$ in this order, then for all $l_1 < l_2 \leq l$, no
          path $Q\in \QQQ'$ intersects any $P_{l_2}$ after the first
          vertex $Q$ has in common with $P_{l_1}$.
        \end{enumerate}
	
	Furthermore, if $\WWW$ is well-linked then so is $(\PPP',
	\QQQ')$.
\end{corollary}

Note that in contrast to Lemma
\ref{lem:split-grid-segmentation-refined}, in the previous corollary
it is no longer true that for every path $P\in \PPP$, either
$V(P)\subseteq V(\PPP')$ or
$V(P) \cap V(\PPP') = \emptyset$.

Consider  the case that the outcome of the previous lemma is a
$y$-split. This case is
illustrated in
Figure~\ref{fig:split-segment} b). We call the structure that we
obtain in this case a \emph{pseudo-fence}.

\begin{definition}[pseudo-fence]\label{def:pseudo-fence}
  A \emph{$(p, q)$-pseudo-fence} is a pair $(\PPP := (P_1, \dots,
  P_{2p})$, $\QQQ)$ of pairwise disjoint paths, where $|\QQQ|=q$,
  such that each $Q\in \QQQ$ can be divided into segments $Q_1, \dots,
  Q_{2p}$ occurring in this order on $Q$ such that for all $i$, each $P_i$ intersects all $Q\in \QQQ$ in their
  segment $Q_i$ and $P_i$ does not intersect any $Q$ in
  any other
  segment. Furthermore, for all $1\leq i \leq p$, there is an edge $e_i$ connecting
  the endpoint of $P_{2i}$ to the start point of $P_{2i-1}$ or an edge $e_i$ connecting
  the endpoint of $P_{2i-1}$ to the start point of $P_{2i}$.

  A \emph{weak $(p, q)$-pseudo-fence} $(P, Q)$ is defined in the same way except that instead of the edges $e_i$ there is for each $1 \leq i \leq p$ a path $L_i$ which connects the endpoint of one of the paths $P_{2i}$ or $P_{2i-1}$ to the start vertex of the other path such that $L_i$ does not intersect any $P_s$, $1 \leq s \leq 2p$ and for all $Q \in \QQQ$, $V(L_i \cap Q) \subseteq V(Q_{2i} \cup Q_{2i-1})$. 

  The \emph{top} of $(\PPP, \QQQ)$ is the set of start vertices and
  the bottom the set of end vertices of $\QQQ$.
\end{definition}

The next lemma follows immediately from the definitions.

\begin{proposition}\label{prop:split-is-pseudo-fence}
  Let $(\PPP', \QQQ')$ be a $(2y, q)$-split of some pair
  $(\PPP, \QQQ)$ of linkages. Then $(\PPP', \QQQ')$ form a $(y,
  q)$-pseudo-fence.
\end{proposition}

In the following three lemmas
(which generalise the results in \cite{ReedRST96}),
we show how in each of the two cases of the Lemma~\ref{lem:split-grid-segmentation-refined} we get
an acyclic grid.
We first need some preparation.

\begin{lemma}\label{lem:new-good-tuples}
  There are functions $f_r, f_p \sth \N\rightarrow \N$ such that for
  every $k\geq 1$, if $\PPP$ and
  $\RRR$ are linkages of order at least $f_p(k)$ and $f_r(k)$, respectively, and $(\PPP,
  \RRR)$ is a strong $(f_p(k), f_r(k))$-segmentation, then there is a sequence $\hat P_1,
  \dots, \hat P_k \in \PPP$ and a path $A$ consisting of subpaths of
  paths in $\PPP$ and $\RRR$ such that
   $A$ intersects $\hat P_1, \dots, \hat P_k$ in this order,
   $A\cap \hat P_i$ is a path for all $1\leq i \leq k$ and the start
   and end vertex of $A$ are on paths from $\RRR$.

   Furthermore,
   \begin{enumerate}
   \item there exists such a sequence $\hat P_1, \dots, \hat P_k$ and
     path $A$
     such that $A$ starts at the start vertex of a path $R \in \RRR$ and
   \item there exists such a sequence $\hat P_1, \dots, \hat P_k$ and
     path $A$
     such that $A$ ends at the endpoint of a path $R\in \RRR$.
   \end{enumerate}
 \end{lemma}
 \begin{proof}
  We define $f_r(k) :=  3k\cdot f_p(k)\cdot \binom{f_p(k)}{3k}\cdot {3k}!  + (k+1)\cdot
  \binom{f_p(k)}{k}\cdot k!$. To define $f_p(k)$ we set $f_p(1) = 1$. For
  $k>1$ let $b_0 := 2^{f_p(k-1)+k}$ and let, for $0<i\leq k$, $b_i :=
  (2b_{i-1})^{b_{i-1}+k}$. We set $f_p(k) := b_k$.

  We assume that $|\PPP|=f_p(k)$, otherwise we choose an arbitrary
  subset of $\PPP$ of order exactly $f_p(k)$ and work with that.

  We first need to define some notation used throughout the proof.
  For any path $X$ we denote by $\accentset{\circ}{X}$  the interior of
  $X$, i.e.~$X$ without its
  endpoints.
  Let $X$ be a subpath of a path in $\RRR$ and let $\PPP' \subseteq
  \PPP$ be such that $X$ intersects every $P\in \PPP'$. Let $P_1\in \PPP'$
  be the first path in $\PPP'$ that $X$ intersects. Let $x_1, \dots,
  x_s$ be the vertices in $V(X)\cap V(P_1)$ in the order in which they
  appear on $X$. Let $X_i$ be the subpath of $X$ from $x_i$ to
  $x_{i+1}$.
  If $\accentset{\circ}{X}_i$ intersects at least $l$ paths in $\PPP'$ then we
  call $X_i$ \emph{a $\PPP'$-bridge of $X$ of order $l$}.  It will be
  important in the sequel that a $\PPP'$-bridge has its endpoints on
  the first path in $\PPP'$ that $X$ intersects.

  \smallskip

  We are now ready to prove the lemma. We will prove the lemma
  satisfying Condition (1). The proof for Condition (2) follows by the same argument applied after reversing the directions of all edges. 

  For every $R\in\RRR$ we construct a sequence $(\BBB_i, \PPP_i,
  \tilde R_i)_{0 \leq i \leq k}$ inductively as follows. We will maintain the
  property that $|\PPP_i|\geq b_{k-i}$.
  We set $\BBB_0 := \emptyset, \tilde
  R_0 := R$, and $\PPP_0 := \PPP$. As $|\PPP| = f_p(k) = b_k$ this satisfies
  the condition on $|\PPP_0|$.
  Now suppose $\BBB_i, \tilde R_i, \PPP_i$ have
  already been defined. If $\tilde R_i$ contains a $\PPP_i$-bridge $B$
  of order $b_{k-(i+1)}$, then we set $\tilde R_{i+1} := \accentset{\circ}{B}$, $\PPP_{i+1}
  := \{ P\in \PPP_i \sth \accentset{\circ}{B}\cap P \not=\emptyset\}$,
  and $\BBB_{i+1} := \BBB_i \cup \{ B \}$.

  If there is no such bridge the construction stops after $i(R) := i$
  steps. If $i(R) = k$, then we proceed as follows. Let $\BBB_k :=
  (B_1, \dots, B_k)$. Let  $P_1, \dots, P_k$ be the paths in $\PPP$
  containing the first vertex
  of $B_1, \dots, B_k$, respectively. By construction, for $i \geq 1$,
  $B_{i+1}$ is a subpath of $B_i$ constituting a $\PPP_i$-bridge.
Furthermore, the first vertex of $B_{i+1}$ is on the path
  $P_{i+1}$ and $P_{i+1}$ is the first path in $\PPP^i$ that $B_i$ intersects.
  Thus, each $B_i$ contains an initial subpath $A_i$ linking $P_i$ to
  $P_{i+1}$ which is internally vertex disjoint from $P_1, \dots,
  P_k$. We call $(A_1, \dots, A_{k-1}, P_1, \dots, P_k)$ the \emph{routing
  sequence of $R$ of type 1} and the sequence $(P_1, \dots, P_k)$ the
  \emph{routable paths} of the sequence.

  \begin{figure}
    \includegraphics[height=4cm]{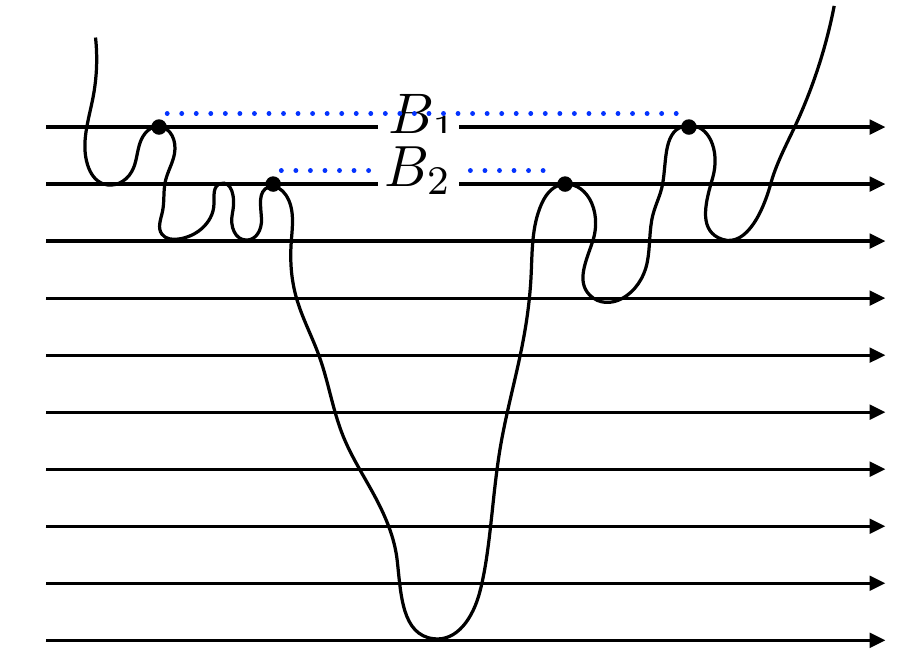}
    \includegraphics[height=4cm]{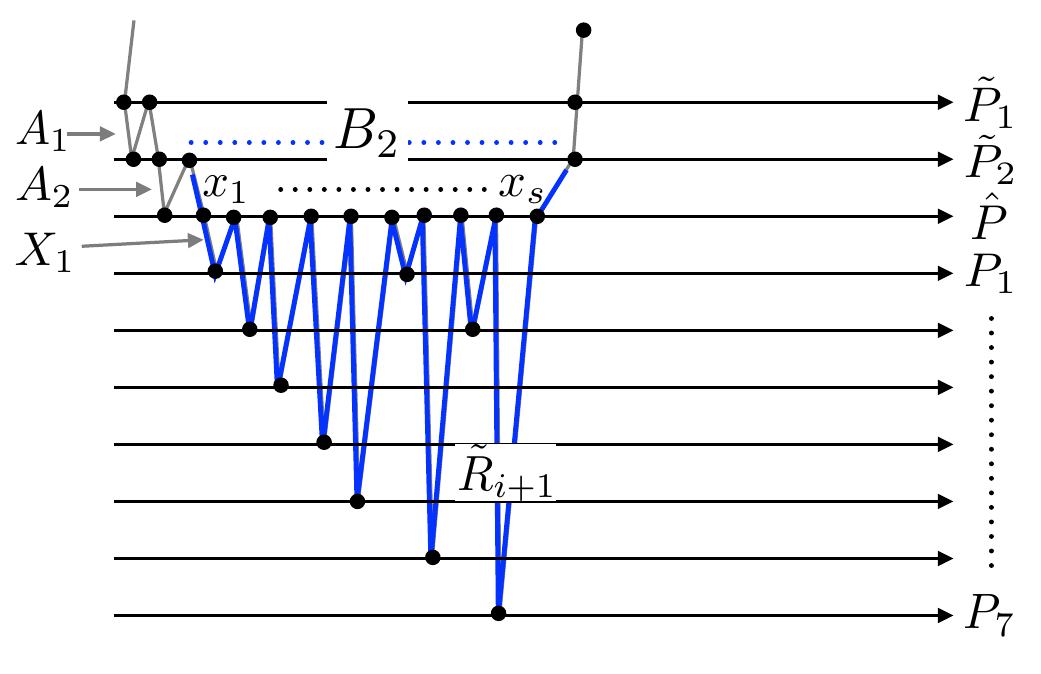}
    \caption{Construction in the first part of Lemma~\ref{lem:new-good-tuples}.}
    \label{fig:good-tuples-1}
  \end{figure}

  Now suppose the construction stops at some step $i(R) < k$. See
  Figure~\ref{fig:good-tuples-1} for an illustration of the following
  construction. Let
  $\tilde P_1, \dots, \tilde P_i$ be the paths in $\PPP$
  containing the first vertex
  of $B_1, \dots, B_i$, respectively.
    As
  before, there are paths $A_1,
  \dots, A_{i-1}$ such that $A_j$ links $\tilde P_j$ to $\tilde
  P_{j+1}$ and is internally vertex disjoint from $\tilde P_1, \dots, \tilde
  P_i$ and also from $\bigcup \PPP_i$. Furthermore, the current path $
  \tilde{R}_{i+1}$ does not intersect any of
  $\tilde P_1, \dots, \tilde P_i$.
  Let $\hat P$ be the first path in $\PPP_i$ that $\tilde{R}_{i+1}$
  intersects and let $A_i$ be the initial subpath of $\tilde{R}_{i+1}$ that connects $P_i$ and $\hat{P}$.
  Let $x_1, \dots, x_s$ be the vertices of $V(\hat P) \cap V(\tilde{R}_{i+1})$ in
  the order in which they occur on $\tilde{R}_{i+1}$. For $1 \leq j < s$,
  let $X_j$ be the subpath of $\tilde{R}_{i+1}$ from $x_j$ to 
  $x_{j+1}$ and let $\YYY_j := 
  \{ P \in \PPP_i \setminus \{ \hat P \} \sth X_j \cap P
  \not=\emptyset\}$.  As the construction above
  cannot be
  extended to $i+1$,
  $B_i$ does not contain a $\PPP_i$-bridge of order $b_{k-(i+1)}$ and
  hence $|\YYY_j| < b_{k-(i+1)}$, for all $1\leq j < s$.

  \begin{Claim}
   There are paths
    $P_1, \dots, P_k \in \PPP_i\setminus \{ \hat P \}$ and indices
    $i_1 <  \dots < i_k$ such that for all $1 \leq j \leq k$: $X_{i_j}$
    has a non-empty intersection with $P_j$ but does not intersect
    $P_{j'}$ for all $j\not= j'$.
  \end{Claim}
  \begin{ClaimProof}
    Let $b := b_{k-(i+1)}$.
    For $0 \leq i\leq b + k$ let $h_i := (2b)^{b + k - i}$.
    For $i \geq 0$ we will construct a sequence $(c_i, \III_i, \CCC_i)$ where
    \begin{itemize}
    \item $\CCC_i$ is a set of sets such that  
      $|\bigcup \CCC_i| \geq h_i$ and for each $\YYY \in \CCC_i$ there is
      $\YYY_l \in \{ \YYY_1, \dots, \YYY_{s-1}\}$ with $\YYY \subseteq \YYY_l$,
    \item $0 < c_i \leq b$ and $|\YYY| \leq c_i$ for all $\YYY\in \CCC_i$,
    \item $\III_i \subseteq \PPP$ and $\III_i \cap \YYY = \emptyset$ for all $\YYY \in \CCC_i$, and
    \item for every $P \in \III_i$ there is an $X_{j}$ such that $X_j$
      intersects $P$ but no other path in $\III_i$.
    \end{itemize}

    We initialise the construction by setting
    $\CCC_0 := \{ \YYY_l \sth \YYY_l \neq \emptyset, 1 \leq l \leq s-1 \}$,
    $\III_i = \emptyset$, and $c_0 = b$. As $|\bigcup\CCC_0| =
    |\PPP_i| \geq (2b)^{b+k}$ and, by assumption above, $|\YYY_j| < b$, this satisfies the conditions above.

    So suppose $(c_i, \III_i, \CCC_i)$ have already been
    defined. Let $l$
     be minimal such that $\YYY_l\in\CCC_i$. Let
    $P\in \YYY_l$ be the first path that is being hit by $X_l$. If
    $|\bigcup \{ \YYY' \in \CCC_i \sth P\not\in \YYY'\}\setminus
    \YYY_l| \geq h_{i+1}$
    then set $\CCC_{i+1} := \{ \YYY' \setminus
    \YYY_l\in \CCC_i \sth \YYY' \in \CCC_i, P\not\in \YYY'\}$, 
     $\III_{i+1} := \III_i \cup \{ P \}$, and
    $c_{i+1} := c_i$. Otherwise, set
    $\CCC_{i+1} := \{ \YYY' \in \CCC_i \sth P\in \YYY'\}$, 
    $\III_{i+1} := \III_i$, and $c_{i+1} := c_i-1$. As $h_i = (2b)^{b+k-i} \geq 2\cdot (2b)^{b+k-(i+1)} + b = 2\cdot h_{i+1} + b$ we have $|\bigcup
    \CCC_{i+1}|\geq h_{i+1}$. Thus, again this satisfies the conditions above.

    We claim that after $l \leq b + k$ steps we have found
    the sequence $\III_l$ of order $k$ as required. For, in every step
    we either add a path to $\III_i$ or decrease $c_i$. As $|\bigcup \CCC_i| \geq h_i$, 
    $c_i$ can never become $0$ and therefore after $<b$ steps in which $c_i$ is decreased we always add a path to $\III_i$ in the remaining $k$ steps.
  \end{ClaimProof}
  For the path $R$ as above we now choose paths $P_1,
  \dots, P_k$ as in the
  previous claim.

  We call $(A_1, \dots, A_{i-1}, \tilde P_1, \dots,
  \tilde P_i, A_i, \hat{P}, X_{i_1}, \dots, X_{i_k}, P_1, \dots, P_k)$
  a \emph{routing sequence
    of type 2} and $\tilde P_1, \dots, \tilde P_i, \hat{P}, P_1, \dots, P_k$ the \emph{routable paths}.
  This concludes the construction of routing sequences.

  \smallskip

  We now have for every path $R\in \RRR$ a routing sequence of type 1
  or 2. Suppose first there are $3k \cdot f_p(k)\cdot\binom{f_p(k)}{3k}\cdot {3k}!$ paths in $\RRR$
  with a routing sequence of type $2$. See Figure~\ref{fig:good-tuples-2}
  for an illustration of the following construction.
  Then there must be $3k$ paths
  $R_1, \dots, R_{3k}$ with the same sequence of routable paths $(\tilde
  P_1, \dots, \tilde P_l, \hat{P}, P_1,
  \dots, P_k)$. 
  
  Let 
  \[
  (A^i_1, \dots, A^i_l, A^i_{l+1}, \tilde P_1, \dots, \tilde
  P_l, \hat{P}, X^i_1, \dots, X^i_k, P_1, \dots, P_k)
  \]
  be the
  routing sequence
  of $R_i$. %
  Recall that the paths $X^i_j$ all start on $\hat{P}$. 
  We assume that $R_1, \dots, R_{3k}$ are ordered in the order
  in which they appear on the paths in $\PPP$. For all $1\leq i,j \leq
  k$ let $I^i_j$ be the initial subpath of $X^i_j$ from the first
  vertex of $X^i_j$ to the first vertex $X^i_j$ has in common with
  $P_i$ and let $S^i_j$ be the suffix of $X^i_j$ from the last vertex
  $X^i_j$ has in common with $P_i$ to the end of $X^i_j$, which again lies on $\hat{P}$.

	\begin{figure}
          \includegraphics[height=4cm]{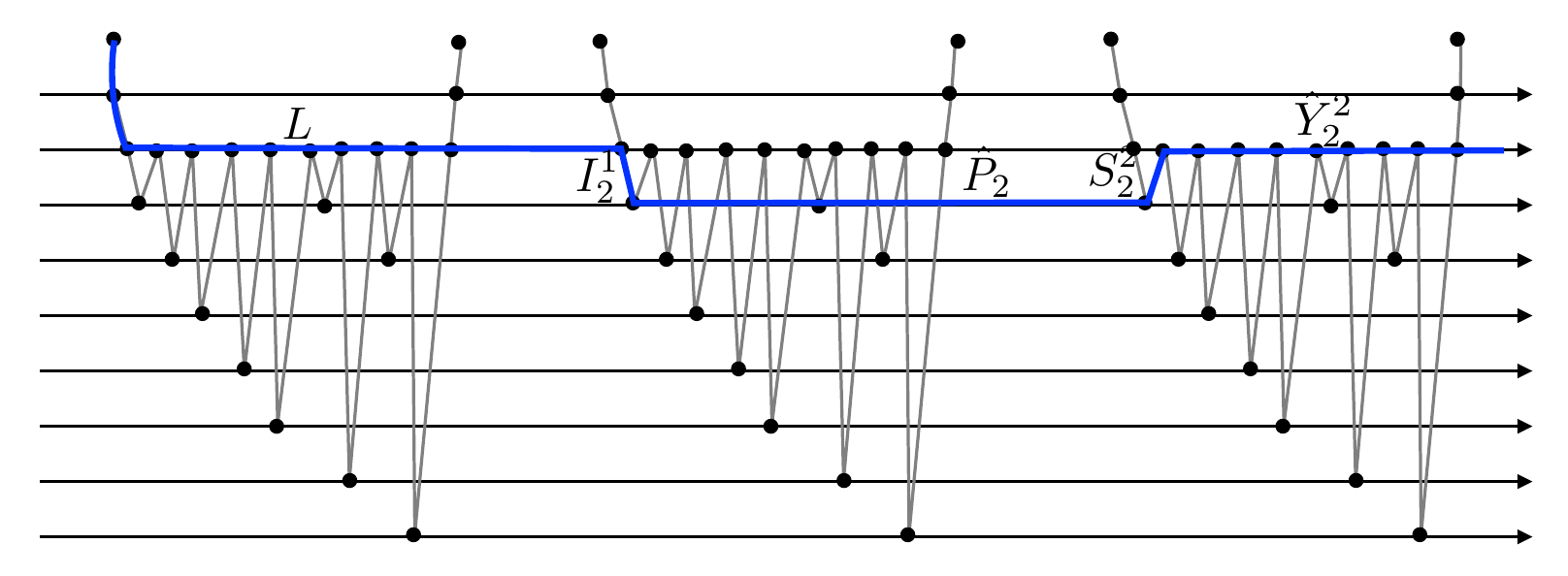}
          \caption{Constructing the path in the second part of Lemma~\ref{lem:new-good-tuples}.}
	\label{fig:good-tuples-2}
	\end{figure}

  Let $\hat P^j_i$ be the subpath of $P_i$ from the endpoint of $I^j_i$
  to the start vertex of $S^{i+1}_i$ Finally, let $\hat Y^i_j$ be the
  subpath of $\hat{P}$ from the end vertex of $S^i_j$ to the first
  vertex of $I^{i+1}_{j+1}$.

  Then,
  \[
     I^l_{1} \cdot \hat P^l_{1} \cdot S^{l+1}_{1} \cdot \hat Y^{l+1}_{1} \cdot I^{l+2}_{2}
     \cdot \hat P^{l+2}_{3} \cdot S^{l+3}_{2} \dots S^{l+2k-1}_{k}
  \]
  is a path $A'$ intersecting $P_1,
  \dots, P_k$ in this order and such that $A'\cap P_i$ is a path, for
  all $1\leq i \leq k$. What is left to do is to connect the first
  vertex of a path in $\RRR$ to the first vertex of $A'$. For this we
  use the first part $(A_1^i, \dots, A^i_l, \tilde P_1, \dots, \tilde
  P_l)$ of the routing sequence for $R_i$, $1\leq i \leq l$. By
  construction, $A^i_1, \dots, A^i_l$ are paths internally disjoint
  from $\tilde P_1, \dots, \tilde P_l, P_1, \dots, P_k$ with $A^i_j$
  linking $\tilde P_j$ to $\tilde P_{j+1}$. Furthermore, $\tilde P_1$
  is the first path in $\PPP$ hit by $R_1$. Hence, $\tilde P_1\cup
  \dots \cup \tilde P_l \cup A^1_1\cup \dots \cup A^l_l$ contains a
  subpath $L$ from the first vertex of $R_1$ to the first vertex of $A'$ not
  intersecting any of $P_1, \dots, P_k$. Thus, $A = L\cdot A'$ is a
  valid outcome of the lemma.

So now suppose there are no $2k\binom{f_p(k)}{2k}\cdot {2k}!$ paths with a
  routing sequence of type
  $2$. Hence, there are at least $(k+1)\cdot
  \binom{f_p(k)}{k}\cdot k!$ paths with a routing sequence of type $1$
  and therefore there are $k+1$ paths
  $R_0, \dots, R_k \in \RRR$ with a routing sequence $(A^i_1, \dots, A^i_{k-1}$, $P_1,
  \dots, P_k)$ of type  $1$, for $0\leq i \leq k$, such that $R_0,
  \dots, R_k$ occur in this order on the paths $P_i$.
  Let $A^0_0$ be the initial subpath of $R_0$ from the first vertex of
  $R_0$ to the first vertex on $P_1$. For $0\leq j<k$ let $\hat P_j$
  be the subpath of $P_j$ from the last vertex of $A^{j-1}_{j-1}$ to the
  first vertex of $A^j_j$. Then $A = A^0_0\cdot \hat P_1 \cdot A^1_1 \cdot \hat
  P_2 \dots A^{k-1}_{k-1}$ constitutes the required outcome of the
  lemma.
In both cases the construction implies that the path
  $A$ starts and ends on a vertex of a path in $\RRR$.
\end{proof}

Note that, in general, it is not possible to satisfy Condition (1) and
(2) of the previous lemma simultaneously.

We will use the lemma to construct a directed grid from a given web.
We first consider the case of splits.

\begin{lemma}\label{lem:split-grid}
  Let $f_p, f_r$ be the functions defined in
  Lemma~\ref{lem:new-good-tuples}. For
  all $k\geq 0, q \geq f_p(k)$ and $p\geq f_r(k){q\choose k}k!k$, if $G$
  contains a weak $(p, q)$-split $(\Ssplit, \QQQ)$, then $G$ contains a
  $(k,k)$-grid $(\PPP', \QQQ')$ with $\QQQ'\subseteq\QQQ$.
  Finally, if $(\Ssplit, \QQQ)$ is
  well-linked then so is $(\PPP', \QQQ')$.
\end{lemma}
\begin{proof}
  W.l.o.g.~we assume that $p = f_r(k){q\choose k}k!k$.
  Let $r := {q\choose k}k!k$ and let $\Ssplit = (P^1_1,\dots,
  P_1^{f_r(k)}, \dots, P_r^1, \dots, P_r^{f_r(k)})$ be ordered in the order in which the
  paths in $\QQQ$ traverse the paths in $\Ssplit$.
  By definition, for all $1\leq i \leq r, 1\leq j \leq f_r(k)$, the path $P_i^j$ intersects
  every path in $\QQQ$ and every path $Q\in \QQQ$ can be split into
  disjoint segments $Q_1, \dots, Q_p$ occurring in this order on $Q$
  such that for all $1\leq i \leq r, 1\leq j \leq f_r(k)$, the path
  $Q$ intersects $P_i^j$ only in segment $Q_{p+1 - (i-1)\cdot
    f_r(k)-j}$. For all $1\leq i \leq r$ and $Q\in \QQQ$ let $Q^i$ be
  the minimal subpath of $Q$ containing $V(Q)\cap (\bigcup_{1\leq j
    \leq f_r(k)}V(P_i^j))$ and let $\QQQ^i := \{ Q^i \sth Q\in
  \QQQ\}$.

  As $|\QQQ|\geq f_p(k)$,
  we can now apply Lemma~\ref{lem:new-good-tuples} to $(\QQQ^i,
  \{P_i^1, \dots, P_i^{f_r(k)}\})$, for all $1\leq i \leq r$, to obtain a sequence $\hat\QQQ^i :=
  (Q^i_1, \dots, Q^i_k)$ of
  paths $Q^i_l\in \QQQ^i$ and a path $A_i$ as in the statement of the
  lemma.

  As $r =  {q \choose k}k!k$, there are paths $Q_1, \dots, Q_k\in Q$
  and numbers $j_1 <
  \dots < j_{k}$ such that $Q^{j_l}_l$ is a subpath of $Q_l$, for
  all $1\leq l \leq k$. Hence,
  $(\{A_{i_1}, \dots, A_{i_k}\}, \{Q_1, \dots, Q_k\})$ is a $(k,k)$-grid.
  As we do not split any
  path in $\QQQ$, well-linkedness is preserved.
\end{proof}

We now consider the case where the result of
Lemma~\ref{lem:split-or-segment} is a segmentation.

\begin{lemma}\label{lem:segmentation-grid}
  Let $f_p, f_r$ be the functions defined in
  Lemma~\ref{lem:new-good-tuples}.
  Let $t$ be an integer and let $q \geq f_r(3t)\cdot {f_p(3t) \choose
    3t}\cdot (3t)!\cdot 4t$ and $r \geq f_p(3t)\cdot q!$.
  If $G$ contains an $(r, q)$-segmentation $(\Sseg, \QQQ)$, then $G$ contains a
  $(t,t)$-grid $W' = (\PPP', \QQQ')$ such that $\PPP'\subseteq \Sseg$.
  Furthermore, the set of start and
  end vertices of $\QQQ'$ are subsets of the start and end vertices of
  $\QQQ$. In particular, if the set of start and end vertices of $\QQQ$ is
  well-linked, then so is $W'$.

  Finally, the grid $(\PPP', \QQQ')$ can be chosen so that one (but
  not both) of the
  following properties is satisfied. Let $\PPP' := (P_1, \dots, P_t)$
  be an ordering of $\PPP'$ in order in which they occur on the paths
  $\QQQ'$ of the grid.
  \begin{enumerate}
  \item For every $Q'\in \QQQ'$, let $Q\in\QQQ$ be the path with the
    same start vertex as $Q'$. Then the first path $P\in \PPP'$ hit by $Q$
    is $P_1$.
  \item For every $Q'\in \QQQ'$, let $Q\in\QQQ$ be the path with the
    same end vertex as $Q'$. Then the last path $P\in \PPP'$ hit by $Q$
    is $P_t$.
  \end{enumerate}
\end{lemma}
\begin{proof}
  Let $\Sseg := (P_1, \dots,
  P_r)$ and $\QQQ := (Q_1, \dots, Q_q)$.

  Note that in some sense this case is symmetric to the case of
  Lemma~\ref{lem:split-grid} in that here each $P_i$ can be partitioned into segments
  $P_{i, 1}, \dots, P_{i,q}$ so that $Q_j$ intersects $P_i$ only in
  $P_{i,j}$. So in principle the same argument as in Lemma~\ref{lem:split-grid} with the role of
  $\PPP$ and $\QQQ$ exchanged applies to get a
  grid. However, in this case the paths in $\QQQ$ would be split so
  that the well-linkedness would not be preserved. We therefore need some
  extra arguments to restore well-linkedness.

  By Part (1) of Lemma~\ref{lem:strong-segmentation},
  there is a subset $\PPP'\subseteq \{ P_1, \dots, P_r\}$ of order $f_p(3t)$
   such that $(\PPP', \QQQ)$ is a strong segmentation.
   By renumbering the paths we may assume that $\PPP' = \{
   P_1, \dots, P_{f_p(3t)}\}$.

   Let $q' := \binom{f_p(3t)}{3t}\cdot (3t)!\cdot 4t$.
   Let $\QQQ := (Q_1^1, \dots, Q_1^{f_r(3t)}, \dots, Q_{q'}^1, \dots,
   Q_{q'}^{f_r(3t)})$ be ordered in the order in which the paths in
   $\QQQ$ appear on the paths in $\PPP'$.
   As $|\PPP'| = f_p(3t)$, for every $1\leq i \leq q'$, applying
   Lemma~\ref{lem:new-good-tuples} to $(\PPP', \{ Q_i^1, \dots,
   Q_i^{f_r(3t)}\})$, yields a sequence $\PPP_i\subseteq \PPP'$ of order
   $3t$ and a path $A_i$ as in the statement of
   Lemma~\ref{lem:new-good-tuples}. By choosing the paths $A_i$
   according to Property (1) or (2) in
   Lemma~\ref{lem:new-good-tuples}, we can satisfy Condition 1 or 2 of
   the present lemma. In the sequel, we present the proof in case we
   choose the first option. The case of the second option is
   completely analogous. So suppose the sets $\PPP_i$ and the paths
   $A_i$ have all been chosen according to Property (1) of Lemma~\ref{lem:new-good-tuples}.

   As $q' =  \binom{f_p(3t)}{3t}\cdot (3t)!\cdot 4t$, there are at least $4t$ values
   $i_1 < \dots < i_{4t}$ such that $\PPP_{i_j} = \PPP_{i_{j'}}$, for
   all $1\leq j \leq j' \leq 4t$.

  Let $\HHH := (H_1, \dots,
  H_{3t}) := \PPP^i$ for some (and hence all)  $1\leq j
  \leq 4t$.
  $\HHH$ will be the set
  of horizontal paths in the grid we construct, i.e.~$\HHH$ will play
  the role of $\PPP'$ in the statement of the lemma. This implies the
  condition of the lemma that $\PPP' \subseteq \Sseg$. For $1\leq j
  \leq 4t$, let $V_j := A_{i_j}$ and let $\VVV := \{ V_1, \dots, V_{4t}\}$.
  Hence $(\HHH, \VVV)$ is an
  acyclic $(3t, 4t)$-grid, but it is not yet well-linked.

  \begin{figure}
    \centering
      \includegraphics[height=5.5cm]{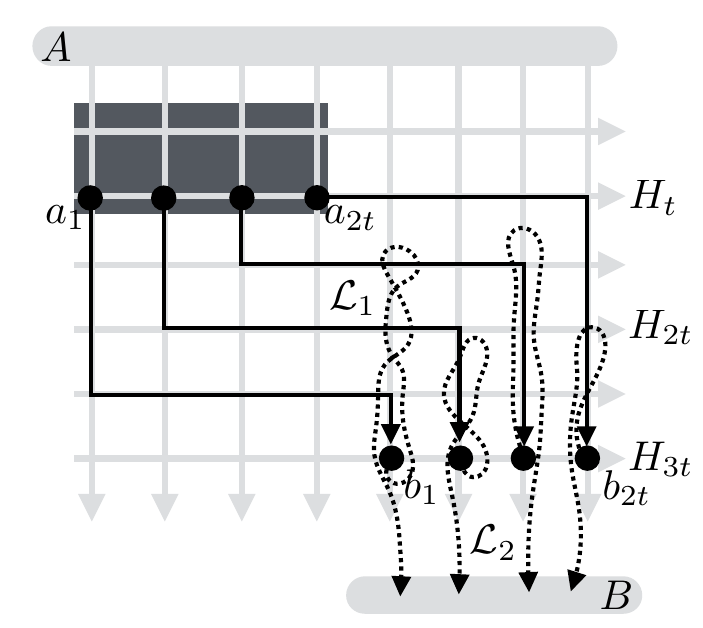}%
    \caption{Illustration of the construction in the proof of
      Lemma~\ref{lem:segmentation-grid}.}\label{fig:lem-5-17}
  \end{figure}

  In the following, let $A$ be the set of start vertices of the paths
  in $\QQQ$ and let $B$ be the set of their end vertices. As we have
  chosen the sequence $\PPP^i$ so that it satisfies Condition (1) of
  Lemma~\ref{lem:new-good-tuples}, all $V_i$ start at a vertex $v_i
  \in A$.  For $1\leq i \leq 2t$ let $a_i$ be a vertex in $V(V_i \cap H_{t+1})$
  and let $b_i$ be the end vertex of $V_{2t+i}$. Then $\HHH\cup \VVV$
  contains a linkage $\LLL_1$ from $\{a_1, \dots, a_{2t}\}$ to $\{ b_1, \dots,
  b_{2t}\}$ of order ${2t}$. See Figure~\ref{fig:lem-5-17}  for an
  illustration (where $t=2$).

   By construction, each $b_i$ is on a path $Q_i$ and hence $Q_i$ contains a
  subpath $T_i$ from $b_i$ to its endpoint in $B$. By construction,
  this path $T_i$ does not intersect any $V_1, \dots, V_{2t}$ and does
  not contain any vertex of the initial subpaths of the $H_i$ from the
  beginning to their intersection with $V_{2t}$. Hence,  $T_1, \dots,
  T_{2t}$ forms a linkage $\LLL_2$ from $b_1, \dots, b_{2t}$ to $B$.
  The linkage $\LLL_2$ is illustrated by the dotted lines in
  Figure~\ref{fig:lem-5-17}.

  By construction, $\LLL_1$ and $\LLL_2$ can be combined to form a
  half-integral linkage from $\{a_1, \dots, a_{2t}\}$ to $B$ which
  does not intersect any $V_i$, with $1\leq i\leq 2t$, at any vertex on
  $V_i$ before $a_i$ and does not intersect any $H_i$ at any vertex
  before $H_i$ intersects $V_{2t}$, i.e.~it does not intersect the
  area marked in dark grey in Figure~\ref{fig:lem-5-17}.
  By Lemma~\ref{lem:half-integral},
  there is an (integral) linkage $\LLL'$ from $\{a_1, \dots,
  a_{2t}\}$ to $B$ in the graph $\bigcup \LLL$ of order $t$. Hence,
  $\LLL'$ contains $t$ pairwise vertex disjoint paths $L_1, \dots,
  L_t$ from a subset $C
  := \{ a_{i_1}, \dots, a_{i_t}\}$ to $B$, such that $L_l$ has
  $a_{i_l}$ as the start vertex. By deleting all vertical
  paths $V_j$ with $j \not= i_l$, for all $1\leq i \leq t$, and
  joining $V_{i_j}$ and $L_j$ to form a new path $V''_{i_j}$, we
  obtain a well-linked acyclic $(t, t)$-grid $(\{ H_1, \dots, H_t\},
  \{ V''_{i_1}, \dots, V''_{i_t}\})$.
\end{proof}

We are now ready to prove Theorem~\ref{thm:grid1}.

\medskip

\begin{proof}[Proof of Theorem~\ref{thm:grid1}]
  Let $f_p, f_r$ be the functions defined in Lemma~\ref{lem:new-good-tuples}.
  Let $t, d$ be integers and let a $(p, q)$-web of avoidance $d$ in a
  digraph $G$ be given. We will
  determine the minimal value for $p$ and $q$ in the course of the proof. By
  Lemma~\ref{lem:web-avoidance}, $G$ contains a $(p_1, q_1)$-web with
  avoidance $0$ as long as
  \begin{equation}
    \label{eq:1}
    p \geq \frac{d}{d-1} p_1 \qquad \text{and} \qquad  q \geq q_1{p\choose \lceil\frac1d p\rceil}.
  \end{equation}
  As noted above, any such a $(p_1, q_1)$-web is a $(p_1, q_1)$-web with
  linkedness $c = p_1$. By
  Lemma~\ref{lem:split-or-segment}, there is a minimal value for $q_1$
  such that if we set
  \begin{equation}
    \label{eq:2}
    p_1 = p_2, \qquad x = p_2
    \qquad \text{and} \qquad y = f_p(3t)q_2!
  \end{equation}
  then $G$ contains a $(p_2, q_2)$-split or a $(f_p(3t)q_2!, q_2)$-segmentation. In the first case, if
  \begin{equation}
    q_2 \geq f_p(t) \qquad \text{and} \qquad p_2
    \geq f_q(t)\binom{q_2}{t}\cdot t!\cdot t\label{eq:3},
  \end{equation}
  then Lemma~\ref{lem:split-grid}
  implies that $G$ contains an acyclic well-linked $(t,t)$-grid. In the
  other case, if
  \begin{equation}
    \label{eq:4}
    q_2 \geq f_r(t)\binom{f_p(3t)}{3t}\cdot (3t)!\cdot 4t \qquad \text{and} \qquad
    p_2 \geq f_p(3t)q_2!
  \end{equation}
  then Lemma~\ref{lem:segmentation-grid} implies that $G$ contains an acyclic
  well-linked $(t, t)$-grid as required. Clearly, for any $t, d\geq 0$ we can always
  choose the
  numbers $p, p_1, p_2, q, q_1, q_2$, $x, y$ so that all inequalities above are
  satisfied, which concludes the proof.
\end{proof}

As noted at the beginning of this section,
the main result of this section, Theorem~\ref{theo:main-fence}, follows from Theorem~\ref{thm:grid1}.

\smallskip

So far we have shown that every digraph which
contains a large well-linked web
also contains a large well-linked fence $(\PPP, \QQQ)$.
The well-linkedness of $(\PPP, \QQQ)$ implies the existence of a
minimal bottom-up linkage as defined in the following definition.

\begin{definition}\label{def:back-linkage}
  Let $(\PPP, \QQQ)$ be a fence.
  A \emph{$(\PPP, \QQQ)$-bottom-up linkage} is a linkage $\RRR$ from
  $\Bot(\PPP, \QQQ)$ to $\Top(\PPP, \QQQ)$.
  It is called \emph{minimal $(\PPP, \QQQ)$-bottom-up linkage}, if
  $\RRR$ is $(\bigcup\PPP \cup \bigcup
    \QQQ)$-minimal.
\end{definition}

We close this section by establishing a simple
routing principle in fences which will be needed below. This is
$(3.2)$ in \cite{ReedRST96}.

\begin{lemma}\label{lem:rerouting}\showlabel{lem:rerouting}
  Let $(P_1, \dots, P_{2p}, Q_1, \dots, Q_q)$ be a $(p,q)$-fence in a
  digraph $G$, with the top $A$ and the bottom $B$. Let $A'\subseteq A$ and
  $B'\subseteq B$ with $|A'|=|B'|=r$ for some $r\leq p$. Then there
  are vertex disjoint paths $Q_1', \dots, Q_r'$ in $\bigcup_{1\leq i \leq 2p} P_i
  \cup \bigcup_{1\leq j\leq q} Q_j$ such that $(P_1, \dots, P_{2p},
  Q'_1, \dots, Q'_r)$ is a $(p,r)$-fence with top $A'$ and bottom $B'$.
\end{lemma}

We will also need the analogous statement for pseudo-fences which we prove next.
Note that, as pseudo-fences are a special form of weak pseudo-fences, the next lemma also applies to pseudo-fences.
\begin{lemma}\label{lem:routing-pseudo-fences}
 	Let $(\PPP, \QQQ)$ be a  weak $(p, q)$-pseudo-fence and let $A$ be its top 
 	and $B$ be its bottom. Let $A' \subseteq A$ and $B' \subseteq B$ be such that 	
 	$|A'| = |B'| \leq p$. Then there is an $A'{-}B'$-linkage $\LLL$ of order 
 	$\frac13|A|$ in $\PPP \cup \QQQ$.
 	 \end{lemma}
 \begin{proof}
 	Let $k := |A'|$ and let $A' := \{a_1, \dots, a_k\}$ and $B' :=
  \{ b_1, \dots, b_k\}$.
  Let $Q_1, \dots, Q_k$ be the paths in $\QQQ$ with start vertices
  $a_1, \dots, a_k$ and let $Q'_1, \dots, Q'_k$ be the paths in
  $\QQQ$ with end vertices $b_1, \dots, b_k$. For each $1\leq
  i\leq k$ let $P_i^1 := P_{2i}$ and $P_i^2 := P_{2i-1}$. This is possible as $k\leq
  p$. By definition of a weak pseudo-fence, there is a path $L_i$ connecting 
  the end of $P_i^1$ to the beginning of $P_i^2$ or vice versa. 
  Thus, for all $1\leq i\leq k$, $Q_i\cup P_i^1\cup L_i \cup P_i^2 \cup Q'_i$
  contains a path $S_i$ from $a_i$ to $b_i$. 

	Recall that for $i \neq i'$ the paths $L_i$ and $L_{i'}$ are pairwise disjoint. 
	Moreover, $L_i$ is contained in the segment $Q_{2i} \cup Q_{2i-1}$ of $Q$. 
  This implies that no vertex of $G$ is contained in more than three paths of $\SSS := \{ S_1, \dots, S_k \}$. This implies that there cannot be a $A'{-}B'$-separator of order $< \frac13 k$ and therefore, by Menger's theorem, there is an integral $A{-}B$-linkage of order $\frac13k$ as required. 
 \end{proof}

\section{From Fences to Cylindrical Grids}
\label{sec:cylindrical-grids}

So far we have seen that every digraph of sufficiently high directed
tree-width either contains a cylindrical grid or a well-linked
fence. In this section we complete the proof of our main result by
showing that if $G$ contains a well-linked fence of sufficient order,
then it contains a cylindrical grid of large order as a butterfly minor. The main
result of this section is the following theorem, which completes the proof of Theorem \ref{thm:main}.

\begin{theorem}\label{thm:main-cylindrical}
  Let $G$ be a digraph.  For every $k\geq 1$ there are integers
  $p, r\geq 1$ such that if $G$ contains a $(p, p)$-fence $\FFF$ and a
  minimal $\FFF$-bottom-up linkage $\RRR$ of order $r$ then $G$ contains a
  cylindrical grid of order $k$ as a butterfly minor.
\end{theorem}

Let $\FFF$ and $\RRR$ be as in the statement of
Theorem~\ref{thm:main-cylindrical}. We prove the theorem by analysing
how $\RRR$ intersects $\FFF$. Essentially, we follow the paths in
$\RRR$ from the bottom of $\FFF$ (i.e., the start vertices of $\RRR$) to
its top (i.e., the end vertices of $\RRR$) and somewhere along the
way we will find a cylindrical grid of large order as a butterfly
minor, either
(i) because $\RRR$ avoids a
sufficiently large sub-fence, or (ii) because it contains subpaths that
``jump'' over large fractions of the fence or (iii) because $\RRR$ and
$\QQQ$ intersect in a way that they generate a cylindrical grid
locally. We will consider the three cases separately in the following subsections.

\subsection{Bottom up linkages which avoid a sub-fence}

We first prove the easiest case, namely when $\RRR$ ``avoids'' a
sufficiently large sub-fence of $\FFF$ (i.e., Case (i) above). This is needed in the
arguments below.

\begin{definition}[sub-fence]
 Let $\FFF := (\PPP, \QQQ)$ be a fence. A
  \emph{sub-fence}\index{fence!sub-fence} of $\FFF$ is a fence $\FFF'
  := (\PPP',
  \QQQ')$ with $E(\PPP')\cup E(\QQQ')
  \subseteq E(\QQQ)\cup E(\PPP)$ such that
  $\Top(\FFF')\cup\Bot(\FFF') \subseteq V(\QQQ)$
  and there are disjoint linkages $\LLL, \LLL'$ of order $|\QQQ'|$ from $\Top(\FFF)$ to $\Top(\FFF')$ and from $\Bot(\FFF')$ to $\Bot(\FFF)$ which are internally vertex disjoint from $\PPP' \cup \QQQ'$. 
  
  A \emph{sub-grid} of a grid is defined analogously.
\end{definition}

To deal with Case (i) we first prove a technical
lemma which essentially is \cite[$(3.3)$]{ReedRST96}.

\begin{lemma}\label{lem:grid-reorder}
  For every $t$, if
  $\FFF := (\PPP, \QQQ)$ is a $(q, q)$-fence, where $q := (t-1)(2t-1)+1$,
  and $\RRR$ is an
  $\FFF$-bottom-up linkage of order $q$ such that no
  path in $\RRR$ has any internal vertex in $\PPP\cup \QQQ$, then
  $\PPP\cup\QQQ\cup\RRR$ contains a cylindrical grid of order $t$ as a butterfly minor.
\end{lemma}
\begin{proof}
  Let $(\PPP, \QQQ)$ be a $(q,q)$-fence
  and let $\RRR$ be a linkage of order $q$ from the bottom $B := \Bot(\FFF)$ to the top $A := \Top(\FFF)$ of the
  fence such that no internal vertex of
  $\RRR$ is in $\PPP\cup \QQQ$.

  Let $a_1, \dots, a_q$ be the elements of $A$ and $b_1, \dots, b_q$
  be the elements of $B$ such that $\QQQ$ links $a_i$ and $b_i$, for
  all $1\leq i\leq q$. For all $1\leq j\leq q$ let $i_j$ be such that
  $\RRR$ contains a path linking $b_j$ and $a_{i_j}$. By
  Theorem~\ref{thm:szekeres}, there is a sequence $j_1 < j_2 < \dots <
  j_t$ such that $i_{j_s} < i_{j_{s'}}$ whenever $s < s'$ or there is
  a sequence $j_1 < j_2 < \dots <
  j_{2t}$ such that $i_{j_s} > i_{j_{s'}}$ whenever $s < s'$. In either
  case, let $\RRR'$
  be the paths in $\RRR$ linking $b_{j_s}$ to $a_{i_{j_s}}$ for all
  $1\leq s\leq t$ (or $1\leq s\leq 2t$ respectively).

  In the first case, by Lemma~\ref{lem:rerouting}, there are two sets
  $\PPP', \QQQ'$ of vertex disjoint paths in $\PPP\cup\QQQ$ such that $(\PPP', \QQQ')$ is a
  $(t,t)$-fence with top $\{a_{i_{j_s}} \sth 1\leq s\leq t\}$ and
  bottom $\{b_{j_s} \sth 1\leq s\leq t\}$ and it is easily seen that
  $\QQQ'$ can be chosen so that it links $a_{i_{j_s}}$ to $b_{i_s}$.
  Hence, $(\PPP', \QQQ')$ together with $\RRR$ yields a
  cylindrical   grid of order $t$, obtained by contracting each path in
  $\RRR$ into a single edge.

  In the second case, again by Lemma~\ref{lem:rerouting}, there are
  vertex disjoint paths $\PPP'$ and
  $\QQQ'$ in $\PPP\cup\QQQ$ such that $(\PPP', \QQQ')$ is a
  $(t,2t)$-fence with  top $\{a_{i_{j_s}} \sth 1\leq s\leq 2t\}$ and
  bottom $\{b_{j_s} \sth 1\leq s\leq 2t\}$ and it is easily seen that
  $\QQQ'$ can be chosen so that it links $a_{i_{j_{2t+1-s}}}$ to
  $b_{j_s}$. Let $\QQQ' = (Q_1, \dots, Q_{2t})$ be ordered from left
  to right, i.e.~$Q_s$ links $a_{i_{j_{2t+1-s}}}$ to $b_{j_s}$. To obtain a cylindrical grid of order $t$, we take for
  each $P\in\PPP'$ the minimal subpath $P^*$ of $P$ containing all
  vertices of $V(P)\cap \bigcup_{1\leq i\leq t} V(Q_i)$. Hence, from
  each such $P\in\PPP'$ we only take the ``left half''.  Let $\PPP^* := \{
  P^* \sth P\in\PPP'\}$. Then $(\PPP^*, \{Q_1, \dots, Q_t\})$ form a
  fence of order $t$. Furthermore, for all $1\leq s\leq t$, $Q_s\cup
  R_{s} \cup Q_{2t+1-s} \cup R_{{2t+1-s}}$ constitutes a cycle $C_s$.
Here, $R_{s}$ is the path in $\RRR'$ linking $b_{i_s}$
  to $a_{i_{j_{2t+1 - s}}}$.
  Furthermore,  $C_i$ and $C_j$ are pairwise vertex disjoint whenever
  $i\not= j$. Hence, $C_1, \dots, C_t$ and $\PPP^*$ together contain a
  cylindrical grid of order $t$ as butterfly minor.
\end{proof}

The previous lemma shows that whenever we have a fence and a bottom-up
linkage $\RRR$ disjoint from the fence, this implies a cylindrical
grid of large order as a butterfly minor.
We show in the next lemma that it suffices if the linkage $\RRR$
is only disjoint
from a sufficiently large sub-fence rather than from the entire fence. This
lemma completes Case (i) above and will be applied frequently in the sequel to ensure that the
bottom-up linkage hits every part of a very large fence.

\begin{lemma}\label{lem:R-avoids-sub-fence}
  For every $p\geq 1$ let $t' := 2\big((p-1)(2p-1)+1\big)$ and $t :=
  3t'$. Let $G$ be a digraph containing a  $(t, t)$-fence $\FFF :=
  (\PPP, \QQQ)$
  and a linkage $\RRR$ of order
  $t'$ from the bottom of $\FFF$ to the top. Furthermore, let $(\PPP',
  \QQQ')$ be a
    $(t',t')$-sub-fence $\FFF'$ of $\FFF$ such that
  \begin{enumerate}%
  \item $V(P') \cap V(R')=\emptyset$ for all $P' \in \mathcal{P'}$ and
    $R' \in \mathcal{R}$,
  \item $V(Q') \cap V(R')=\emptyset$
    for all $Q' \in \mathcal{Q'}$ and $R' \in
    \mathcal{R}$ and
  \item $\FFF'$ is ``in the middle'' of the fence $\FFF$, i.e.~if $\PPP = (P_1,
    \dots, P_{2t})$ is ordered from top to bottom and $\QQQ = (Q_1, \dots,
    Q_t)$ is ordered from left to right, then 
    $(\bigcup \PPP' \cup \bigcup \QQQ') \cap (P_1 \cup \dots
    \cup P_{2t'} \cup P_{2t-2t'} \cup \dots \cup P_{2t} \cup Q_1 \cup \dots
    \cup Q_{t'} \cup Q_{t-t'} \cup \dots \cup Q_{t}) = \emptyset$.
  \end{enumerate} Then $G$ contains a cylindrical grid of order $p$ as
  a butterfly minor.
\end{lemma}
\begin{proof}
  Choose a set $A := \{ a_1,
  \dots, a_{t'}\}$ and $B := \{ b_1, \dots, b_{t'}\}$ of vertices from the
  top and the bottom of $\FFF$ such that $\RRR$ links $A$ to $B$. Further,
  let $A':= \{ a'_1, \dots, a'_{t'}\}$ and $B':= \{b'_1, \dots, b'_{t'}\}$
  be the top and the bottom of $\FFF'$ such that $\QQQ'$ contains a path
  linking $a'_i$ to $b'_i$, for all $i$. We fix a plane embedding of
  $\FFF$ and assume that the vertices $a_1, \dots, a_{t'}$ are ordered so
  that they appear from left to right on the top of $\FFF$ and likewise for
  $b_1, \dots, b_{t'}$, $a'_1, \dots, a'_{t'}$ and $b'_1, \dots, b'_{t'}$.

  As $\FFF'$ is in the middle of $\FFF$, by Lemma~\ref{lem:rerouting}, there
  is a linkage $\LLL$ in $G[\PPP\cup\QQQ]$ linking
  $B'$ to $B$. Furthermore, there is an
  $A$-$A'$-linkage $\LLL'$ of order
  $t'$. Note that $\LLL$ and $\LLL'$ are vertex
  disjoint and moreover
  no path in $\LLL\cup \LLL'$ contains a vertex from $\FFF'$ except
  from the vertices in $A'\cup B'$. Hence, the linkages $\LLL, \RRR$ and $\LLL'$ can be
  combined to form a half-integral linkage from $B'$ to $A'$. By
  Lemma~\ref{lem:half-integral}, this yields an integral linkage
  $\LLL''$ from a subset $B''\subseteq B'$ to a subset $A''\subseteq
  A'$ of order $t'':= \frac12{t'}$.
  By Lemma~\ref{lem:rerouting} there are paths $\PPP'', \QQQ''$ in
  $G[\PPP'\cup \QQQ']$ such
  that $(\PPP'', \QQQ'')$ yields a $(t'', t'')$-fence with top
  $A''$ and bottom $B''$. Now we can apply Lemma~\ref{lem:grid-reorder}
  to obtain the desired cylindrical grid of order $p$ as a butterfly minor.
\end{proof}

In the sequel, we will also need the analogous statement of
Lemma~\ref{lem:R-avoids-sub-fence} for acyclic grids instead of
fences. The only difference between the two cases is that in a
fence, when routing from the top to the bottom, we can route paths to the left
as well as to the right (as the paths in $\PPP$ hit the vertical paths
in $\QQQ$ in alternating directions) whereas in an acyclic grid we can only
route from left to right. Furthermore, we need to apply
Lemma~\ref{grid} to the middle section to obtain a fence there. Otherwise, the same proof as before
establishes the next lemma.

\begin{lemma}\label{tech00}\showlabel{tech00}
  For every $p\geq 1$ there is an integer $t'$ with the following
  properties. Let $t := 3t'$.
  Let $G$ be a
   digraph containing a  $(t, t)$-grid $W$ with bottom $(b_1, \dots, b_t)$ and top $(a_1,
   \dots, a_t)$, both ordered from left to right, and a linkage $\RRR$ of order
   $t'$  such that $\RRR$ joins the last third $(b_{2t'+1},\dots, b_t)$
    of the bottom vertices
    to the first third $(a_1,\dots,a_{t'})$ of the top vertices.
 Furthermore, let $W' := (\PPP', \QQQ')$ be a
     $(t',t')$-subgrid of $W$ such that
   \begin{enumerate}
   \item $V(P') \cap V(R')=\emptyset$ for all $P' \in \mathcal{P'}$ and $R' \in
     \mathcal{R}$,
   \item $V(Q') \cap V(R')=\emptyset$
     for all $Q' \in \mathcal{Q'}$ and $R' \in
     \mathcal{R}$,
   \item $W'$ is ``in the middle'' of the grid $W$, i.e.~if $\PPP =
     (P_1, \dots, P_{t})$ is ordered from top to bottom and $\QQQ =
     (Q_1, \dots, Q_t)$ is ordered
     from left to right, then $(\bigcup \PPP' \cup \bigcup \QQQ') \cap (P_1 \cup \dots \cup P_{t'} \cup P_{2t'+1} \cup
      \dots \cup P_{t} \cup Q_1 \cup  \dots \cup Q_{t'} \cup Q_{2t'+1}
      \cup \dots \cup  Q_{t}) = \emptyset$.
   \end{enumerate}
   Then $G$ contains a cylindrical grid of order $p$ as a butterfly minor.
 \end{lemma}
\begin{proof}
  Let $t'$ be the integer such that Lemma~\ref{grid}
  guarantees that every $(t', t')$-grid contains a fence of order
  $p' := 2\big((p-1)(2p-1)+1\big)$.
  We first apply Lemma~\ref{grid} to the subgrid $W'$ to get a $(p', p')$-fence
  $\FFF$ in $W'$ whose top and bottom
  are part of the top and bottom of
  $W'$.  Let $A' := (a'_1, \dots, a'_{p'})$ be the top and let $B' :=
  (b'_1, \dots, b'_{p'})$ be the bottom of $\FFF$, both ordered from
  left to right.

  Choose a subset $\RRR'\subseteq \RRR$ of order $p'$ and let $B$ be
  the set of start vertices of the paths in $\RRR'$ and let $A$ be
  their end vertices. Choose in $W$ a linkage $L'$ of order $p'$ from
  $B'$ to $B$ and a linkage $L''$ of order $p'$ from $A$ to $A'$. Then
  $L'\cup \RRR' \cup L''$ is a half-integral linkage of order
  $2\big((p-1)(2p-1)+1\big)$ from $B'$ to
  $A'$ which is internally disjoint from $\FFF$. By
  Lemma~\ref{lem:half-integral}, there is also an integral linkage $L$
  of order $\big((p-1)(2p-1)+1\big)$ from $B'$ to $A'$. Hence we can
  apply Lemma~~\ref{lem:grid-reorder}
  to obtain the desired cylindrical grid of order $p$ as a butterfly minor.
\end{proof}

\subsection{Taming jumps}

The previous results solve the easy cases in our argument, i.e.~where the
bottom-up linkage $\RRR$ avoids a large part of
the fence (i.e., Case (i) from the beginning of
the section). It has the
following consequence that we will use in all our arguments below. Suppose
$\FFF$ is a huge fence and $\RRR$ is an $\FFF$-bottom-up linkage. If
$\RRR$ avoids any small sub-fence, where ``small''
essentially means $4k^2$,
then this implies that $\FFF\cup\RRR$ contains a cylindrical grid of
order $k$ as a butterfly minor. Hence, for that not to happen, almost all paths of $\RRR$
must hit every small sub-fence of $\FFF$. In
particular, this
observation will be used in the next lemma to show that $\RRR$
not only must hit every small sub-fence, but it must in fact go through
$\FFF$ in a very nice way, namely going up ``row by row''.
This analysis will imply Case (ii) above.
We give a formal definition and a proof of this
statement next.

\begin{definition}\label{def:canonical-order}
    Let $\FFF := (\PPP, \QQQ)$ with $\PPP :=
    (P_1, \dots, P_{2p})$ and
    $\QQQ := (Q_1, \dots, Q_q)$ be a $(p,q)$-fence.
  For $1\leq i\leq 2p$, the \emph{$i$-th row} of $\FFF$, denoted $\row_i(\FFF)$, is
  $P_i \cup \bigcup_{1\leq j\leq q} Q_j^i$, where $Q_j^i$ is the
  subpath of $Q_j$ starting at the first vertex of $Q_j$ after the
  last vertex of $V(Q_j\cup P_{i-1})$ and ending at the last vertex of
  $V(Q_j\cap P_i)$ on $Q_j$. For $i=1$, we take
  the initial
  subpath of $Q_j$ up to the last vertex of $Q_j\cap P_1$. For convenience, for $i=2p$ we let $Q_j^i$ end at the last vertex of $Q_j$.
\end{definition}

Hence, the $i$-th row of a fence is the union of the vertical paths
between $P_{i-1}$ and $P_i$, including  $P_i$ but none of
$P_{i-1}$, so that rows are disjoint.

\begin{definition}\label{def:jump}
  Let $\FFF$ be a fence and let $\RRR$ be an $\FFF$-bottom-up linkage.
  Let $R\in \RRR$ be a
  path. For some $i>j$ with $i-j \geq 2$, a
  \emph{jump from $i$ to $j$} in $R$ is a subpath
  $J$ of $R$ which is internally vertex disjoint from $\FFF$
  such that the start vertex $u$ of $J$ is in row $i$ of $\FFF$ and its
  end vertex $v$
  is in row $j$. The \emph{length} of the
  jump $J$ is $i-j$.
\end{definition}

Note that if $J$ is a jump linking $u$ and $v$,
then there is no path from $u$ to $v$ in $\FFF$.

\begin{lemma}\label{lem:jumps}
  For every $t, t'\geq 1$ there are integers $p,q,r\geq 1$ such that
  if $\FFF := (\PPP, \QQQ)$ is a $(p,q)$-fence and $\RRR$ is an $\FFF$-minimal
  $\FFF$-bottom-up-linkage of order $r$ then either
  \begin{enumerate}
  \item $\PPP\cup\QQQ\cup \RRR$ contains a cylindrical grid of order
    $t$ as a butterfly minor or
  \item a sub-fence $(\PPP', \QQQ')$ of $(\PPP, \QQQ)$ of order
    $t'$ and a $(\PPP', \QQQ')$-minimal $(\PPP', \QQQ')$-bottom-up-linkage $\RRR'$%
     of order
    $t'$ such that $\RRR'$ goes up row by row, i.e.~for every $R\in
    \RRR'$, the last vertex of $R$ in the
    $i$-th row of $(\PPP',
    \QQQ')$ occurs on $R$ before the first vertex of the $j$-th row for
    all $j<i-1$.
  \end{enumerate}
\end{lemma}
\begin{proof}
  Let $t'' := 6((t-1)(2t-1)+1)$, $d_{2t''} :=  2t''$ and, for $0\leq l <
  2t''$, let $d_i := 20t'^2\cdot d_{i+1}$. We define $p := 2d_0$ and $q := d_0$.

  We prove the lemma by eliminating jumps in $\RRR$. In the first
  step, we eliminate jumps that jump over a
  large part of the fence.

  \smallskip

  \textit{Step 1. Taming long jumps. }
  For $0\leq l\leq 2t''$, we inductively construct a sequence
  of sub-fences $\FFF_l$ of $\FFF$ of order at least $d_l$, $\FFF_l$-minimal
  $\FFF_l$-bottom-up-linkages
  $\RRR_l$ of order $d_l$ and
   sets $J_l$ of jumps such that $|J_l|=l$ and each jump in $J_l$
  is disjoint from $\FFF_l$ and intersects $\FFF_{l-1}$ only at its endpoints.   Note that $J_l$ is increasing in size while both $\FFF_l$ and $\RRR_l$ are decreasing.

We set
  $J_0 := \emptyset$, $\FFF_0 := \FFF$ and $\RRR_0 := \RRR$, which satisfy the case $l=0$.  Now suppose $\FFF_l, \RRR_l$ and
  $J_l$ have already been constructed, for some $l\geq 0$.

  If there is a path $R\in \RRR_l$ which contains a jump $J$ in $\FFF_l$
  of length at least $5d_{l+1}$, then let $i_l$ be
  the row containing its start vertex $u_l$ and $j_l$ be the row containing its
  end $v_l$.   We set $J_{l+1} :=
  J_l \cup \{ J \}$.
  By construction, $i_l-j_l \geq
  5d_{l+1}$. Hence, between row $i_l$ and $j_l$ there is a
  $(5d_{l+1}, d_l)$-fence $\FFF'_{l+1}$. We choose a $(d_{l+1},
  2d_{l+1})$ sub-fence $\FFF''_{l+1}$ of $\FFF'_{l+1}$ between row $i_l-2d_{l+1}$ and $j_l+2d_{l+1}$ which does not
  contain the at most four vertical paths containing the endpoints of
  $J$ and $R$. This is possible as $d_i >5 d_{i+1}+4$. 
  
  	Using Lemma~\ref{lem:rerouting}, we construct in $\FFF_l\cup \RRR_l$ a half-integral
  linkage of order $2d_{l+1}$ from the bottom of $\FFF''_{l+1}$ to its
  top as follows: we choose a subset $\RRR'_l$ of $\RRR_l$ of order $2d_{l+1}$, a linkage $L$ of order $2d_{l+1}$ from the bottom of $\FFF''_{l+1}$ to the set of start vertices of $\RRR'_l$, and a linkage  $L'$ from the set of end vertices of the paths in $\RRR'_l$ to the top of $\FFF''_{l+1}$ and then concatenate the paths in $L, \RRR'_l$ and $L'$. 
	The linkages $L, L'$ exist by Lemma~\ref{lem:rerouting}. Therefore, by Lemma~\ref{lem:half-integral}, there also is an
  integral linkage $\RRR'_{l+1}$ of order $d_{l+1}$ from the bottom of $\FFF''_{l+1}$ to
  its top.  Let $\FFF_{l+1}$ be a $(d_{l+1}, d_{l+1})$-sub-fence of
  $\FFF''_{l+1}$ whose top and bottom are the endpoints of
  $\RRR'_{l+1}$, which exists by Lemma~\ref{lem:rerouting}. Let $\RRR_{l+1}$ be an $\FFF_{l+1}$-minimal
  linkage with the same endpoints as $\RRR'_{l+1}$.
  This completes the
  construction of $J_l, \FFF_l, \RRR_l$ in case a path $R \in \RRR_l$ as above exists.

  Otherwise, i.e.~if there is no such $R\in \RRR_l$, the construction
  stops here.

  Suppose first that we have constructed $J_l, \FFF_l, \RRR_l$ for all
  $l\leq 2t''$. Then we have found a sub-fence
  $\FFF_{2t''}$ of order $d_{2t''} = 2t''$ contained in
  rows $j$ to $i$ of
$\FFF$, for some $j < i$,  and a set $\JJJ$ of
$2t''$ jumps $J_1, \dots, J_{2t''}$ such that the
jump $J_l$ that starts at $u_l$ and ends in $v_l$
satisfies the following:
  \begin{itemize}
  \item $i_l > i+2t''$ and $j_l < j-2t''$, for all $1\leq l \leq 2t''$, and
  \item $i_l \geq i_{l+1}+2t''$ and $j_l \leq j_{l+1}-2t''$, for all $1\leq l < 2t''$.
  \end{itemize}
  Recall that each $J_i$ starts at $u_i$ and ends in $v_i$.
  Let $s_1, \dots, s_{2t''}$ be $2t''$ vertices in $\Top(\FFF_{2t''})$ ordered from
  left to right and let
  $s'_1, \dots, s'_{2t''}$ be $2t''$ vertices in $\Bot(\FFF_{2t''})$ such that
  $s_r$ and $s'_r$ are on the same path $Q_r\in\QQQ$, for all $r$.
    It is easily seen that there is a linkage
  $\LLL_1$ with $E(\LLL_1)\subseteq
  E(\FFF)$ of
  order $2t''$ linking $\{v_1, \dots, v_{2t''}\}$ to $\{s_1, \dots, s_{2t''}\}$ and a linkage
  $\LLL_2$ with $E(\LLL_2)\subseteq
  E(\FFF)$ linking $\{s'_1, \dots, s'_{2t''}\}$ to $\{u_1, \dots, u_{2t''}\}$. Obviously, $V(\LLL_1)
  \cap V(\LLL_2) = \emptyset$ and they are internally disjoint from $\FFF_{t''}$. Hence, $\LLL := \LLL_1\cup \LLL_2 \cup \JJJ$
  is a half-integral linkage of order $2t''$ from $\{s'_1, \dots,
  s'_{2t''}\}$ to $\{ s_1,
  \dots, s_{2t''} \}$. By Lemma~\ref{lem:half-integral}, $\LLL$ contains an
  integral linkage $\LLL'$ of order $t''$ from $\{s'_1, \dots,
  s'_{2t''}\}$ to $\{ s_1, \dots, s_{2t''} \}$.
  By Lemma~\ref{lem:R-avoids-sub-fence}, $\FFF_{2t''}\cup \LLL'$ contains a
  cylindrical grid of order $t$ as a butterfly minor, which is the first outcome of the lemma.

  \medskip

  \textit{Step 2. Taming short jumps. } So now suppose that the
  construction stops after $l < 2t''$ steps.
  Hence, we now have a sub-fence $\FFF' := \FFF_l = (\PPP', \QQQ')$ of
  order $d := d_l$ and we can choose an $\FFF'$-minimal $\FFF'$-bottom-up
  linkage $\RRR' \subseteq \RRR_{l}\cup \FFF_l$ of order $t'$. Note that any jump in $\RRR'$ must come from a jump in $\RRR_l$. Hence, as $\RRR_l$ does not contain any jump of length $5d_{l+1}$, $\RRR'$ cannot contain any such long jumps.

  Let $\PPP' = (P_1, \dots, P_{2d_{l+1}})$
  be the ordering of $\PPP'$ ordered from top to bottom.
  Let $\PPP'' := \{ P_{1+10t'd_{l+1}\cdot i + (i\mod
  2)} \in \PPP' \sth 0\leq i<
  2t'\}$. 
  This is well-defined as $d \geq 20t'^2d_{l+1}$.
  Note that $\PPP''$ contains paths in alternating directions as we
  have added $(i\mod 2)$ in each step.

  Let
  $\QQQ''$ be a linkage in $\FFF'$ of order $t'$
  linking the
  set of end vertices of $\RRR'$ to the set of start vertices of $\RRR'$. 
  Such a linkage exists by Lemma~\ref{lem:rerouting}.
  Then $\FFF'' := (\PPP'', \QQQ'')$ is a sub-fence of $\FFF'$ of
  order $t'$.

  We claim that $\RRR'$ traverses $\FFF''$ row by row.
  Towards a contradiction, suppose there are $i, j$ such that $j<i-1$
  and $\RRR'$ contains a path $R$ which hits a vertex in row $j$ of
  $\FFF''$ before the last vertex $R$ has in common with row $i$.
  Note that the distance between $i$ and $j$ in $\FFF'$ is at least
  $(i-j)\cdot 10t'd_{l+1}$.
  
  We group the rows between row $i$ and row $j$ into groups $H_s$ of $5d_{l+1}$ rows each. For $1 \leq s \leq 2t'$ let $H_i$ be the union of the rows $h$ of $\FFF'$ for $i-5d_{l+1} \cdot s \leq h < i-5d_{l+1} \cdot (s-1)$. 

  By construction,
  $\RRR'$ does not contain any
  jump of length $5d_{l+1}$. Hence, as $R$ goes from the bottom of $\FFF'$ to
  its top, $R$ must intersect at least one row of each $H_s$, $1 \leq s \leq 2t'$, before the first vertex it has in common with row $j$. 
  
Furthermore, $R$ continues
  from the last vertex $v$ it has in common with row $i$ to the top of $\FFF'$. Hence,
  after $v$ $R$ must again intersect at least one row in each group $H_s$, for $1 \leq s \leq 2t'$. 

	Now let $e$ be the first edge of $R$
  after $v$ that is not in $E(\FFF')$ (this edge must exist as $R$
  goes up) and let $R_1, R_2$ be the two components of $R-e$ with $R_1$
  being the initial subpath of $R$. Then for each $1 \leq s \leq t'$, there is a path $L_s$ from $R_1$ to $R_2$ in $H_{2s} \cup H_{2s-1}$. By construction, if $s \neq s'$ then $L_s \cap L_{s'} = \emptyset$. Thus there are $t'$ vertex
  disjoint paths in $\RRR'\cup\FFF'$ from $R_1$ to $R_2$
  which, by Lemma~\ref{lem:no-forward-paths}, contradicts the assumption that $\RRR'$ is $\FFF'$-minimal.
\end{proof}

\subsection{Avoiding a pseudo-fence}

In this subsection we prove two lemmas needed later on in the
proof. Essentially, the two lemmas deal with the case that we have a
fence $\FFF = (\PPP, \QQQ)$ and a bottom-up linkage which avoids the
paths in $\QQQ$. This is proved in Lemma~\ref{ll2} below. We also need
a variant of it where instead of a fence we only have a
\emph{pseudo-fence}. Recall Definition~\ref{def:pseudo-fence} of a (weak) pseudo-fence. 

\begin{lemma}
  \label{lem:ll2-general}
  For every $p\geq 1$ there are  integers $t', t''$ such that if $G$ is a
  digraph containing a  $(t'', t')$-pseudo-fence $W = (\PPP,
  \QQQ)$ and a
  linkage $\RRR$ of order $t'$ from the bottom of
  $W$ to
  the top of $W$ such that no internal vertex of any path in $\RRR$ is
  contained in $V(\QQQ)$, then
  $G$ contains a cylindrical grid of order $p$ as a butterfly minor.

  The same is true if $\WWW$ is only a weak pseudo-fence.
\end{lemma}
\begin{proof}
  \setcounter{claimcounter}{0}
Let $t'' := 2t' + t''_1\cdot
  \binom{t'}{\lceil\frac{t'}2\rceil}$.
  Let $\PPP := (P_1, \dots, P_{2t''})$ be ordered from top to bottom,
  i.e.~in the order in which the paths in $\PPP$ appear on the paths
  in $\QQQ := (Q_1, \dots, Q_{t'})$. We will
  state various conditions on $t'$ and $t''$ during the
  proof which will give us the necessary bound on
  $t', t''$.

  Let $U$ be the subgraph of $W$ containing $P_1, \dots,
  P_{2t'}$ and for each $Q\in \{Q_1, \dots, Q_{t'}\}$ the 
	  minimal initial segment of $Q$ containing all
  vertices of
  \[
     V(Q)\cap \bigcup_{1\leq i\leq 2t'} V(P_i).
  \]

  Analogously, let $D$ be the lower part of $W$, i.e. the part formed
  by $P_{2t''-2t'}, \dots,$ $P_{2t''}$ and the 
  minimal final segments of the $Q_i$
  containing all vertices $Q_i$ has in common with $P_{2t''-2t'}, \dots, P_{2t''}$.
  Finally, let $M$ be the middle part, i.e.~the subgraph of $W$
  induced by the paths $P_{2t'+1}, \dots, P_{2t''-2t'-1}$ and the
  subpaths of the $Q_i$ connecting $U$ to $D$.
  See Figure~\ref{fig:lem1.2-general} a) for an
  illustration of the construction in this proof.

  We will write $\PPP_U, \PPP'_M, \PPP_D$ for the paths in $\PPP$
  contained in $U, M, D$, respectively. Similarly, we write $\QQQ_U, \QQQ_M, \QQQ_D$ for the subpaths of the paths in $\QQQ$ connecting the paths in $\PPP_U, \PPP'_M, \PPP_D$, resp. 
	 
  \begin{figure}
    \centering
   \hspace*{-0.7cm} \begin{tabular}{ccc}
    \includegraphics[height=6.0cm]{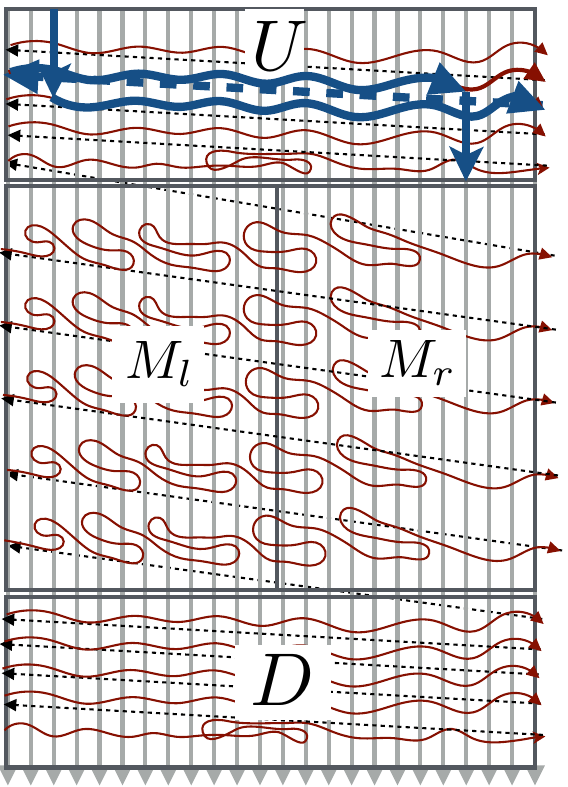}&
    \includegraphics[height=6.0cm]{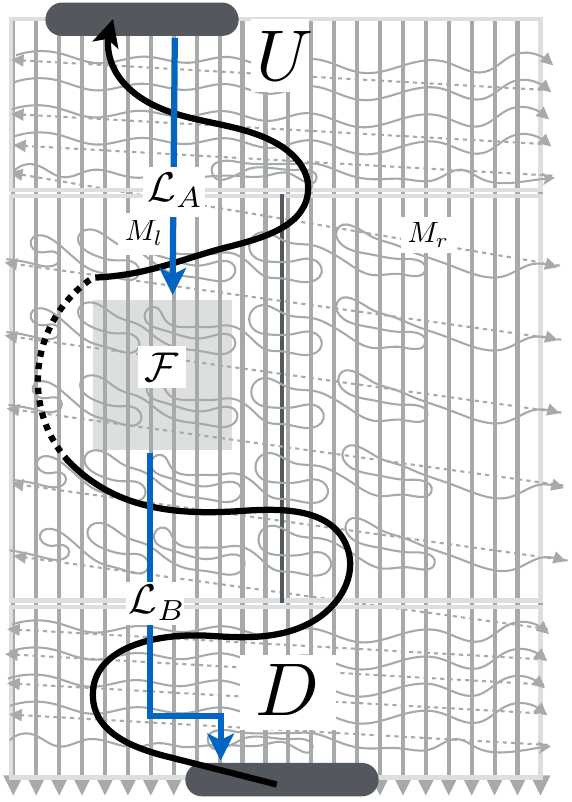}&
    \includegraphics[height=6.0cm]{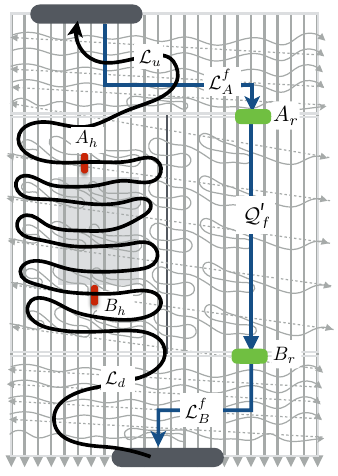}\\
    a) The various zones & b) Situation in Claim~\ref{claim:ll2-general:2} & c)
    The rest of the proof.
  \end{tabular}
    \caption{Illustration for the Proof of Lemma~\ref{lem:ll2-general}.}
    \label{fig:lem1.2-general}
  \end{figure}

  We now divide $M$ into two parts, $M_l$ and $M_r$. For every
  $P\in \PPP'_M$ let $\QQQ(P)\subseteq \QQQ$ be the set of the first
  $\frac{t'}2$ paths in $\QQQ$ that $P$ intersects. By the pigeon
  hole principle, as $|\PPP'_M| \geq
  2t''_2\binom{t'}{\lceil \frac{t'}2\rceil }$ there is a
  subset $\PPP_M \subseteq \PPP'_M$ of order $2t''_2$
  such that $\QQQ(P) = \QQQ(P')$ for
  all $P, P'\in \PPP'_M$.  We define $M_l$ as $\QQQ(P)$ for some (and hence all) $P\in \PPP_M$ together with the minimal initial subpaths of
  the $P\in \PPP'_M$ containing all vertices of $\QQQ(P)$. $M_r$ contains the other paths of $\QQQ_M$
  and the parts of the paths of $\PPP'_M$ not contained in $M_l$.
  We write $\PPP_{M_l}, \PPP_{M_r}$ and $\QQQ_{M_l}, \QQQ_{M_r}$ for
  the corresponding paths.

  To simplify the presentation we rename the paths in $\PPP_U,
  \PPP_M, \PPP_D$ such that $\PPP_U := ( P^u_1, \dots,
  P^u_{2t'})$, $\PPP_M := ( P^m_{1}, \dots,
  P^m_{2t''})$ and $\PPP_D := ( P^d_{1}, \dots,
  P^d_{2t'})$.

Observe that $(\QQQ_{M_l}, \PPP_{M_l})$ forms a strong $(\frac12t', 2t'')$-segmentation.  
Let $f_r, f_p$ be the functions defined in Lemma~\ref{lem:new-good-tuples}. 
We require that $t' \geq f_p(t'_1)$ and $t'' \geq t''_1 \cdot f_r(t'_1)$, for some $t'_1, t''_1$ determined below. For each $1 \leq j \leq t''_1$ we apply Lemma~\ref{lem:new-good-tuples} to each of the strong $(f_r(t'_1), f_p(t'_1))$-segmentations $(\QQQ_{M_l},  \{ P^m_{(j-1)\cdot f_p(t'_1) + 1}, \dots, P^m_{j\cdot f_p(t'_1)}\})$ to obtain for each $j$ a sequence $\hat{\QQQ}^j := (\hat{Q}^j_1, \dots, \hat{Q}^j_{t'_1})$ of paths in $\QQQ_{M_l}$ and a path $L_j$ as in the statement of the lemma. In particular, $L_j$  intersects each $\hat{Q}^j_i$ in a path and the paths $\hat{Q}^j_1, \dots, \hat{Q}^j_{t'_1}$ occur in this order on $L_j$.

Now, as we require $t_1'' \geq t''_2\cdot {\frac12t' \choose t'_1} \cdot t'_1!$, there  are indices $1 \leq j_1 < j_2 \dots < j_{t''_2} \leq t''_1$ such that $\hat{\QQQ}^{j_i} = \hat{\QQQ}^{j_{i'}}$ for all $1 \leq i < i' \leq t''_2$. 
W.l.o.g.~we may assume that $j_i = i$, for all $1 \leq i \leq t''_2$, and $\hat{\QQQ}^{j_i} = (Q_1, \dots, Q_{t'_1})$. Let $\LLL := (L_1, \dots, L_{t''_2})$ and $\hat{\QQQ} := (Q_1, \dots, Q_{t'_1})$. Thus, $(\LLL, \hat{\QQQ})$ form an acyclic $(t''_2, t'_1)$-grid $\HHH$. 

We require that $t''_2, t'_1$ are large enough so that when we apply 
Lemma~\ref{grid} to $\HHH$ we obtain a $(t''_3, t'_2)$-fence $\HHH'$ whose top and bottom is part of the top and bottom of $\HHH$, for some numbers $t''_3, t'_2$ to be determined below. 
 
  Let $A$ be
  the end vertices of $\RRR$ and $B$ be the start vertices of $\RRR$. Hence,
  $A$ is contained in the top of the pseudo-fence $W$ and $B$ is part of
  its bottom.
  We now take a new linkage $\RRR_1 \subseteq G[V(\RRR) \cup V(\PPP_U) \cup V(\PPP_D) \cup (V(\QQQ)\setminus
  V(\QQQ_{M_r} \cup \QQQ_{M_l}) \cup \HHH']$ from $B$ to $A$ of order $|\RRR|$
  such that $\RRR'$ is $\HHH'$-minimal. Since no internal vertex of any path in $\RRR$ is contained in $V(\QQQ)$,
  such a choice is possible. Note that $\RRR_1\cap \QQQ_{M_r} =
  \emptyset$.

We now apply Lemma~\ref{lem:jumps} to $\HHH'$ and $\RRR_1$. For this, we require that $t', t''_3, t'_2$ are big enough so that if we apply Lemma~\ref{lem:jumps} to $\HHH'$ and $\RRR_1$ then we either get a cylindrical grid of order $p$ or a sub-fence $\HHH''$ of $\HHH'$ of order $r_2$ and a $\HHH''$-minimal bottom-up linkage $\RRR_2 \subseteq \RRR_1$ 
of order $r_2$ which goes up $\HHH''$ row-by-row as in the statement of the lemma. 

In the first case, i.e.~if we get a cylindrical grid of order $p$ as outcome, we are done. 
So we may assume that the second case applies and hence we get a sub-fence $\HHH'' \subseteq \HHH'$ and a bottom-up linkage $\RRR_2 \subseteq \RRR_1 \cup \HHH'$ of order $r_2$ which goes up $\HHH''$ row-by-row.
Let $\HHH_2$ be a sub-fence of $\HHH''$ of order $t_4$ for some number $t_4$ to be determined below which consists of $t_4$ vertical paths of $\HHH''$ and $2t_4$ horizontal paths. 

Let $H_1, \dots, H_{2t_4}$ be the horizontal paths of $\HHH_2$ and $V_1, \dots, V_{t_4}$ be the vertical paths. 
Let $w' = 2\big((p-1)(2p-1)+1\big)$ and $w := 18w'$. We require that $t_4 \geq 2\cdot t_5\cdot w$ for some number $t_5$ determined below.
For each $1 \leq j \leq 2t_5$ let $\HHH^j$ be the sub-fence of $\HHH''$ comprising the paths $H_{(j-1)\cdot w+1}, \dots, H_{j\cdot w}$ and the minimal subpaths of the $V_i$ connecting $H_{(j-1)\cdot w+1}$ and $H_{j\cdot w}$.

\begin{Claim}\label{claim:ll2-general:2}
 There is a number $r_6$ such that if  there is an index $1 \leq j \leq 2t_5$ and a sub-linkage $\RRR_3 \subseteq \RRR_2$ of order $r_6$ such that $\RRR_3$ is disjoint from $\HHH^j$ then $G$ contains a cylindrical grid of order $p$ as a butterfly minor.
\end{Claim}
\begin{ClaimProof}
  See Figure~\ref{fig:lem1.2-general} b)   for an illustration of this step of the proof. 
  Let $A'$ be the top of $\HHH^j$ and $B'$ be its bottom. Recall that by construction the vertices in $A'$ are on distinct paths of $\QQQ_l$ and so are the vertices of $B'$.
  Let $B''$ be the start vertices of $\RRR_3$. Thus taking the paths in $\QQQ_l$ down from $B'$ to $D$ and applying Lemma~\ref{lem:routing-pseudo-fences}, yields a linkage $\LLL_B$ of order 
$\frac13w$ from $B'$ to $B''$. 
Let $B'''$ be the end vertices of the linkage $\LLL_B$ and let $A'''$ be the end vertices of the paths in $\RRR_3$ starting at $B'''$. 
  Again, by applying Lemma~\ref{lem:routing-pseudo-fences}, we obtain a linkage $\LLL_A$ from $A'''$ to $A'$ of order 
$\frac19w$. Thus, $\LLL_A \cup \RRR_3 \cup \LLL_B$ contains a half-integral linkage $\LLL$ from $B'$ to $A'$ of order $\frac1{18}w$ which, by construction, is disjoint from $\HHH^j$. Let $A_l \subseteq A'$ be the end vertices of $\LLL$ and $B_l$ be its start vertices. 

  We now apply Lemma~\ref{lem:rerouting} to $\HHH_j$ to get a sub-fence $\HHH_j'$ of $\HHH_j$ with top $A_l$ and bottom $B_l$ of order $\frac1{18}w = w' \geq (p-1)(2p-1)+1$. Applying Lemma~\ref{lem:grid-reorder} to $\HHH_j'$ and $\LLL$ yields a cylindrical grid of order $p$ as required. 
\end{ClaimProof}

By Claim~\ref{claim:ll2-general:2} we may now assume that for each block $\HHH_j$ there are fewer than $r_6$ paths in $\RRR_2$ with an empty intersection with $\HHH^j$. 

Now we take the blocks $\HHH_{2j}$ for $1 \leq j \leq t_5$. Thus there is a linkage $\RRR_3 \subseteq \RRR_2$ of order $r_3 := r_2 - t_5\cdot r_6$ such that every path in $\RRR_3$ intersects in each $\HHH_{2j}$, $1 \leq j \leq t_5$, at least one path. 
We require that $r_3 \geq r_4 \cdot (t_4+2w)^{t_5}$ (observe that $t_4 + 2w$ is the total number of $t_4$ vertical and $2w$ horizontal paths in any $\HHH^j$). By the pigeon hole principle, we can thus take a linkage $\RRR_4 \subseteq \RRR_3$ of order $r_4$ such that every path $R \in \RRR_4$ intersects in each block $\HHH^j$ the same path $I^j$.

As $\RRR_2$ goes up $\HHH_2$ row by row, $(\{I^j \sth 1 \leq j \leq t_5\}, \RRR_4)$ is a weak  $(t_5, r_4)$-split.
We require that $t_5$ and $r_4$ are large enough so that applying Lemma~\ref{lem:split-grid} yields an $(p_5, p_5)$-grid $\HHH_3 = (\III, \RRR_5)$ with $\RRR_5 \subseteq \RRR_4$. Furthermore, we require that $p_5$ is large enough so that we can apply Lemma~\ref{grid} to get a $(p_6, p_6)$-fence $\FFF = (\III', \RRR_6)$ whose top and bottom is contained in the top and bottom of $\HHH_3$. 
Note that the fence $\HHH_3$ "goes upwards" with respect to the original pseudo-fence $W$, i.e.~the top $B_h$ of $\HHH_3$ occurs on the paths in $\RRR_5$ before its bottom $A_h$. 
 
 Recall that $A$ is the set of end vertices of $\RRR$ and $B$ the set of start vertices. Thus $\RRR_5$ contains a linkage $\LLL_u$ from $A_h$ to $A$ of order $p_6$. Let $A_f \subseteq A$ be the set of end vertices of $\LLL_u$. Similarly, $\RRR_5$ contains a linkage $\LLL_d$ from $B$ to $B_h$ of order $p_6$. Let $B_f$ the set of start vertices of $\LLL_d$.  

We choose a linkage $\QQQ_f \subseteq \QQQ_r$ of order $p_6$ in $\QQQ_r$. Let $A_r$ be the start vertices of $\QQQ_f$. See Figure~\ref{fig:lem1.2-general} c) for an illustration of this step.%

By Lemma~\ref{lem:routing-pseudo-fences} there is a linkage $\LLL^f_A$ from $A_f$ to $A_r$ of order $\frac13p_6$. We require $p_6 \geq 18p$. Let $A'_r \subseteq A_r$ be the set of end vertices of $\LLL^f_A$, let $\QQQ'_f \subseteq \QQQ_f$ be the set of paths with start vertices in $A'_r$, and let $B_r$ be the end vertices of $\QQQ'_f$. Again by Lemma~\ref{lem:routing-pseudo-fences}, there is a linkage $\LLL^f_B$ from $B_r$ to $B_f$ of order $\frac19p_6$. 

Thus $\LLL_u \cup \LLL^f_A \cup \QQQ'_f \cup \LLL^f_B \cup \LLL_d$  contains a half-integral linkage from $A_h$ to $B_h$ and therefore an integral linkage $\LLL^f$ from $A_h$ to $B_h$ of order $\frac1{18}p_6$.

Finally, we require that $p_6$ is large enough so that we can apply Lemma~\ref{lem:R-avoids-sub-fence} to $\FFF$ and $\FFF^f$ to get a cylindrical grid of order $p$ as required. 
\end{proof}

We now prove the result mentioned above that if we have a fence with a
bottom-up linkage avoiding the vertical paths then we also have a
cylindrical grid of large order as a butterfly minor.
Towards this aim, observe that if $(\PPP, \QQQ)$ is a $(p, q+1)$-fence, with $\PPP := (P^1_1, P^2_1, \dots, P^1_p, P^2_p)$ ordered from top to bottom and $\QQQ := (Q_1, \dots, Q_{q+1})$ ordered from left to right,  then 
for each $1 \leq i \leq p$ the path $Q_{q+1}$ contains an edge from the last vertex of $P^1_i$ to the start vertex of $P^2_i$. I.e. $(\PPP, \QQQ)$ contains a $(p, q)$-pseudo-fence as subgraph. The next lemma therefore follows immediately from the previous lemma. 

\begin{lemma}
\label{ll2}\showlabel{ll2}
 For every $p\geq 1$ there is an integer $t'$ such that if $G$ is a
   digraph containing a  $(t, t)$-fence $W = (\PPP, \QQQ)$, for some
   $t\geq 3\cdot t'$,  and a linkage $\RRR$ of order
   $t'$ from bottom of $W$ to top of $W$ such that no path in $\RRR$
   contains any vertex of $V(\QQQ)$, then
   $G$ contains a cylindrical grid of order $p$ as a butterfly minor.
\end{lemma}

\subsection{Constructing a Cylindrical Grid}
\label{sec:critical}

In this section we complete the proof of our main result,
Theorem~\ref{thm:main-bramble}, and thus also of
Theorem~\ref{thm:main}. 

The starting point is
Theorem~\ref{theo:main-fence},
i.e.~we assume that there are linkages $\PPP$ of order $6p$ and
$\QQQ$ of order $q$ forming a well-linked fence.
Let $\FFF := (\PPP, \QQQ)$ with $\PPP := (P_1, \dots, P_{6p})$ and
$\QQQ := (Q_1, \dots, Q_q)$ be a $(3p,q)$-fence with top $A := \{ a_1,
\dots, a_q\}$ and bottom $B := \{ b_1, \dots,
b_q\}$. Let $a_i$ and $b_i$ be the endpoints of
$Q_i$,
for all $1\leq i \leq q$. Recall that we assume
that $\PPP$ and $\QQQ$
are ordered from top to
bottom and from left to right, respectively.
We divide $\FFF$ into three parts $\FFF_1, \FFF_2, \FFF_3$,
where, for $i=1,2,3$, $\FFF_i$ consists of the paths $P_j$ with $2(i-1)p+1 \leq j \leq
2ip$ together with the subpaths of the paths $Q$ in $\QQQ$ from the first
vertex that $Q$ has in common with $P_{2(i-1)p+1}$ (or the
first vertex of $Q$ in case $i=1$) and the last vertex of $Q$ before
$Q$ intersects $P_{2p+1}$ (or the last vertex of $Q$ in case $i=3$).
See Figure~\ref{fig:critical-schema}.

Let $\RRR=\{R_1,\dots, R_{\frac q3}\}$ be such that
the linkage $\RRR$ joins the last third
$(b_{\frac23q+1},$ $\dots,$ $b_q)$ of the
bottom vertices
to the first third $(a_1, \dots, a_{\frac
  q3})$ of the top vertices. By Lemma~\ref{lem:jumps}, we may assume that $\RRR$ goes up row by row in $\FFF$.   We define the
  following notation for the rest of this section.

\begin{itemize}
\item Let $x_i$ be the last vertex of $R_i$ in $\FFF_3$ for
  $i=1, \dots, \frac q3$. Let $X=\{x_1,\dots,x_{\frac q3}\}$.
\item Let $y_i$ be the
  first vertex of $R_i$ in $\FFF_1$ for $i=1,\dots,\frac q3$. Let
  $Y=\{y_1,\dots,y_{\frac q3}\}$.
\item Let $\RRR'$ be the linkage obtained from
  $\RRR$ by taking the subpath of each path in $\RRR$ between one
  endpoint in $X$ and the other endpoint in $Y$.
\item   Let $a'_i$ be the first vertex of $Q_i$ in $\FFF_2$, for $1\leq i \leq q$. Let $A'=\{a'_1,\dots, a'_{q}\}$.
\item Let $b'_i$ be the last vertex of $Q_i$ in
  $\FFF_2$, for $1 \leq i \leq q$. Let
  $B'=\{b'_1,\dots, b'_{q}\}$.
\item Let $Q'_i$ be the subpath of $Q_i$
  between $a'_i$ and $b'_i$ for $i=1,\dots, q$. Let $\QQQ'= \{Q'_1,\dots,Q'_{q}\}$.
\end{itemize}

\begin{figure}[th]
  \centering
  \includegraphics[height=8cm]{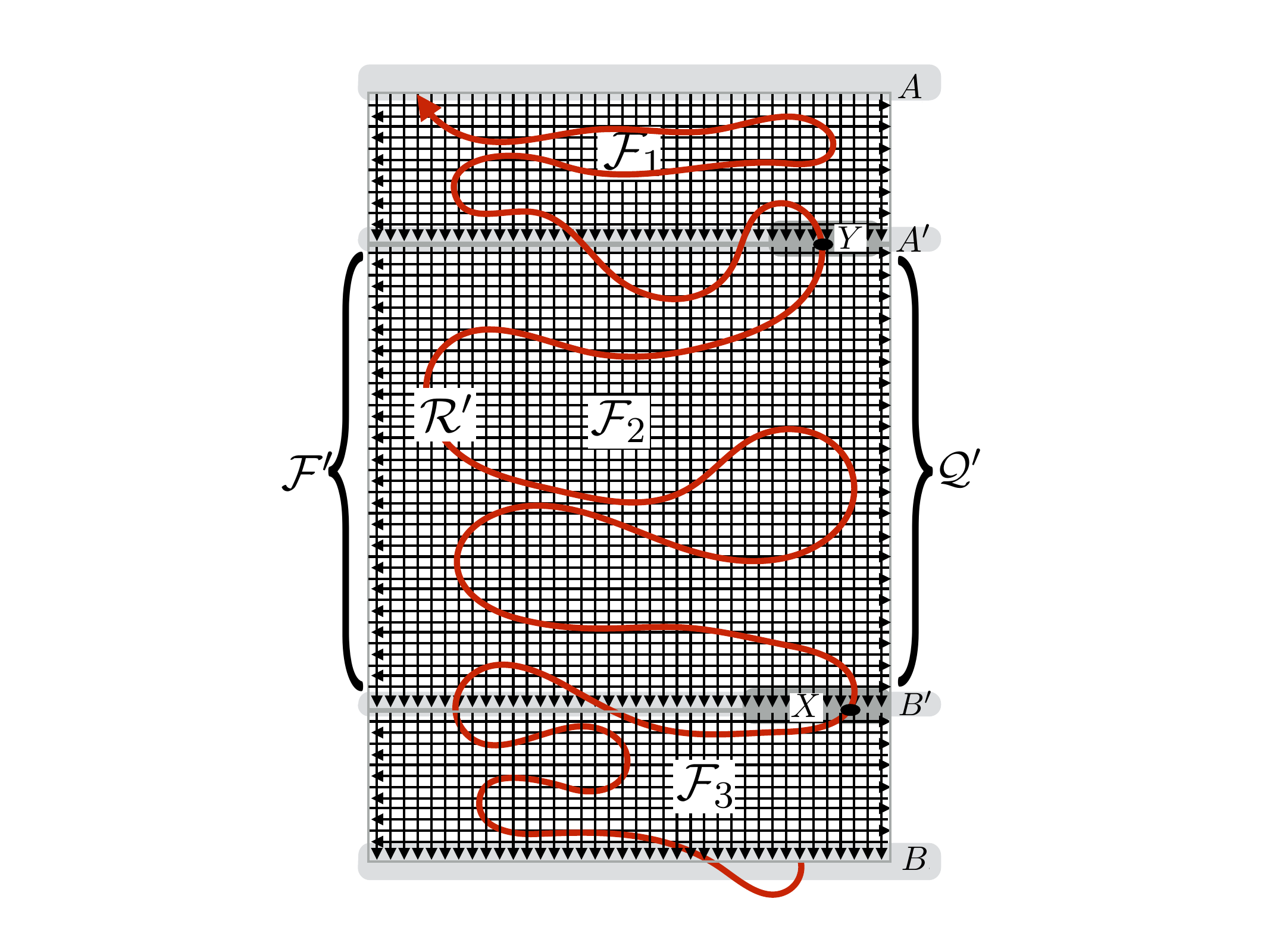}
  \caption{Schematic overview of the situation in Section~\ref{sec:critical}.}
  \label{fig:critical-schema}
\end{figure}

Let $\FFF' = \bigcup \QQQ' \cup \bigcup \RRR' \cup \bigcup \{P_i \sth 2p+1\leq i \leq 4p\}$. Figure~\ref{fig:critical-schema}
illustrates the notation introduced so far.
By our assumption,
no vertex in $\FFF'$ is
in $\FFF_1$ or in $\FFF_3$, except for the endpoints of paths in
$\QQQ' \cup \RRR'$.

The next goal of our proof, which is Lemma~\ref{lem:reroute-R*} -- the most technical in this section -- is to further develop the structure in $\FFF'$ in the following way: 
we have a linkage $\RRR'' \subseteq \RRR'$ of large order  and a linkage $\QQQ''$ in
  $\FFF'$ of large order
 such that  for every $Q\in \QQQ''$ there is a split edge $e(Q) \in
 E(Q)\setminus E(\RRR'')$ splitting
 $Q$ into two subpaths
 $l(Q)$ and $u(Q)$ with $Q = u(Q)e(Q)l(Q)$.
Furthermore, for every $R\in \RRR''$ there are distinct edges $e_1(R), e_2(R)$ splitting $R$
 into subpaths $l(R), u_1(R)$ and $u_2(R)$ such that $R =
 l(R)\,e_1(R)\,u_1(R)\,e_2(R)\,u_2(R)$ and
 \begin{enumerate}
 \item the subpath $u_1(R)e_2(R)u_2(R)$ does not intersect
   $l(Q)$ for every $Q\in \QQQ''$
 \item $u_1(R)$ and $u_2(R)$ both intersect every $u(Q)$ for $Q\in \QQQ''$
 \item $l(R)$ intersects every $l(Q)$ for $Q\in \QQQ''$ (but may
   also intersect $u(Q)$).
 \end{enumerate}

More precisely, we show the following lemma.

\begin{lemma}\label{lem:reroute-R*}
  For every $t, r', q'$ there are $r, q, q^*$ and $p$, where $q^*$
  only depends on $q'$ and $t$ but not on $r'$, such that if $\RRR'$,
  $\QQQ'$ and $\PPP$ are as above and of order $r$, $q$ and $3p$,
  respectively, then either $G$ contains a cylindrical grid of order
  $t$ as a butterfly minor or there is a linkage
  $\RRR'' \subseteq \RRR'$ of order $r'$ and a linkage $\QQQ''$ in
  $\FFF'$ of order $q'$ such that the start and endpoints of paths in $\QQQ''$  come from the set of start and endpoints of the paths in $\QQQ'$ and such that every
  $Q\in \QQQ''$ hits every $R\in \RRR''$. Furthermore, for every
  $Q\in \QQQ''$ and every $e\in E(Q)\setminus E(\RRR'')$ there are at
  most $q^*$ paths from $Q_1$ to $Q_2$ in $\RRR''\cup \QQQ'' - e$,
  where $Q = Q_1\,e\,Q_2$.

  In addition, for every $Q\in
  \QQQ''$ there is an edge $e = e(Q)\in E(Q)\setminus E(\RRR'')$ splitting
  $Q$ into two subpaths
  $l(Q)$ and $u(Q)$ with $Q = u(Q)\,e\,l(Q)$
  and for every $R\in \RRR''$ there are edges $e_1 = e_1(R), e_2 = e_2(R)$ splitting $R$
  into three subpaths $R_1, R_2, R_3$ with $R =
  R_1e_1R_2e_2R_3$ such that $R_2$ and $R_3$
  both intersect $u(Q)$ for all $Q\in \QQQ''$ but not $l(Q)$ and $R_1$
  intersects $l(Q)$ for all $Q\in \QQQ''$.
\end{lemma}
\begin{proof}
  \setcounter{claimcounter}{0}
  Let $g$ be the function implicitly defined in Lemma~\ref{ll2},
  i.e.~let $g\sth \N\rightarrow \N$ be such that if $t'=g(p)$ then the
  condition of Lemma~\ref{ll2} is
  satisfied. Furthermore, let $f_r, f_p\sth \N\rightarrow \N$ be
  the
  functions defined in Lemma~\ref{lem:new-good-tuples}.
  Starting from $\RRR', \QQQ'$ and $\PPP$, in the course of the proof we will take several subsets of these of decreasing order and split paths in other ways. Instead of calculating numbers directly we will state the necessary conditions on the order of these sets as we go along.

  In analogy to Definition~\ref{def:canonical-order} we divide $\FFF_2$
  into rows $\ZZZ_0, \dots, \ZZZ_{k+1}$ ordered from top to
  bottom, each containing
    $2\cdot q_1$ paths from $\PPP$. We require
    \begin{eqnarray}
	p 	&\geq& (k+2)\cdot 2 q_1\\
	k&\geq& 2\cdot q_2.
    \end{eqnarray}

  Recall that \marginpar{$\RRR'$} $\RRR'$ consists of subpaths of paths $R_i \in \RRR$
  starting at $x_i$ and ending at $y_i$. As by assumption $\RRR$ goes up row by row, the initial
  subpath of $R_i$ from its beginning to $x_i$ may contain some vertices
  of $\FFF_2$ but only in row $\ZZZ_{k+1}$. Similarly, the final subpath
  of $R_i$ from $y_i$ to the end can contain vertices of $\FFF_2$ but
  only in $\ZZZ_0$. Hence, $R_i \cap \bigcup_{j=1}^k \ZZZ_j \subseteq
  \RRR'$. %

  Let $\ZZZ := \bigcup_{i=1}^k \ZZZ_i$  \marginpar{$\ZZZ$, $\PPP_Z \subseteq \PPP$} and 
  let $\PPP_Z \subseteq \PPP$ 
  be the set of paths in $\PPP$ contained in $\ZZZ$ and let $\QQQ_Z$ \marginpar{$\QQQ_Z$} be the
  maximal subpaths of paths in $\QQQ$ which are entirely contained in
  $\ZZZ$. Finally, for every $1\leq i \leq k$ let $\QQQ_{\ZZZ_i}$ be the maximal subpaths of the paths in $\QQQ$ contained in $\ZZZ_i$.

  Recall that $\RRR'$ consists of subpaths of $\RRR$ and
  that $\RRR$ is a bottom up linkage which, by the assumption above,
  goes up row by row. Hence, $\RRR'$ also goes
  up row by row in terms of $\ZZZ$ within $\FFF_2$.

  The rest of the proof of this lemma goes as follows: 
  We first show that there is no row $\ZZZ_i$ and sets $\QQQ_1 \subseteq
    \QQQ$ and $\RRR_1 \subseteq \RRR'$, both of large order, such that
    $V(\QQQ_1)\cap V(\ZZZ_i) \cap V(\RRR_1) = \emptyset$. This actually means 
    that there is no subset $\RRR_1 \subseteq    \RRR'$ that ``jumps'' over some row $\ZZZ_i$.
  Otherwise, a reduction to Lemma~\ref{ll2} will lead to a cylindrical grid of order $t$ as a butterfly minor.
  
  We then start constructing a desired structure in several steps. 
  In Step 1 below, roughly, we obtain a grid $\HHH := (\BBB', \RRR_3)$ such that $\RRR_3\subseteq \RRR'$ and the paths in $\BBB'$
  start and end on vertices of paths in $\QQQ'$. To obtain such a grid we apply Lemma~\ref{lem:new-good-tuples}. Thus the statements in Lemma~\ref{lem:new-good-tuples} also hold for this grid $\HHH$. 
  
  Next in Step 2  we take an $\RRR_3$-minimal linkage $\QQQ_3$ such that the start and end vertices of	the paths in $\QQQ_3$ are from the set of start and end
        vertices of	the maximal subpaths of $\QQQ'$ in $\ZZZ$. 
        Then in Claim 2 below, we do have many ``jumps'' over some big part of the grid $\HHH$. Note that unlike in the lemmas above, here
  the  jumps go downwards. See Figure~\ref{fig:Rstar} for an illustration
    of constructing a cylindrical grid of order $t$ as a butterfly minor, if many such  jumps exist.
      
    Now Steps 1 and 2 allow us to show that there are desired linkages $\RRR'' \subseteq \RRR'$ of order $r''$ and a linkage $\QQQ''$ in  $\FFF'$ of order $q'$, that are obtained from $\QQQ_3$ and $\RRR_3$ by possibly shrinking or taking some subset of them. 
  
  \medskip
  Now let us begin with Claim 1. 
  For the following claim we require
  \begin{equation}
  q_1 \geq 2\cdot g(t).
  \end{equation}

  \begin{Claim}
    If there is  a row $\ZZZ_i$, for some $1\leq i\leq k$,
    and sets $\QQQ_1 \subseteq
    \QQQ$ and $\RRR_1 \subseteq \RRR'$, both of order $q_1$, such that
    $V(\QQQ_1)\cap V(\ZZZ_i) \cap V(\RRR_1) = \emptyset$  then $G$
    contains a cylindrical grid of order $t$ as a butterfly minor.
  \end{Claim}
  \begin{ClaimProof}
    Suppose $\QQQ_1$ and $\RRR_1$ exist. Let $\PPP_{Z_i}$ be the set
    of paths from $\PPP$ contained in $\ZZZ_i$
    and let $\QQQ^1_{Z_i}$ be
    the subpaths of paths in $\QQQ_1$ restricted to $\ZZZ_i$. Then
    $(\PPP_{Z_i}, \QQQ^1_{Z_i})$ form a $(q_1,
    q_1)$-fence such that $\RRR_1$ avoids
    $\QQQ^1_{Z_i}$.
    Let $A_{\ZZZ_i}\subseteq A$ and $B_{\ZZZ_i}
    \subseteq B$ be the end points and start points of the paths in
    $\RRR_1$, respectively.
    Let $A'_{\ZZZ_i}$ be the set of start points and $B'_{\ZZZ_i}$ be
    the set of end points of $\QQQ^1_{\ZZZ_i}$.
    Then, in $\FFF$, there is a linkage
    $\LLL_A$ from $A_{\ZZZ_i}$ to $A'_{\ZZZ_i}$ of order $q_1$ and
    there is a linkage $\LLL_B$ of order $q_1$ from $B'_{\ZZZ_i}$ to
    $B_{\ZZZ_i}$ such that these two linkages are
    internally disjoint from $(\PPP_{\ZZZ_i},
    \QQQ^1_{\ZZZ_i})$.
    Hence, $\LLL_A\cup\LLL_B\cup \RRR_1$ forms a
    half-integral linkage from the bottom $B'_{\ZZZ_i}$ to the top
    $A'_{\ZZZ_i}$ of the fence $(\PPP_{\ZZZ_i},
    \QQQ^1_{\ZZZ_i})$. By
    Lemma~\ref{lem:half-integral}, there also exists an integral
    linkage $\LLL \subseteq \LLL_A\cup\LLL_B\cup \RRR_1$ of order
    $\frac{q_1}2 = g(t)$ from
    $B'_{\ZZZ_i}$ to $A'_{\ZZZ_i}$. Hence, by Lemma~\ref{ll2}, $G$
    contains a cylindrical grid of order $t$ as a butterfly minor as required.
  \end{ClaimProof}

  By the previous claim, we can now assume that in each row $\ZZZ_i$
  at most $q_1\cdot
  \binom{q}{q_1}$ paths in $\RRR'$ avoid at least $q_1$ paths in
  $\QQQ'$ restricted to $\ZZZ_i$. For otherwise, by the pigeon hole
  principle there would be a row $\ZZZ_i$ and a set  $\RRR_1\subseteq
  \RRR'$ 
  of order $q_1$
  such that every $R\in \RRR_1$ avoids the same set $\QQQ_1\subseteq
  \QQQ_{\ZZZ_i}$ 
  of at least $q_1$ paths. Hence, by
  the previous claim, this would imply a cylindrical grid of order
  $t$ as a butterfly minor. The rest of the proof needs several steps.

  \smallskip \noindent\textbf{Step 1. }
  Let us now consider the rows \marginpar{$\ZZZ'$}$\ZZZ' := \{ \ZZZ_{2i} \sth 1\leq i \leq
  q_2\}$, which is possible as $k\geq 2q_2$. It follows that there is a set $\RRR_2^* \subseteq \RRR'$ 
  of
  order $r_2^* := r-q_2\cdot q_1\cdot \binom{q}{q_1}$ such that in each $Z\in
  \ZZZ'$ each path of $\RRR_2^*$ hits all but at most  $q_1$ paths in
  $\QQQ'$ restricted to row $Z$. As we require
  \begin{equation}
    r \geq q_2\cdot q_1\cdot
    \binom{q}{q_1} + \binom{q}{q_1}^{q_2}\cdot r_2
  \end{equation} 
  and therefore $r_2^* \geq
  \binom{q}{q_1}^{q_2}\cdot
  r_2$, we can find a set $\RRR_2\subseteq \RRR_2^*$ \marginpar{$\RRR_2\subseteq \RRR_2^*$} of order
  $r_2$ such that
  any two paths in $\RRR_2$ hit in each row $Z \in \ZZZ'$ exactly the
  same set $\QQQ'_Z\not=\emptyset$ of paths in $\QQQ'$ restricted to $Z$.
  We now choose in each row $\ZZZ_{2i}\in \ZZZ'$, for $1\leq i\leq
  q_2$,  a path $Q_{2i} \in \QQQ'_{Z_{2i}}$ 
  such that every $R\in \RRR_2$ has a non-empty intersection with $Q_{2i}$
  in $\ZZZ_{2i}$. For all $1\leq i\leq q_2$ let $Q^2_i$ be the restriction of $Q_{2i}$ to row
  $\ZZZ_{2i}\in \ZZZ'$ and let $\QQQ_2 := \{ Q^2_i \sth 1\leq i \leq q_2\}$\marginpar{$\QQQ_2$}.

  As $\RRR'$, and hence $\RRR_2$, goes up row by row, it follows that
  all paths in $\RRR_2$ go through the paths in $\QQQ_2$ strictly in the same
  order $Q^2_{q_2}, \dots, Q^2_1$. Hence, $(\QQQ_2, \RRR_2)$ forms a
  weak $(q_2, r_2)$-split of $(\QQQ', \RRR_2)$ (recall Definition~\ref{def:weak-split} of weak splits).
  We require
  \begin{eqnarray}
  r_2 &\geq& f_p(r_3)\\
  q_2&\geq& f_q(q_5)  {r_2 \choose r_3}r_3!q_5,
  \end{eqnarray}
where $f_p, f_q$ are the functions defined in Lemma~\ref{lem:new-good-tuples}.

  W.l.o.g.~we assume that $q_2 = f_q(q_5)\binom{r_2}{q_5}(q_5)!q_5$.
  Let $q_2^* :=  {r_2 \choose r_3}r_3!q_5$. For ease of presentation
  we renumber the paths in $\QQQ_2$ as $(Q^1_1,\dots,
  Q_1^{f_q(q_5)}, \dots, Q_{q_2^*}^1, \dots, Q_{q_2^*}^{f_q(q_5)})$ in the order in which the
  paths in $\RRR_2$ traverse the paths in $\QQQ_2$.
  By definition, for all $1\leq i \leq {q_2^*}, 1\leq j \leq f_q(q_5)$, the path $Q_i^j$ intersects
  every path in $\RRR_2$ and every path $R\in \RRR_2$ can be split into
  disjoint segments $R_1, \dots, R_{q_2}$ occurring in this order on $R$
  such that for all $1\leq i \leq {q_2^*}, 1\leq j \leq f_q(q_5)$, the path
  $R$ intersects $Q_i^j$ only in segment $R_{q_2+1 - (i-1)\cdot
    f_q(q_5)-j}$.

For all $1\leq i \leq {q_2^*}$ and $R\in \RRR_2$ let $R^i$ be
  the minimal subpath of $R$ containing $V(R)\cap (\bigcup_{1\leq j
    \leq f_q(q_5)}V(Q_i^j))$ and let $\RRR_2^i := \{ R^i \sth R\in
  \RRR_2\}$\marginpar{$\RRR_2^i$}.
  As $|\RRR_2|\geq f_p(r_3)$,
  we can now apply Lemma~\ref{lem:new-good-tuples} to $(\RRR_2^i,
  \{Q_i^1, \dots, Q_i^{f_q(q_5)}\})$, for all $1\leq i \leq q_2^*$, to
  obtain a sequence $\hat\RRR_2^i :=  (R^i_{1}, \dots, R^i_{r_3})$ \marginpar{$\hat\RRR_2^i$} of
  paths $R^i_{j}\in \RRR_2^i$ and a path $A_i$ as in the statement of the
  lemma. 
  As $q_2 \geq q_2^*f_q(q_5)$ and $q_2^* =  {r_2 \choose r_3}r_3!q_5$, there are paths $R_1, \dots, R_{r_3}\in \RRR_2$
  and   $q_5$ values $i_1 <
  \dots < i_{q_5}$ such that $R^{i_l}_{s}$ is a subpath of $R_s$, for
  all $1\leq l \leq q_5, 1\leq s \leq r_3$. For $1\leq l \leq q_5$ let
  $B_l := A_{i_l}$ and let $\BBB' := \{ B_1, \dots B_{q_5}\}$\marginpar{$\BBB'$} and  
  $\RRR_3 := \{R_1, \dots, R_{r_3}\}$\marginpar{$\RRR_3$}.  Hence, 
  $\HHH := (\BBB', \RRR_3)$ \marginpar{$\HHH$} forms a
  $(q_5,r_3)$-grid such that $\RRR_3\subseteq
  \RRR_2$. 

\color{black}

W.l.o.g.~let $\RRR_3 := (R_1, \dots, R_{r_3})$ be ordered in the order
in which the paths in $\RRR_3$ occur on the grid $\HHH$, from
left to right, and let $\BBB' := (B_1, \dots, B_{q_5})$ be
ordered in the order in which they appear in $\HHH$ from top
to bottom. That is, the paths in $\RRR_3$
        go through the paths in $\BBB'$ in the order $B_{q_5}, \dots,
        B_1$.

  Note that, by Lemma~\ref{lem:new-good-tuples}, the paths $A_{i_j}$
  start and end on vertices of paths in $\QQQ'$. As the paths in
  $\QQQ'$ are subpaths of paths in $\QQQ$ chosen from different rows
  $\ZZZ_i$, there is a path in the original fence linking the endpoint
  of $B_{2i-1}$ to the start vertex of $B_{2i}$. 
 \color{black}

  Let $\BBB$\marginpar{$\BBB$} be
  the set of maximal subpaths of paths in $\BBB'$ with both endpoints
  on paths in $\RRR_3$ but which are internally vertex disjoint from
  $\RRR_3$. 

Finally, for every $B\in \BBB'$ let $r_u(B) := \min \{ l \sth V(B)\cap
V(\ZZZ_l)\not=\emptyset\}$\marginpar{$r_u(B), r_l(B)$} and let $r_l(B)$ be the maximal number in
this set. That is, the path $B$ is formed by subpaths of paths in
$\QQQ_2$ in the rows $\bigcup_{l=r_u(B)}^{r_l(B)} \ZZZ_l$.

        \medskip

\begin{recap}
   Let us quickly recap the notation that is still needed in the remainder of the proof. $\RRR_3$ is the bottom up linkage which together with the paths in $\BBB'$ constitutes the grid $\HHH$. $\BBB$ is the set of maximal subpaths of paths in $\BBB'$ which start and end on paths in $\RRR_3$ but are otherwise disjoint from $\RRR_3$. $\QQQ'$ are the subpaths of the vertical paths in the original fence $\FFF$ connecting $\FFF_1$ to $\FFF_3$. Finally, $r_u(B)$ and $r_l(B)$ are defined as in the previous paragraph. 
\end{recap}

\noindent\textbf{Step 2. } Let $\tilde{\QQQ}_2 := \{ Q\in
        \QQQ' \sth V(Q)\cap V(\BBB')\not=\emptyset\}$, i.e.~$\tilde{\QQQ}_2$ contains those paths from $\QQQ'$ that contain a subpath which is part of the grid $\HHH$. Now let
        $\QQQ_3\subseteq \bigcup(\RRR_3\cup (\QQQ'\setminus \tilde{\QQQ}_2))\cap
        \ZZZ$ 
   be an $\RRR_3$-minimal linkage of
	order $q_3$, for some value of $q_3$ to be determined below,
        such that the start and end vertices of
	the paths in $\QQQ_3$ are from the set of start and end
        vertices of
	the maximal subpaths of $\QQQ'$ in $\ZZZ$.  This is possible
        as we require
	\begin{equation}
	q \geq q_3+q_2.
	\end{equation}
	We set $q^* := q_3$\marginpar{$q^*$}.
Note that, by minimality, for every $Q\in \QQQ_3$ and every
	edge $e\in E(Q)\setminus E(\RRR_3)$, if $Q = Q_1\, e\, Q_2$ then
	there are at most $q^*$ paths from $Q_1$ to $Q_2$ in
	$\big(\RRR_3\cup\QQQ_3)-e$.
	Furthermore, note that as $\RRR_3$ and $\QQQ_3$ are sets of
        pairwise disjoint paths, this property even holds for edges in
        $E(Q)\cap E(\RRR_3)$. For, if $e\in E(R)\cap E(Q)$ for some
        $R\in \RRR_3$ and $Q\in \QQQ_3$ then there is a maximal common
        subpath $P$ of $R$ and $Q$ that contains $e$. Let $x, y$ be
        the endpoints of $P$ with $x$ being the start vertex. But then $Q$
        must contain an edge $e' \in E(Q)\setminus E(R)$ with endpoint
        $x$ or such an edge $e''$ with start point $y$. Assume $e'$
        exists and let $Q = Q_1 e Q_2$ and $Q = \hat Q_1 e' \hat
        Q_2$. But then,
        any linkage  $L$ of order $q^*$ from $Q_1$ to $Q_2$ in
        $(\RRR_3\cup \QQQ_3)-e$ must be a linkage from $\hat Q_1$ to
        $\hat Q_2$ as no path in $\RRR_3\cup \QQQ_3$ other than $R$
        and $Q$ contains any vertex of $P$. This is needed below.

	We require that
	\begin{eqnarray}
	r_3 &\geq& t_w\cdot t_c,\\
	q_5 & \geq & t_w\cdot t_r,
	\end{eqnarray}
	for some values of $t_c, t_w, t_r$ to be determined below.
	
        Recall that $\RRR_3 := (R_1, \dots, R_{r_3})$ is ordered in the order
        in which the paths in $\RRR_3$ occur on the grid $\HHH$, from
        left to right, and that $\BBB' := (B_1, \dots, B_{q_5})$ is 
        ordered in the order in which the paths appear in $\HHH$ from top
        to bottom. 
	For $1\leq j \leq t_c$ let $\RRR^j := \{ R_{(j-1)\cdot t_w+1},
        \dots, R_{j\cdot t_w}\}$.
	Furthermore, for every $1\leq i \leq t_r$ we let $\RRR^j_i$  \marginpar{$\RRR^j_i$} be
        the set of subpaths of paths $R\in \RRR^j$ starting at the
        first vertex of $R\cap B_{i\cdot t_w}$ and ending at the
        last vertex of $R\cap B_{(i-1)\cdot t_w+1}$. Finally, let
        $\SSS_{i,j}$\marginpar{$\SSS_{i,j}$} be the subgrid of $\HHH$ induced by the paths
        $\RRR^{j}_i$ and the minimal subpaths of the paths $B\in \{
        B_{i\cdot t_w}, \dots,  B_{(i-1)\cdot t_w+1}\}$ which contain
        all of $V(B)\cap V(\RRR^{j}_i)$.
	We set $I := \{ 1, \dots, t_r\}$ and $J := \{ 1, \dots,
        t_c\}$.
	Finally, we set $r_u(\SSS_{i,j}) := r_u( B_{(i-1)\cdot t_w+1})$
        and $r_l(\SSS_{i,j}) := r_l(B_{i\cdot t_w})$\marginpar{$r_u(\SSS_{i,j})$, $r_l(\SSS_{i,j})$}. 
        As the paths $B$
  start and end on vertices of paths in $\QQQ'$  by Lemma~\ref{lem:new-good-tuples}
        all
        paths $B$ which intersect $\SSS_{i,j}$ are contained in
        $\bigcup_{r_u(\SSS_{i,j})\leq l \leq r_l(\SSS_{i,j})} \ZZZ_l$.
	
	For all $i\in I$ and $j\in J$ let $\alpha(S_{i,j})$\marginpar{$\alpha(S_{i,j})$}
	be the set of paths $Q\in \QQQ_3$  which contain a subpath $Q^*\subseteq
        Q$ with first vertex in $V(\RRR_3)\cap \bigcup_{l<
          r_u(\SSS_{i,j})}V(\ZZZ_l)$, last vertex in $V(\RRR_3)\cap
        \bigcup_{l> r_l(\SSS_{i,j})}V(\ZZZ_l)$ and internally vertex
        disjoint from $V(\RRR_3 \cap \SSS_{i,j})$. We call such a subpath $Q^*$ a
        "\emph{jump}" of $Q$ over $\SSS_{i,j}$.
  Note that the paths $Q^*$ do not really "jump" over $\SSS_{i,j}$, they merely avoid the paths in $\RRR_3$ within $\SSS_{i,j}$. However, if there are sufficiently many of such paths this will allow us to apply Lemma~\ref{lem:ll2-general}. 
  
  \begin{figure}
\includegraphics{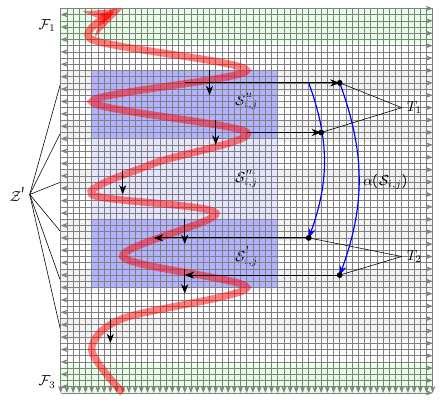}    
        \caption{Illustration of the construction
          in Claim~\ref{claim:R*:2}
          of
          Lemma~\ref{lem:reroute-R*}.}
        \label{fig:Rstar}
		 \end{figure}	

  \begin{Claim}\label{claim:R*:2}
    There is a number $t_2$ depending only on $t$
    such that if there is a pair $i \in I, j\in
    J$ such that
    $|\alpha(\SSS_{i,j})|\geq t_2$ then $G$ contains a cylindrical
    grid of order $t$ as a butterfly minor.
  \end{Claim}
  \begin{ClaimProof}
    See Figure~\ref{fig:Rstar} for an illustration
    of the following construction.
    First, we require that
    \begin{eqnarray}
      t_2 &\geq& t_3\\
      t_w &\geq& t_v' + 2\cdot t_3,
    \end{eqnarray}
    for some numbers $t_3$ and $t_v'$ determined below.
    Let $\SSS_{i,j}^u$\marginpar{$\SSS_{i,j}^u$} be the subgrid of
    $\SSS_{i,j}$
    induced by the paths $B_{(i-1)t_w+1}$, $\dots, B_{(i-1)t_w+t_3}$, let
    $\SSS_{i,j}^l$\marginpar{$\SSS_{i,j}^l$} be the subgrid of $\SSS_{i,j}$ induced by the paths
    $B_{it_w-t_3}, \dots, B_{it_w}$ and let $\SSS_{i,j}^m$\marginpar{$\SSS_{i,j}^m$} be the
    remaining subgrid, i.e.~the subgrid of $\SSS_{i,j}$ induced by the
    paths $B_{(i-1)t_w+t_3+1}, \dots, B_{it_w-t_3-1}$.
		
    For every $Q\in \alpha(\SSS_{i,j})$ choose a subpath $Q^*$
    as above, i.e.~a path $Q^*$ that jumps over the subgrid
    $\SSS_{i,j}$. Recall that $\SSS_{i,j}$ is the subgrid of $\HHH$ formed
    by the paths in $\RRR^j_i$ and  the minimal subpaths of the paths
    $B\in \{ B_{(i-1)\cdot t_w+1}, \dots, B_{i\cdot
    t_w}\}$ which
    contain all of $V(B)\cap V(\RRR^{j}_i)$.
    Each $B_l$, for $(i-1)t_w+1 \leq l
    \leq it_w$, is constructed from subpaths of paths $Q'\in
    \QQQ'$
    within distinct rows $\ZZZ_{l'}\in \ZZZ'$. As $\RRR'$ and therefore
    $\RRR_3$ is going up row by row within $\FFF$, and $\ZZZ'$ only
    contains every second row of $\ZZZ$, 
       there is no path from any
    vertex in some row $\ZZZ_{2l}\in \ZZZ'$ to a vertex in row
    $\ZZZ_{2(l+1)}$ that does not contain a vertex of some path $Q'\in
    \QQQ'$. See Figure~\ref{fig:Rstar} for an illustration.

    Hence, there is an initial subpath $\hat{Q}_1$ of
    $Q^*$ which has a non-empty intersection with at least $t_3$ distinct
    rows in $\SSS_{i,j}^u$ and a terminal subpath $\hat{Q}_2$ of $Q^*$
    which has a non-empty intersection with $t_3$ rows in
    $\SSS_{i,j}^l$. Let $\hat{Q}$ be the remaining subpath of $Q^*$
    such that $Q^* = \hat{Q}_1\,\hat{Q}\,\hat{Q}_2$.

    As $|\alpha(\SSS_{i,j})|\geq t_2 \geq t_3$ we can choose a
    set $\{Q_1, \dots,Q_{t_3}\}$ of $t_3$ distinct paths $Q \in
    \alpha(\SSS_{i,j})$ and $t_3$ distinct rows
    $\{ \ZZZ_i \sth i\in I'\}$ that intersect 
    $\SSS_{i,j}^u$, for an index set $I'$ of
    order $t_3$, and a set $T_1 :=
    \{ v_i \sth i\in I'\}$ of vertices on these
    rows which are on distinct
    paths $\hat{Q}_1$, for $Q\in \{Q_1, \dots,
    Q_{t_3}\}$.  Furthermore, we can choose $t_3$ rows $\{\ZZZ_i
    \sth i\in J' \}$ that intersect $\SSS_{i,j}^l$, for some
    index set $J'$ of order
    $t_3$, and a set $T_2 := \{ u_i \sth i\in
    J'\}$ of vertices on these
    rows which are on distinct paths $\hat{Q}_2$ for $Q\in \{Q_1, \dots,
    Q_{t_3}\}$.
    
    Now observe that the paths in $\RRR_3$ also intersect every row in $\SSS_{i,j}^u$ and thus there is a linkage $\LLL_u$ of order $t_3$ from the bottom of $\SSS_{i,j}$, i.e.~the last vertices $\RRR_3$ has in common with $\SSS^m_{i,j}$, to the set $T_1$. 
    
    Similarly, there is a linkage $\LLL'$ of order $t_3$ from $T_2$ to the top of $\SSS_{i,j}^m$. Thus, 
    \[
    \bigcup\LLL\cup\bigcup\LLL'\cup \bigcup\RRR_3 \cup \bigcup \{ \hat Q \sth Q\in \{Q_1, \dots, Q_{t_3}\}\}
    \]
    contains a half-integral linkage of order $t_3$, and thus an integral linkage $\LLL''$ of order $\frac12t_3$, from the bottom of $\SSS_{i,j}^m$ to its top which avoids $V(\RRR_3 \cap \SSS_{i,j}^m)$. 
   See Figure~\ref{fig:Rstar} for an illustration.
   Note that this
    linkage $\LLL''$  satisfies the condition of 
    Lemma~\ref{lem:ll2-general}. We require that
    \begin{eqnarray}
      \label{eq:14}
      \frac12t_3 & \geq & t'\\
      t_v' & \geq & t'',
    \end{eqnarray}
    where $t', t''$ are the integers defined in Lemma~\ref{lem:ll2-general}. We can now apply Lemma~\ref{lem:ll2-general} to obtain a cylindrical grid of
    order $t$ as a butterfly minor, as required.
  \end{ClaimProof}

  Thus, we can now assume that $|\alpha(\SSS_{i,j})|\leq t_2$ for all
  $i\in I$ and $j\in J$. In particular, this implies that every $Q\in
  \QQQ_3\setminus \alpha(\SSS_{i,j})$ intersects a path in $\SSS_{i,j}$.

  As we require that
  \begin{equation}
    t_r \geq \binom{q_3}{t_2}^{t_c}\cdot t_r',
  \end{equation}
  for some number $t_r'$ determined below, there
  is a subset $I'\subseteq I$ of order
  $t_r'$ such that $\alpha(\SSS_{s, j}) = \alpha(\SSS_{s', j})$ for all
  $s, s'\in I'$ and $j\in J$.  Furthermore, as
  \begin{equation}
    t_c \geq t'_c\binom{q_3}{t_2},
  \end{equation}
  for some number $t'_c$ determined below, there
  is a subset $J'\subseteq
  J$ of order
  $t'_c$ such that $\alpha(\SSS_{i,j}) = \alpha(\SSS_{i,j'})$ for all
  $i\in I'$ and $j,j'\in J$.

  Now let $\QQQ_4 := \QQQ_3 \setminus \alpha(\SSS_{i,j})$\marginpar{$\QQQ_4$} for
  some (and hence all) $i\in I', j\in J'$. Let
  $q_4 := |\QQQ_4|$. So every $Q\in \QQQ_4$ has a
  non-empty intersection with some $R\in \RRR^j_i$.

  For every $Q\in \QQQ_4$ and all $i\in I'$ let $v_i(Q)$ be the
  last vertex on $Q$ in $V(\bigcup_{j\in J'} \RRR^j_i)$, i.e.~the last
  vertex of $Q$ in the row $i$, when traversing $Q$ from beginning to
  end. 

  For simplicity of notation we assume that $I' := \{ 1, \dots,
  t'_r\}$.

  \begin{Claim}\label{claim:R*:3}
    $v_{1}(Q), \dots,
    v_{t'_r}(Q)$ appear on $Q$ in this order and
    furthermore, for all
    $i\in I'$, the subpath of $Q$ from $v_{i}$ to
    the end of $Q$ has a
    non-empty intersection with $V(\RRR^j_{i+1})$, for all $j\in J'$.
  \end{Claim}
  \begin{ClaimProof}
    If, for some $i\in I'$, the subpath of $Q$
    from $v_{i}$ to the end of $Q$ has an
    empty intersection with $V(\RRR^j_{i+1})$,
    for some $j\in J'$, then $Q_i$ would be in
    $\alpha(\SSS_{i, j})$, a contradiction to the
    choice of $I'$ and $\QQQ_4$. Furthermore,
    if there were $i<i' \in I'$
    such that $v_{i}(Q)$ appears on $Q$ after
    $v_{i'}(Q)$, then again the
    subpath of $Q$ from $v_i(Q)$ to the end of $Q$ would not intersect any
    $\RRR^j_{i'}$, for $j\in J'$. Hence, by definition of $\alpha(\SSS_{i', j})$, $Q$ would be contained in $\alpha(\SSS_{i', j})$, contradicting
    the choice of $\QQQ_4$.
  \end{ClaimProof}

  Now, for all $Q\in \QQQ_4$, $i\in I'$ with $i>1$ and $j\in J'$ let
  \[
     \beta_Q(\SSS_{i,j}) := \{ R \in \RRR^j \sth \begin{array}{l} Q \text{ intersects $R\cap \RRR^j_i$ in the
       subpath }\\
       \text{of $Q$ from $v_{i-1}(Q)$ to the end of $Q$}
     \end{array}\}.
   \]
   Note that $\beta_Q(\SSS_{i,j})\not=\emptyset$ by
  Claim~\ref{claim:R*:3}. As
  \begin{equation}
    q_4 \geq 2^{t_w}\cdot q_6,
  \end{equation}
  for some number $q_6$ determined below,
  there is some $\RRR_{i,j} \subseteq \RRR^j$ and
  some $\QQQ_{i,j}\subseteq \QQQ_4$ such that $|\QQQ_{i,j}|=q_6$ and
  $\beta_Q(\SSS_{i,j}) = \RRR_{i,j}$ for all $Q\in \QQQ_{i,j}$. As we
  require that
  \begin{eqnarray}
    t'_r & \geq & \Big(\binom{q_4}{q_6}\cdot 2^{t_w}\big)^{t'_c}\cdot
                  t''_r\\
    t'_c & \geq & \binom{q_4}{q_6}\cdot 2^{t_w}\cdot t''_c
  \end{eqnarray}
  there is a set $I''\subseteq I'$ with $|I''|=
  t''_r$ and a set $J''\subseteq J'$ with $|J'|=t''_c$ such that
  $\RRR_{i,j} = \RRR_{i', j}$, for all $i, i'\in I''$ and $j\in J''$ and
  $\QQQ_{i,j}=\QQQ_{i', j'}$ for all $i, i'\in
  I''$ and $j, j'\in J''$.

  We let $\QQQ'' := \QQQ_{i,j}$ for some (and hence all) $i\in
  I''$ and $j\in J''$ and set $\RRR'' := \bigcup_{j\in J''}\RRR_{i,j}$
  for some $i\in I''$.  We claim that if
  \begin{equation}
  t''_c = r \qquad\text{and}\qquad q_6 = q'
  \end{equation}
  then $\RRR''$ and $\QQQ''$ constitute the second outcome of the
  lemma. We have already argued above that for every $Q\in \QQQ''$ and every edge $e\in E(Q)\setminus E(\RRR'')$  there are at
  most $q^*$ paths from $Q_1$ to $Q_2$ in $\RRR''\cup \QQQ'' - e$,
  where $Q = Q_1\,e\,Q_2$.
  Hence, by our construction of $\RRR''$ and $\QQQ''$ and by Claim 3 it remains to define the edges $e_1(R), e_2(R)$ and
  $e(Q)$ for all $R\in \RRR''$ and $Q\in \QQQ''$.

  We require that
  \begin{equation}
     t_r'' \geq 3.
  \end{equation}
  Hence we can choose $i_1
  < i_2 < i_3\in I''$.  For $R\in \RRR''$
  choose $e_1(R)\in E(R)$ as the
  first edge on $R$ after the last vertex of $R$ in $\bigcup_{j\in J''}
  \RRR_{i_3}^j$ and $e_2(R)$ as the first edge on $R$ after
  $\bigcup_{j\in J''} \RRR_{i_2}^j$.  For every $Q\in \QQQ''$ let $e(Q)$
  be the first edge $e\in E(Q)\setminus E(\RRR'')$ on $Q$ after
  $v_{i_2}(Q)$ (which exists as $\RRR_3$ only
  goes up). Let
  $u(Q)$ and $l(Q)$ be the subpaths of $Q$ such that
  $Q=u(Q)\,e\,l(Q)$. Then,
  if $R = R_1\,e_1(R)\,R_2\,e_2(R)\,R_3 \in \RRR''$, then $R_1$ does not
  intersect any $u(Q)$, for $Q\in \QQQ''$ but $R_1$ intersects $l(Q)$
  and each of $R_2$ and $R_3$ intersect $u(Q)$.  Hence, this constitutes
  the second outcome of the lemma.
\end{proof}

\begin{recap}
Towards the end of the proof of Theorem~\ref{thm:main-cylindrical}, let us recall the
 current situation. After Lemma~\ref{lem:reroute-R*}, we have a linkage $\RRR''$ of order $r''$ and the
 linkage $\QQQ''$ of order $q''$ as in the statement of the lemma. In
 particular, for every $Q\in \QQQ''$ there is a split edge $e(Q) \in
 E(Q)\setminus E(\RRR'')$ splitting
 $Q$ into two subpaths
 $l(Q)$ and $u(Q)$ with $Q = u(Q)e(Q)l(Q)$.
Furthermore, for every $R\in \RRR''$ there are distinct edges $e_1(R), e_2(R)$ splitting $R$
 into subpaths $l(R), u_1(R)$ and $u_2(R)$ such that $R =
 l(R)\,e_1(R)\,u_1(R)\,e_2(R)\,u_2(R)$ and
 \begin{enumerate}
 \item the subpath $u_1(R)e_2(R)u_2(R)$ does not intersect
   $l(Q)$ for every $Q\in \QQQ''$
 \item $u_1(R)$ and $u_2(R)$ both intersect every $u(Q)$ for $Q\in \QQQ''$
 \item $l(R)$ intersects every $l(Q)$ for $Q\in \QQQ''$ (but may
   also intersect $u(Q)$).
 \end{enumerate}
 Finally, recal the definition of $q^*$ from Lemma~\ref{lem:reroute-R*}.
 	\end{recap}

 In the sequel we will impose various conditions
 on the size of the
 linkages we construct which will eventually
 determine the values of
 $p$ and $r$ in
 Theorem~\ref{thm:main-cylindrical}. We refrain
 from
 calculating these numbers precisely but rather
 state conditions on the
 size of the linkages. It is straightforward to
 verify that these conditions can always be
 satisfied.

We first need to introduce some notation. 
By construction, at most $q^*$ paths in $\RRR''$ can intersect $l(Q)$ after
it has intersected $u(Q)$.
We can therefore take a subset \marginpar{$\RRR^*,
r^*$}$\RRR^*\subseteq\RRR''$ of
order $r^* \geq r'' - q^*\cdot q''$ such that, for all $Q\in \QQQ''$, no path in $\RRR^*$ intersects $l(Q)$ after
it has intersected $u(Q)$.

For every $R \in \RRR^*$ let us define $i(R)$\marginpar{$i(R)$} to be the last vertex of $R$ in
$l(\QQQ'')$. Let $M'(R)$\marginpar{$M'(R)$} be the subpath of $R$ of
minimal length
which starts at the successor of $i(R)$ and which intersects every $u(Q)$ for
$Q\in\QQQ''$. %
Such a vertex $i(R)$
as well as the subpath $M'(R)$
exist by
construction, i.e.~by Property $(1){-}(3)$ above. See
Figure~\ref{fig:sec6-1a} a) for an illustration of the construction so far.

\begin{figure}[t]
  \centering
  \begin{tabular}{c@{\hspace*{1cm}}c}
  \parbox[m]{6cm}{\includegraphics[height=8cm]{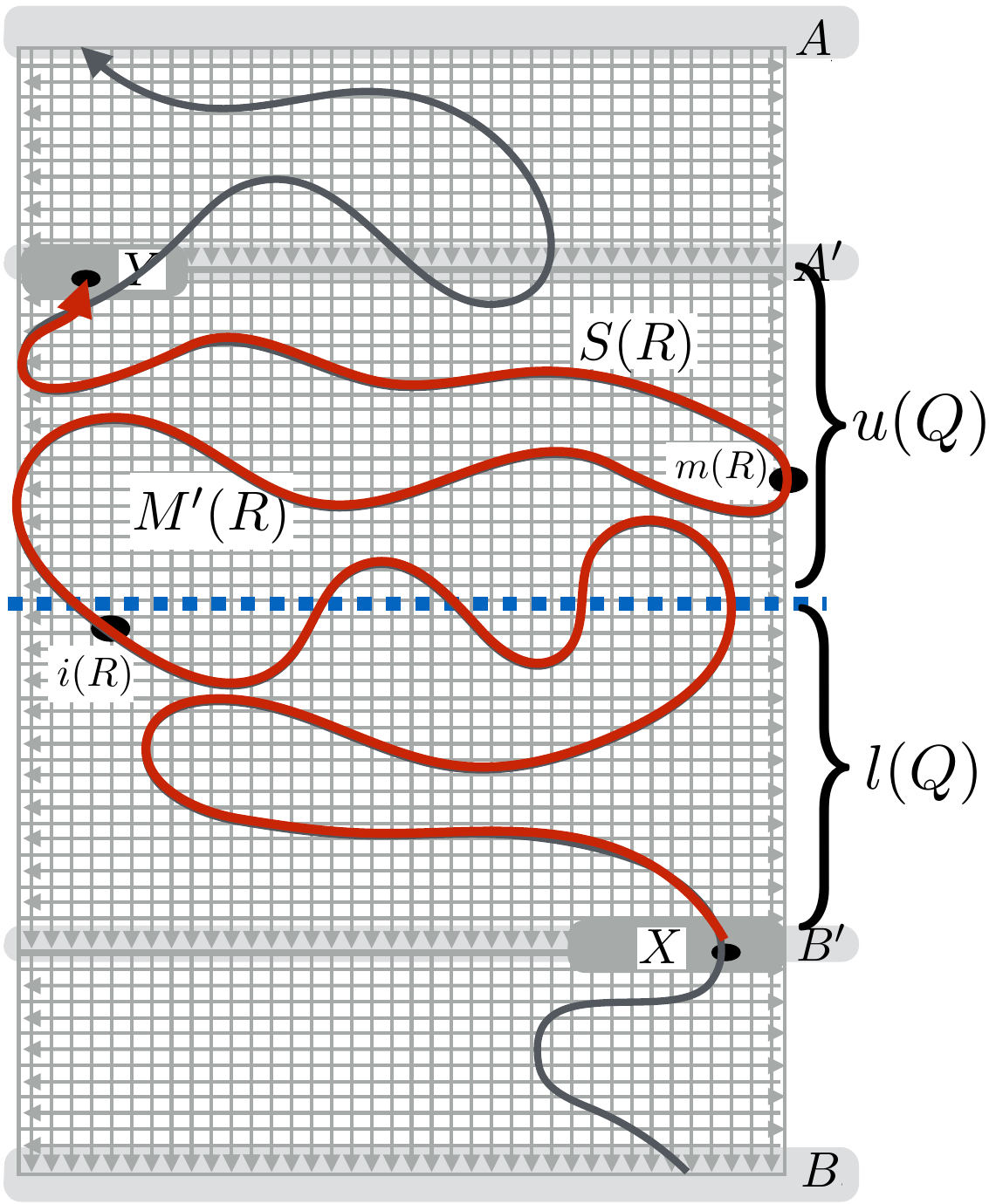}}&
  \parbox[m]{5cm}{\includegraphics[height=5cm]{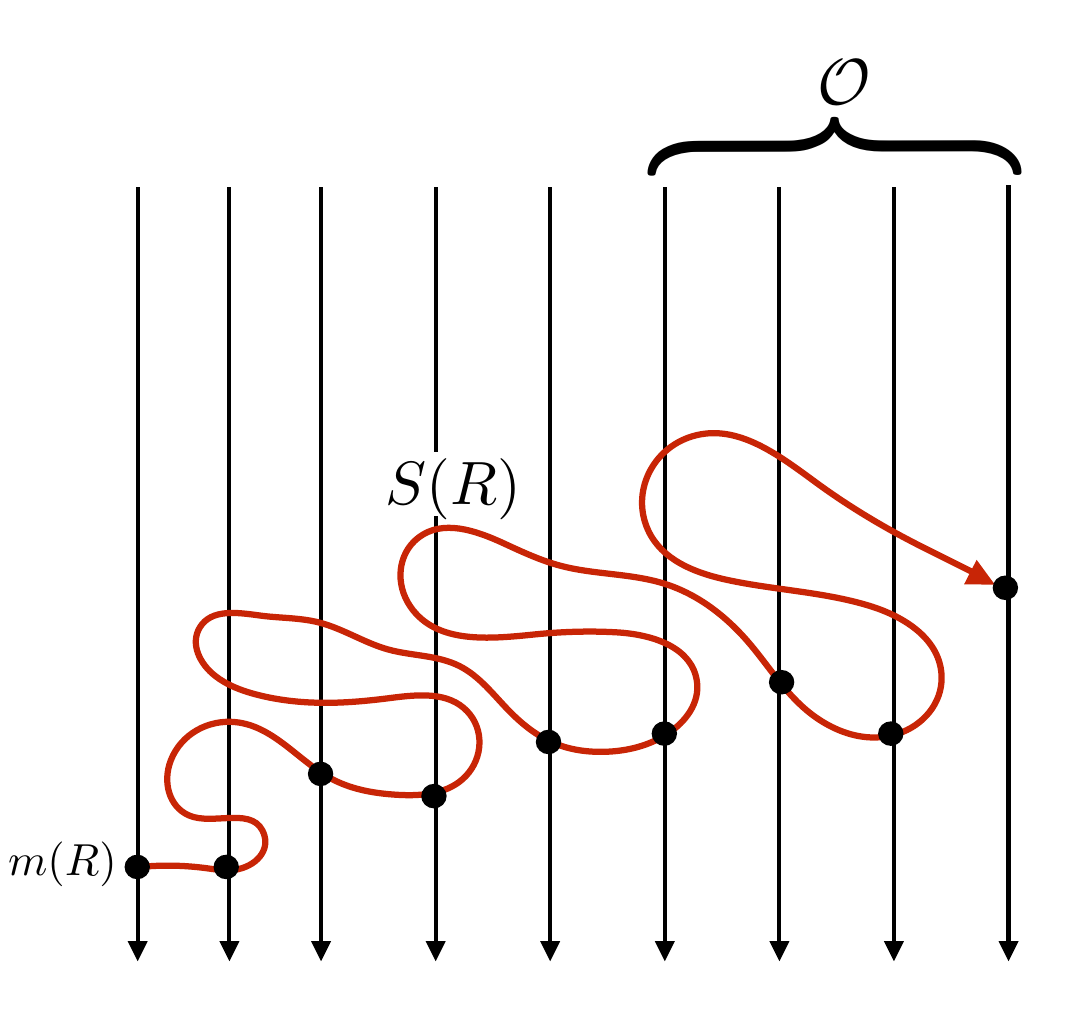}}\\
  a) & b)
\end{tabular}
  \caption{Illustration of a) the final part of the proof and b) of the construction of
    the set $\OOO$. }
  \label{fig:sec6-1a}
  \label{fig:sec6-1b}
\end{figure}

 By the pigeon hole principle and as we require
\begin{equation}
  \label{cond:1}
  r^* \geq r_1\cdot q'
\end{equation}
there is a set $\RRR_1\subseteq \RRR^*$\marginpar{$\RRR_1, r_1$} of order
$r_1$ and a path $Q^i\in \QQQ''$ such that
$i(R)\in V(Q^i)$ for all $R\in \RRR_1$.

For every $R\in \RRR_1$ and every $v\in V(R)\cap
l(\hat{Q})$ for some
$\hat{Q}\in\QQQ''$ let $<^R_v$\marginpar{$<_v$} be the order on
$\QQQ''$ in which  $Q <^R_v Q'$ if on
the subpath $R'$ of $R$ starting at $v$ to the end of $R$, the first
vertex that $R'$ has in common with $u(Q)$ appears before the first vertex
$R'$ has in common with $u(Q')$. As $v$ uniquely determines the path $R$ we drop the extra index $R$ and just write $<_v$. For any such $v$ let $(Q'_1, \dots,
Q'_{q''})$ be the paths in $\QQQ''$ ordered with respect to $<_v$ and define $\omitQ(v)
:= \{ Q'_{{q''}-t}, \dots, Q'_{q''}\}$\marginpar{$\omitQ(v), t$}, for some suitable number $t$
determined below. We call $t$ the \emph{omission width}.
A vertex $v\in V(R)\cap l(Q)$ for some
$Q\in\QQQ''$ is \emph{good}, if $Q\not\in
\omitQ(v)$. The following lemma explains the importance of a good vertex.

\begin{lemma}\label{lem:good-vertex}
  For every $R\in \RRR_1$ there is a path $Q\in \QQQ''$ such that $R$
  contains a good vertex $v(R)\in V(l(Q))$ and the subpath of
  $R$ from the beginning of $R$ to $v(R)$ intersects
  $l(Q)$ for all $Q\in\QQQ''\setminus \omitQ(v(R))$.
\end{lemma}
\begin{proof}
  We give a constructive proof of this lemma.
  For $0\leq i \leq t$, where $t$ is the omission width, we construct a set $\OOO_i\subseteq\QQQ''$, a path
  $Q_i\in
  \QQQ''\setminus \OOO_i$ and a vertex $v_i\in V(R)\cap l(Q_i)$ such
  that
  \begin{enumerate}
  \item $\OOO_i\subseteq \omitQ(v_i)$,
  \item the subpath of $R$
    from the beginning of $R$ to $v_i$ intersects $l(Q)$ for all $Q\in
    \QQQ''\setminus \OOO_i$, and
  \item the subpath of $R$
    from $v_i$ to $i(R)$ intersects $l(Q)$ for all $Q\in \OOO_i$.
  \end{enumerate}

  Let $\OOO_0 := \emptyset$. Let $Q_0 \in \QQQ''$ be the path containing
  the last vertex $v_0 := i(R)$ that $R$ has in common
  with $l(Q)$ for any $Q\in \QQQ''$.
  Clearly this satisfies Property (1)-(3) above.

  So suppose $\OOO_i, Q_i,
  v_i$ have already been constructed.
  If $Q_i\not\in \omitQ(v_i)$ then $v_i$ is good and we are done.
  Otherwise, $Q_i \in \omitQ(v_i)$ and we set $\OOO_{i+1} := \OOO_i\cup \{ Q_i\}$. Let $v_{i+1}$
  be the last vertex that $R$ has in common with $l(Q)$ for any $Q\in
  \QQQ''\setminus \OOO_{i+1}$. Let $Q_{i+1} \in
  \QQQ''\setminus \OOO_{i+1}$ be the path containing
  $v_{i+1}$. By construction, $Q_{i+1}$ is the last path on $R$ before
  $v_i$ such that $R$ intersects $l(Q_{i+1})$ and such that
  $Q_{i+1}\not\in \OOO_{i+1}$. We claim that
  $\OOO_{i+1}\subseteq \omitQ(v_{i+1})$. By induction hypothesis,
  $\OOO_i\subseteq \omitQ(v_i)$. Hence, in the order $<_{v_i}$, every
  path in $\OOO_{i+1}$ was among the last $t$ paths with respect to
  $<_{v_i}$. Now suppose some $Q\in \OOO_{i+1}$ is not in
  $\omitQ(v_{i+1})$. This means that $Q$ is no longer among the last
  $t$ paths hit by $R$ with respect to $<_{v_{i+1}}$. The only reason
  for this to happen is that the subpath of $R$ from $v_{i+1}$ to
  $v_i$ intersects $u(Q)$. But, by Property (3) above, the subpath of
  $R$ from $v_{i}$ to $i(R)$ intersects $l(Q)$. But this violates the
  construction of $\RRR^*$ as in $\RRR^*$, no path $R'\in
  \RRR^*$ intersects any $l(Q)$ after it has intersected
  $u(Q)$. Hence, the subpath of $R$ between $v_{i+1}$ and $v_i$ cannot
  intersect $u(Q)$ and therefore $\OOO_{i+1}\subseteq
  \omitQ(v_{i+1})$ as required. The other conditions are obviously
  satisfied as well.

  This concludes the construction of $\OOO_i, v_i, Q_i$ for all
  $i$. By construction, in every step $i$ in which no good vertex  is
  found (i.e., $Q_i\not\in \omitQ(v_i)$), the set $\OOO_i$ increases. However, as $\OOO_i\subseteq
  \omitQ(v_i)$ and $|\omitQ(v_i)| \leq t$ by definition, this process
  must terminate after at most $j\leq t$ iterations. Hence, $v_j$ is a
  good vertex.
\end{proof}

We require
\begin{equation}
  \label{eq:6}
  r_1 \geq r_2\cdot q''\cdot {q''\choose t}
\end{equation}
so that the previous lemma implies the
next corollary.

\begin{corollary}\label{cor:R2}
  There is a set $\RRR_2\subseteq \RRR_1$ of
  order $r_2$, a set $\OOO_1\subseteq \QQQ''$
  of order $t$ and a path
  $Q\in \QQQ''\setminus \OOO_1$ such that every $R\in \RRR_2$
  contains a good vertex $v(R) \in V(Q)$ satisfying the condition in
  Lemma~\ref{lem:good-vertex}  and $\omitQ(v(R)) =
  \OOO_1$.
\end{corollary}

The point of  $\omitQ(v(R)) = \OOO_1$ is that, as indicated in Figure \ref{fig:sec6-1b} b), we get some ``buffer'' area $\OOO_1$. More precisely, 
we want to apply Lemma~\ref{lem:split-grid-segmentation-refined} 
to $\QQQ''\setminus \OOO_1$ and a subset of $\RRR_2$, and then it is important for us that the outcome of  Lemma~\ref{lem:split-grid-segmentation-refined} does not involve the ''buffer'' area, because this area is needed to connect the ``bottom'' of 
the outcome to the top. This way, we shall construct a a cylindrical grid of order $k$ as a butterfly minor. 

Let us give one more remark: Hereafter, 
if we obtain a large ``split'' when applying Lemma~\ref{lem:split-grid-segmentation-refined}, rather quickly, we can  construct  a cylindrical grid of order $k$ as a butterfly minor. On the other hand, if the outcome is a large segmentation, the situation is more complicated, and almost the entire rest of the paper is devoted to this particular case.  Let us look at more details.

\begin{definition}\label{def:v(R)}
	For every $R\in \RRR_2$
let $v(R)$\marginpar{$v(R)$} be the good vertex as defined in the previous corollary. We
define $M(R)$\marginpar{$M(R)$} to be the subpath of $R$ of minimal length starting at the successor of $v(R)$ on
$R$ so that $M(R)$ intersects every $u(Q)$ for all $Q\in
\QQQ''\setminus \OOO_1$. Let $m(R)$\marginpar{$m(R)$} be the endpoint of $M(R)$.
We define $S(R)$\marginpar{$S(R)$} to be the subpath of $R$ of
minimal length
starting at the successor of $m(R)$ on $R$ such that
$S(R)$ intersects $u(Q)$ for all $Q\in \QQQ''$. Finally, we define
$I(R)$\marginpar{$I(R)$} to be the initial subpath of $R$ from its beginning to $v(R)$ (see Figure~\ref{fig:sec6-1a} a)).
\end{definition}
By construction, $I(R)$ intersects $l(Q)$ for all $Q\in
\QQQ''\setminus \OOO_1$, where $\OOO_1$ is as in the
previous corollary.

For every $R\in \RRR_2$ let $<^S_R$ be the order
on $\QQQ''$ in which $Q
<^S_R Q'$ if the first vertex $S(R)$ has in common with $Q$
appears on $S(R)$ before the first vertex $S(R)$ has in common with
$Q'$. By the pigeon hole principle and as we require
\begin{equation}
  \label{cond:2}
  r_2 \geq r_3\cdot (q'')! \cdot q''
\end{equation}
we can
choose a subset $\RRR_3 \subseteq \RRR_2$\marginpar{$\RRR_3 \subseteq \RRR_2$} of order $r_3$ such that $<^S_R = <^S_{R'}$
for all $R, R'\in \RRR_3$ and $v(R) = v(R')$ for all $R, R' \in \RRR_3$. Let $<^S := <^S_R$ for some (and hence all)
$R\in \RRR_3$.

Let $Q_1, \dots, Q_{q''}$ be the paths in $\QQQ''$ ordered by $<^S$ and
let $\OOO := \{ Q_{{q''}-t}, \dots, Q_{q''}\}$,\marginpar{$\OOO$} where $t$ is the omission width.
We write $M(\RRR_3) := \{ M(R) \sth R\in \RRR_3\}$.

We require that $|\QQQ''\setminus \OOO_1| = q''-t
\geq 2\cdot \max \{ q_s, q_1\}$, for suitable numbers $q_s, q_1$ determined below, and
that $r_3$ is such that \marginpar{$q_s, q_1$}
\begin{equation}
  \label{eq:7}
  \text{\parbox{10cm}{if in Lemma~\ref{lem:split-or-segment} we take
    $p:=q''-t$, $q' := r_3$, $c:=q''$, $y := q_s$, $x := q_1$ and $q := r_5$ then there is a $(q_s, r_5)$-split
    or a  $(q_1, r_5)$-segmentation, }}
\end{equation}
for a suitable number $r_5$\marginpar{$r_5$} 
determined below.

Applying Lemma~\ref{lem:split-or-segment} to
$(\QQQ''\setminus \OOO_1, M(\RRR_3))$,  which has linkedness $q^*$, where $\QQQ''\setminus \OOO_1$
takes on the role of $\PPP$ and $M(\RRR_3)$ plays the role of $\QQQ$, we
either get
\begin{enumerate}
\item a $(q_s, r_5)$-split $(\QQQ_s, \RRR_5)$\marginpar{$\QQQ_s, \RRR_5$}  obtained from a single path $Q\in  \QQQ''\setminus
  \OOO_1$ which is split into $q_s$ subpaths, i.e. $Q = Q_1\cdot
  e_1\cdot Q_2 \dots e_{q_s-1} \cdot Q_{q_s}$, or 
\item we obtain a $(q_1, r_5)$-segmentation $(\QQQ_1,
\RRR_5)$ \marginpar{$\QQQ_1, \RRR_5$} consisting of a
  subset
  $\RRR_5\subseteq M(\RRR_3)$ of order $r_5$ and a set $\QQQ_1$ of
  order $q_1$ of
  subpaths of paths in $\QQQ''\setminus \OOO_1$ satisfying the extra conditions of
  Lemma~\ref{lem:split-grid-segmentation-refined}.
\end{enumerate}

In the first case, as mentioned above, we can get a cylindrical grid of order $k$ as a butterfly minor as
follows. 

\begin{lemma}
	If applying Lemma~\ref{lem:split-or-segment} to $(\QQQ'' \setminus \OOO_1, M(\RRR_3))$ yields a $(q_s, r_5)$-split $(\QQQ_s, \RRR_5)$, then $G$ contains a cylindrical grid of order $k$ as a butterfly minor.
\end{lemma}
\begin{proof}
	See
Figure~\ref{fig:sec6-3} for an illustration of the following
construction. 
Let 
$\RRR_4 = \{ R \in \RRR^* \sth M(R) \in \RRR_5 \}$
be a linkage of order $r_4 := r_5$. Hence, $\RRR_4$ is a linkage
from the bottom of the original fence $\FFF$ to its top and
$\RRR_4$ and $\QQQ_s$ form a pseudo-fence $\FFF_p$.

\begin{figure}[ht]
  \centering
  \includegraphics[height=9cm]{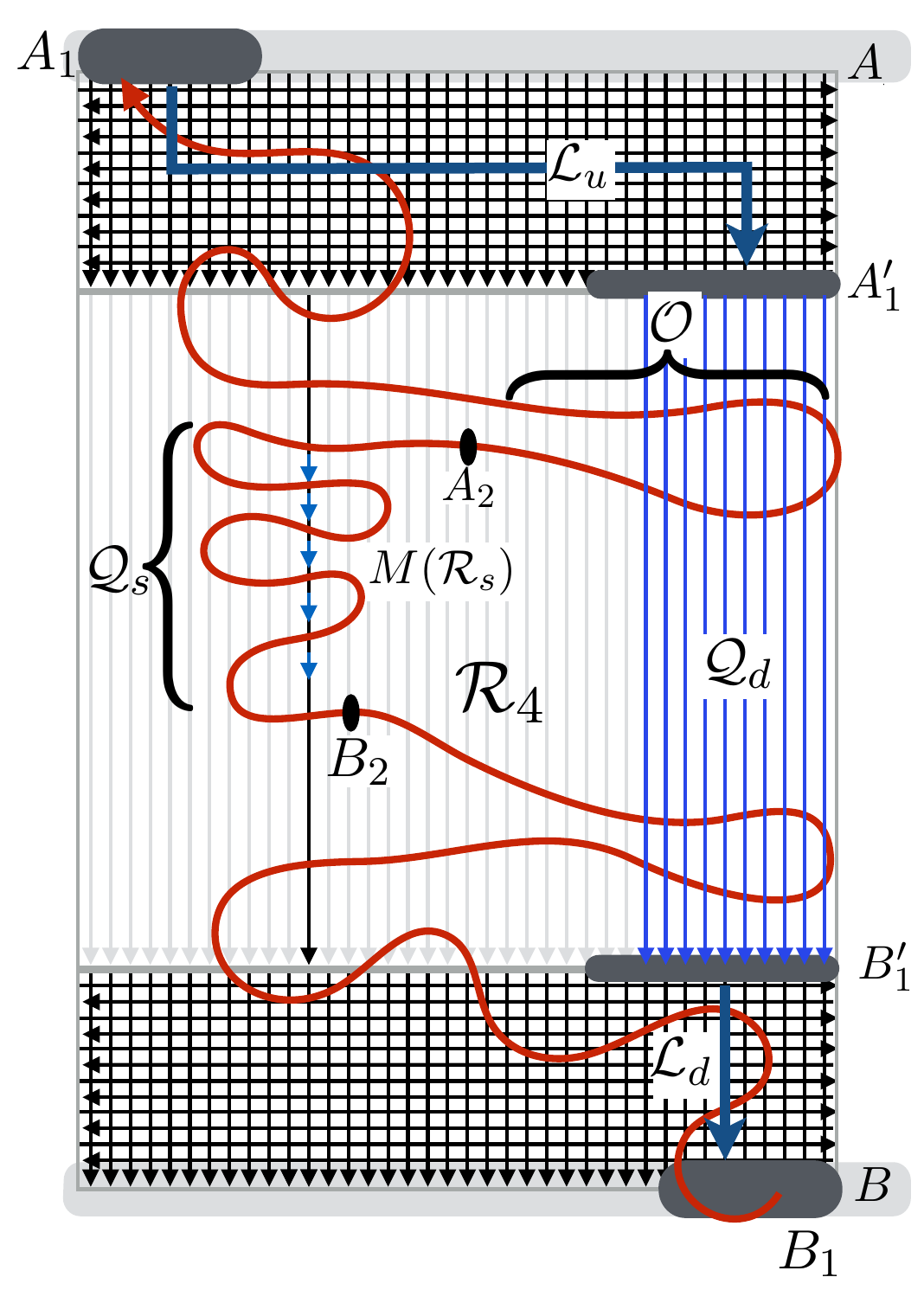}
  \caption{Creating a grid from a split and the resulting pseudo-fence.}
  \label{fig:sec6-3}
\end{figure}

We can now define paths back
from the end vertices of $\RRR_4$ to their start vertices as follows.
Let $A_1 \subseteq A$ be the endpoints of the paths in $\RRR_4$ and
let $B_1\subseteq B$ be the start vertices of the paths in $\RRR_4$.
Choose a set $\QQQ_d \subseteq \OOO_1$ of order $r_4$, which is possible as we require
\begin{eqnarray}
t &\geq& r_4.  %
\end{eqnarray}

Recall the definition of $A', B'$ from above and see Figure~\ref{fig:critical-schema} for an illustration. Let
$A_1' \subseteq A'$ and $B_1' \subseteq B'$ be the set of start and
end vertices of the paths in $\QQQ_d$.  Then
by routing through $\FFF_1$ and through $\FFF_3$ there is a linkage $\LLL_u$ of order $r_4$ from $A_1$ to $A'_1$ and a linkage
$\LLL_d$ of order $r_4$ from $B_1'$ to $B_1$. Hence, $\LLL_u\cup \QQQ_d\cup \LLL_d$
form a linkage $\LLL$ of order $r_4$ from $A_1$ to $B_1$.

Let
$B_2$ be the start vertices of the paths in $M(\RRR_4)$
and $A_2$ be their end vertices.
Every path $R \in
\RRR_4$ can be split into three disjoint subpaths, $D(R), M(R), U(R)$,
where $D(R)$ is the initial
component of $R-M(R)$ and $U(R)$ is the subpath following
$M(R)$. Then, $\LLL \cup \bigcup\{ U(R), D(R) \sth R\in \RRR_4\}$ forms a
half-integral linkage from $A_2$ to $B_2$ of order $r_4$ and hence, by
Lemma~\ref{lem:half-integral}, there is an integral linkage $\LLL'$ of
order $\frac12r_4$ from $A_2$ to $B_2$.

Note that $M(\RRR_4)$ and
$\LLL'$ are vertex disjoint (but $\QQQ_s$ may not be disjoint from $\LLL'$). %
We require that $\frac12r_4$ and $q_s$ are large enough so that we can apply 
Lemma~\ref{lem:ll2-general} to
$\LLL'$, $M(\RRR_4)$ and $\QQQ_s$  to obtain a cylindrical grid of order $k$ as a butterfly minor.
\end{proof}
\medskip

It remains to consider the second case above, i.e.~where we obtain a 
$(q_1, r_5)$-segmentation \marginpar{$\SSS_1$}$\SSS_1 := (\QQQ_1, \RRR_5)$. This case and part of the
following construction is illustrated in
Figure~\ref{fig:before-6.14}.
\begin{recap}
Let us recall the current situation and the notation still relevant for the remainder of the proof. 	

\begin{itemize}
\item 	$\FFF$ is the original fence and $\FFF_1, \FFF_3$ are its upper and lower part used for rerouting paths. 
\item $\QQQ''$ is the vertical linkage of order $q''$ in the "middle" of the fence $\FFF$ taken so that every $Q \in \QQQ''$ can be split into two parts $u(Q)$ and $l(Q)$ occurring in this order on $Q$.
\item $\RRR_3$ is the bottom-up linkage we work with. From the "middle paths" $\{ M(R) \sth R \in \RRR_3 \}$ we obtained the segmentation $\SSS_1$. 
\item $\OOO, \OOO_1 \subseteq \QQQ''$ are linkages where $\OOO_1$ is constructed in Corollary~\ref{cor:R2} and is the set of omitted paths for the paths $M(R)$, for $R \in \RRR_3$, and $\OOO$ plays a similar role for the paths in $\{ S(R) \sth R \in \RRR_3\}$.
\item 
Every $R \in \RRR_3$ contains a \emph{good vertex} $v(R)$ such that $\omitQ(v(R)) = \OOO_1$. Furthermore, $v(R) = v(R')$ for all $R, R' \in \RRR_3$. See Definition~\ref{def:v(R)}.
\item  We defined $M(R)$ as the subpath of $R$ of minimal length starting at the successor of $v(R)$ so that $M(R)$ intersects $u(Q)$ for all $Q \in \QQQ'' \setminus \OOO_1$. $m(R)$ was the end point of $M(R)$ and $S(R)$ is the minimal subpath of $R$ starting at the sucessor of $m(R)$ intersecting $u(Q)$ for all $Q \in \QQQ''$. Finally, $I(R)$ was the initial subpath of $R$ up to $v(R)$. Recall that by construction, $M(R) \cap l(Q) = \emptyset$ for all $R \in \RRR_3$ and $Q \in \QQQ''$. 
\item $\SSS_1 := (\QQQ_1, \RRR_5)$ is a  $(q_1, r_5)$-segmentation. Recall that the paths in $\RRR_5$ are paths $M(R)$ for some $R \in \RRR_3$. 
\item  Finally, recall that $q_s$ is the number such that when we applied Lemma~\ref{lem:split-or-segment} we either got a $(q_s, r_5)$-split or the segmentation $\SSS_1$.
\end{itemize}
\end{recap}

We define \marginpar{$\hat{\RRR}_5$} $\hat{\RRR}_5 := \{ R\in \RRR_3 \sth M(R)\in \RRR_5\}$.
Recall that when obtaining the segmentation $\SSS_1$, some paths in $\QQQ_1$ can
be obtained by splitting a single path in $\QQQ''\setminus \OOO_1$. However,
no path in $\QQQ''\setminus\OOO_1$ is split more than $q_s-1$
times. We define an equivalence relation $\sim$\marginpar{$\sim$} on $\QQQ_1$ by letting
$Q\sim Q'$ if $Q$ and $Q'$ are subpaths of the same path in
$\QQQ''$. Recall that
Lemma~\ref{lem:split-or-segment}
guarantees that either every or no vertex
of a path in $\QQQ''$ occurs on a path in $\QQQ_1$. As $M(\RRR^*)\cap
l(\QQQ'') = \emptyset$, it follows that in each equivalence class of $\sim$ there is
exactly one path containing a vertex in $l(\QQQ'')$.
Let \marginpar{$\QQQ'^l_1$}$\QQQ'^l_1$ be the set of paths in $\QQQ_1$ containing a vertex in
$l(\QQQ'')$. Hence, $(\QQQ'^l_1, \RRR_5)$ form a \marginpar{$q'^l_1$}$q'^l_1$-segmentation
of order $r_5$ for some $q'^l_1 \geq \frac{q_1}{q_s-1}$. See Figure~\ref{fig:before-6.14} for an illustration of the current situation. 

\begin{figure}
\includegraphics{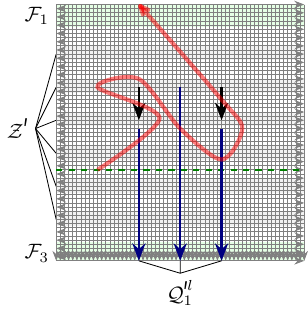}
\caption{The situation before Lemma~\ref{lem:6.14}. The blue paths are the paths in $\QQQ_1^{\prime l}$.}
\label{fig:before-6.14}	
\end{figure}

Our next goal is to show that if there is a set $\RR^S \subseteq S(\hat{\RRR}_5)$ of large order that avoids a big portion of the segmentation $(\QQQ'^l_1, \RRR_5)$ 
we can obtain a cylindrical grid of order $k$ as a butterfly minor. 
This will be shown in the following lemma. 

This lemma is technical, but once we prove it, we will end up with 
a situation where whenever we create another segmentation $S$ that is obtained from 
a subset of $S(\hat{\RRR}_5)$, every row in $S$ has to overlap with the segmentation $(\QQQ'^l_1, \RRR_5)$. This structure will be the key to complete our proof.

\begin{lemma}\label{lem:6.14}
  There are integers $o_l,
  r^s$ depending only on $k$ and $q_s$ such that if there is a set
  $\QQ\subseteq \QQQ'^l_1$ of order $o_l$ and
  a set $\RR^S \subseteq S(\hat{\RRR}_5)$ of
  order $r^s$ such that
  no path in $\RR^S$ intersects any path in $\QQ$, then $G$ contains a
  cylindrical grid of order $k$ as a butterfly minor.
\end{lemma}
\begin{proof}
    Recall the equivalence relation $\sim$ above. For every
  $Q\in \QQQ_1$ we denote the $\sim$-equivalence class of
  $Q$ by $[Q]_\sim$. By Lemma~\ref{lem:split-or-segment},  at most $q_s-1$
  paths belong to the same
  equivalence class. For $Q\in \QQQ_1$ we denote by
  $\sigma(Q) = (Q_1, \dots, Q_j)$\marginpar{$\sigma(Q)$}, for some $j  = j(Q)\leq q_s-1$, the sequence of paths
  $Q_1, \dots, Q_j\in \QQQ_1$ such that $[Q]_\sim = \{
  Q_1, \dots, Q_j\}$ and the paths $Q_1, \dots, Q_j$
  occur in the (reverse) order $Q_j, \dots, Q_1$ on the path $Q''\in
  \QQQ''$ that contains $Q$.

  If $\sigma(Q) = (Q_1, \dots, Q_j)$ for some $Q\in
  \QQQ_1$, then every
  path $R\in \RRR_5$ can be divided into $j$ segments $R_1,
  \dots, R_j$ occurring in this order on $R$ and $V(R)\cap
  V(Q_i)\subseteq V(R_i)$. In particular, $Q_1$ must be the
  path that contains a vertex in $l(\QQQ'')$. We call this
  the \emph{first} member of the equivalence class.

The idea of the proof of this lemma is that 
we are already given the segmentation $(\QQ, \RR_5)$, and all we need in order to find a cylindrical grid of order $k$ as a butterfly minor is to find a linkage of large size from the end of $\RR_5$ to the start by using the linkage $\RR^S$ without intersecting a big part of the segmentation.

  Towards this end, we need the following simple claim.  Note that the conditions on $q_s$
  imposed in the previous parts of the proof guarantee that  $q_s \geq 2$ so
  that $(q_s)^d \geq d$. 
  Let $\RR^M := \{ M(R) \sth S(R)\in \RR^S \}$.

  \begin{Claim}\label{claim:lem-6.14:1}
    Let $R \in \RR^M$.
    For every $d\geq 1$  if $\QQ_1 \subseteq \QQ$ is of order
    at least $(q_s)^{d}$ such that $R$ intersects every $Q\in
    \QQ_1$, then  there is a subset $\QQ(R)\subseteq
    \QQ_1$ of order $d$ and an initial subpath
    $R'$ of $R$ such that for every $Q \in \QQ(R)$, if $\sigma(Q) = (Q_1,
    \dots, Q_j)$ then there is an index $1\leq j'\leq j$ such
    that $R'$ intersects $Q_1, \dots, Q_{j'}$ but none of
    $Q_{j'+1}, \dots, Q_j$. Furthermore, for each $Q \in \QQ_1$ there is a subpath $R'' = R''(Q)$ which contains all of $V(Q) \cap V(R')$ but is disjoint from all other $Q \in \QQ_1$ with $Q \neq Q'$. 
  \end{Claim}
  \begin{ClaimProof}
    We prove the claim by induction on $d$. For $d=1$ we choose the
    first path
    $Q\in \QQ_1$ that $R$ intersects and set $\QQ(R)
    := (Q)$ and $R'$ as the initial subpath of $R$ of minimal length
    such that $R'$ intersects $Q$. As $R$ traverses $\sigma(Q) = (Q_1,
    \dots, Q_j)$ in the order $Q_1, \dots, Q_j$, we are guaranteed
    that $Q=Q_1$ and the conditions of the claim are met.

    Now let $d>1$. Let $Q\in \QQ_1$ be the first path in $\QQ_1$
    that $R$ intersects. Let $\sigma(Q)  = (Q_1, \dots, Q_j)$. Then
    $R$ can be split into $j$ subpaths $R_1, \dots, R_j$ where $R_i$
    is the maximal subpath of $R$ starting at the first vertex $R$ has
    in common with $Q_i$ and which does not include any vertex of
    $V(Q_{i+1})$, or to the end of $R$ in case $i=j$.

    For every $Q'\in\QQ_1\setminus \{ Q\}$ let $i(Q')$ be the minimal
    index such that $V(R_{i(Q')})\cap V(Q') \not= \emptyset$. By the
    pigeon hole principle there is an index $1\leq i \leq q_s-1$ such
    that $i = i(Q')$ for at least $\frac{|\QQ_1\setminus \{
      Q\}|}{q_s-1} \geq \frac{q_s^d-1}{q_s-1} \geq q_s^{d-1}$ paths
    $Q'\in \QQ_1$. Let $\tilde{Q}_1 := \{ Q' \in \QQ_1 \sth i =
    i(Q') \}$. 
    
    Let $\tilde{R}$ be the minimal initial subpath of $R$ including $R_i$. Then applying the construction inductively to  $\tilde{R}$ and $\tilde{Q}_1$ yields an initial subpath
    $\tilde{R}'$ of $\tilde{R}$ and a subset $\QQ(\tilde{R})\subseteq
    \tilde{Q}_1$ of order $d-1$ satisfying the conditions of the
   claim. But then, setting $\QQ(R) := \QQ(\tilde{R})\cup \{ Q \}$,
   taking $R' = \tilde{R}'$ satisfies the claim.
  \end{ClaimProof}

  We require that
  \begin{eqnarray}
    \label{eq:22}
    o_l & \geq & (q_s)^{o'}\\
    r^s & \geq & \binom{o_l}{o'}\cdot p',
  \end{eqnarray}
  for some suitable numbers $o', p'$ to be determined below.
  Then, by Claim~\ref{claim:lem-6.14:1}, there are sets $\QQ_1
  \subseteq \QQ$\marginpar{$\QQ_1$} of order $o'$ and
  $\bar\RR' \subseteq \RR^M$ of order $p'$ so that $\QQ_1 =
  \QQ(R)$ for every path $R \in \bar\RR'$.
  Furthermore, for every $R\in \bar\RR'$  let
  $R'$ be the initial subpath satisfying the condition of the
  claim and let $\hat{R} \in \RRR''$ be the path of which $R$ is a
  subpath. We define $M'(\hat{R}) := R'$ and  define $S'(\hat{R})$ to  be the
  subpath of $R$ starting at the successor of the endpoint of $M'(\hat{R})$
  and ending at the last vertex of $S(\hat{R})$. Hence, $M'(\hat{R})$
  and $S'(\hat{R})$ are obtained from $M(\hat{R})$ and $S(\hat{R})$ by
  shortening $M(\hat{R})$ and adding the removed part of $M(\hat{R})$
  to the beginning of $S(\hat{R})$ to obtain $S'(\hat{R})$.

  Let $\RR^{M'} := \{ M'(R) \sth R\in \bar\RR'\}$\marginpar{$\RR^{M'}$} and let $\RR^{S'} :=
  \{ S'(\hat{R}) \sth M'(\hat{R}) \in \RR^{M'}\}$\marginpar{$\RR^{S'}$}. Let $f_r, f_p \sth \N \rightarrow \N$ be the functions defined in
  Lemma~\ref{lem:new-good-tuples}.

Observe that  $(\QQ_1, \RR^{M'})$ forms an  $(o', p')$-segmentation. 
In particular, Claim~\ref{claim:lem-6.14:1} implies that the paths in $\RRR^{M'}$ go through the paths in $\QQ_1$ as indicated in Figure~\ref{fig:lem-6-14}.

We require that $p'\geq p''\cdot f_q(o)$ and
\begin{equation}
  \label{eq:13}
  o' \geq (p''\cdot f_q(o))!\cdot f_p(o)
\end{equation}
for some suitable values of $p''$ and $o$ to be determined below. 
Let $\tilde\RR^{M'}$\marginpar{$\tilde\RR^{M'}$} be a subset of $\RR^{M'}$ of order $p''\cdot f_q(o)$.
Then
$(\QQ_1, \tilde\RR^{M'})$ forms an $(o', p'')$-segmentation and
therefore, by Lemma~\ref{lem:strong-segmentation} Part (1), there is a
subset $\QQ_2\subseteq \QQ_1$\marginpar{$\QQ_2$} of order $f_q(o)$ such that 
$(\QQ_2,
\tilde\RR^{M'})$ forms a strong $(f_p(o), p''\cdot f_q(o))$-segmentation. Let
$(R_1, \dots, R_{p''\cdot f_q(o)})$ be the order in which the paths in
$\tilde\RR^{M'}$ occur on the paths in $\QQ_2$.
For each $1\leq i \leq p''$ we can now apply
Lemma~\ref{lem:new-good-tuples} to $(\QQ_2, \{ R_{(i-1)f_q(o)+1},
\dots, R_{if_q(o)}\})$ to obtain a sequence $\bar Q_i := (Q^i_1, \dots,
Q^i_{o})$ and a path $A_i$ as in the statement of the
lemma. Furthermore, we can choose the path $A_i$ so that it
satisfies Property~$2$ of the lemma, i.e.~so that it ends at an endpoint of a path in $\{ R_{(i-1)f_q(o)+1},
\dots, R_{if_q(o)}\})$.

We require that
\begin{equation}
  \label{eq:15}
  p'' \geq \binom{f_p(o)}{o}\cdot o! \cdot h_2,
\end{equation}
for some suitable number $h_2$ to be determined below.
Then there are $h_2$ values $i_1< \dots <  i_{h_2}$ such that  $\bar
Q_{i_j} = \bar Q_{i_{j'}}$ for all $1\leq j, j'\leq h_2$. Let $(Q_1,
\dots, Q_o) = \bar Q_{i_1}$ and, for $1\leq j \leq h_2$, let $H_j :=
A_{i_j}$. Then $((H_1, \dots, H_{h_2}), (Q_1, \dots, Q_o))$ form a
grid with the paths $H_1, \dots, H_{h_2}$ occurring in this order from
top to bottom on the paths $Q_1, \dots, Q_o$.

Note that by the construction in Lemma~\ref{lem:new-good-tuples} every
$H_i$ has the following property. 
\begin{quote}
For any $Q_i$ with $1\leq i \leq o$ let $\sigma(Q_i) = (Q^1, \dots, Q^j)$ and let
$j'$ be maximal such that $H_i$ intersects $Q^{j'}$. Then
$H_i$ starts at $Q^1$, then intersects $Q^{j''}$
for all $j''\leq j'$, but once it has intersected $Q^{j'}$
it will never again intersect any $Q^{j''}$ for some
$j''< j'$.
\end{quote}
\begin{figure}
      \includegraphics{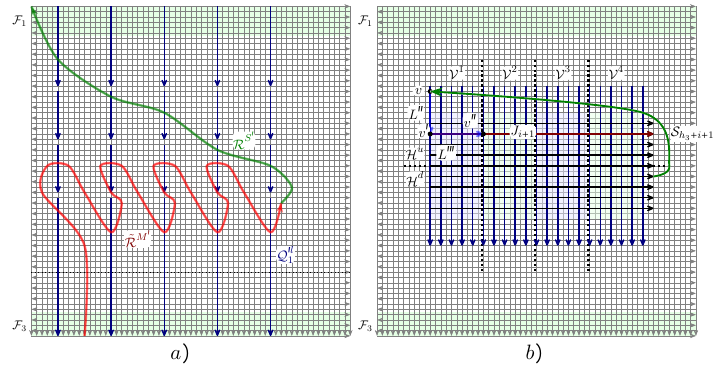}
\caption{Illustration of the construction in Lemma~\ref{lem:6.14} }
\label{fig:lem-6-14}
\end{figure}
This property is needed below. See Figure~\ref{fig:lem-6-14} a) for an illustration of the current situation.
  We require that
  \begin{eqnarray}
    \label{eq:21}
    h_2 &=& h_3 + k', \\
    o & = & 4k'
  \end{eqnarray}
  for suitable numbers $h_3, k'$ determined below.
  Let $\HHH^u := (H_1, \dots, H_{h_3})$\marginpar{$\HHH^u$} and let $\HHH_d
  := (H_{h_3+1}, \dots, H_{h_3+k'})$\marginpar{$\HHH_d$}.  For $1\leq i \leq 4$ let
  $\VVV^i := (Q_{(i-1)k'+1}, \dots, Q_{i\cdot k'})$\marginpar{$\VVV^i$}.

  We will now construct a cylindrical grid of order $k$ as a butterfly minor as
  follows.  By construction, every path $H_i$ can be
  extended by a path
  $R$ for some $R\in \RR^{S'}$. Let $S_{h_3+1}, \dots,
  S_{h_2}\in \RR^{S'}$ be the paths such that $S_i$ extends $H_i$, for
  all $h_3+1\leq i \leq h_2$.
  As the paths $S_i$ do not intersect any $Q \in \QQ_1$
  but do intersect $u(Q)$ for every $Q\in \QQQ''$ this means that for
  every $1\leq i \leq o$ every $S_i$ intersects a path
  $Q_j^u\in \QQQ_1$ such that $Q_j$ and $Q_j^u$ are
  subpaths
  of the same path $\hat{Q}_j \in \QQQ''$. Furthermore,
  by construction of $\QQ_1$, $Q_j^u$ must occur on
  $\hat{Q}_j$ before $Q_j$.

  For $0\leq i \leq k'$ we inductively construct
  a set  $\III_i\subseteq \HHH^u$ and for $1\leq i \leq k'$ a path $J_i \in \HHH^u$
  and a path $L_i$ with the
  following properties. $L_i$ has as
  first vertex the last vertex of $H_{h_3+i}$ and as last vertex the
  first vertex on $J_i$ that $J_i$ has in common with any path in
  $\VVV^2\cup \VVV^3 \cup \VVV^4$ (when traversing $J_i$ from beginning
  to end).
  Furthermore, all $J_i$ occur in the grid $\HHH$ ``higher up'' than
  every path in $\III_i$, i.e.~on the paths in $\VVV$, the paths $J_1,
  \dots, J_i$ occur before the paths in $\III_i$. Finally, $\{ L_1,
  \dots, L_i\}$ is a half-integral linkage. Note that $V(J_i) \cap V(L_i) \not=\emptyset$. 

  Initially we set $\III_0 := \HHH^u$ which obviously
  satisfies the condition.

  Now suppose $0 \leq i < k'$ and $\III_i$, and if $i\geq 1$ also $J_1, \dots,
  J_{i}$ and $L_1, \dots, L_{i}$, have
  already been defined satisfying the conditions above.

  Consider the
  initial subpath of $S_{h_3+i+1}$ which ends at the
  first vertex $v$ that $S_{h_3+i+1}$ has in common with
  $\hat{Q}_{i+1}$. Let $\sigma(\hat{Q}_{i+1}) = (Q^{i+1}_1, \dots,
  Q^{i+1}_{l_{i+1}})$ and let $l$ be the maximal index such that
  the paths $R\in \bar{\RR}'$ intersect $Q^{i+1}_l$.
  Let $j$ be the index such that $v\in V(Q_j^{i+1})$.

  \begin{itemize}
  \item If $j<l$ then let $J_{i+1} =
  H_m$ for the  minimal $m$ such that $H_m\in \III_i$, i.e.~$J_{i+1}$ is the highest path in $\III_i$. We set
    $\III_{i+1} = \III_i\setminus \{ J_{i+1}  \}$.
    We construct $L_{i+1}$ as follows. As $S_{h_3+i+1}$ does not
    intersect $Q^{i+1}_1$, it follows that $j>1$.
    Let $L'$ be the
    initial subpath of $S_{h_3+i+1}$ up to $v$ and let $L''$ be the
    subpath of $\hat{Q}_{i+1}$ from $v$ downwards to a vertex
    $v' \in V(J_{i+1})\cap V(Q^{i+1}_1)$. Finally, let $L'''$ be the
    subpath of $J_{i+1}$ from $v'$ to the first vertex $v''$ that $J_{i+1}$ has
    in common with any path in $\VVV^2 \cup \VVV^3 \cup \VVV^4$. Then
    $L'\cup L''\cup L'''$ contains a path $L_{i+1}$ from the endpoint of
    $H_{h_3+i+1}$ to $v''$. Hence, $\III_{i+1}, J_{i+1}$ and $L_{i+1}$
    satisfy the conditions above. 
  \item If $j>l$ then again let $J_{i+1}$ be the highest path in $\III_i$,
    i.e.~let $J_{i+1} = H_m$ for the  minimal $m$ such that $H_m\in
    \III_i$. We set $\III_{i+1} = \III_i\setminus \{ J_{i+1}
    \}$ and construct $L_{i+1}$ as follows. Let $L'$ be the initial
    subpath of $S_{h_3+i+1}$ that ends at the  vertex $v$
    and let $L''$ be a subpath of
    $\hat{Q}_{i+1}$ of minimal length from $v$ to a vertex  $v' \in
    V(J_{i+1})\cap V(Q^{i+1}_l)$. Finally, let $L'''$ be the
    subpath of $J_{i+1}$ from $v'$ to the first vertex $v''$ $J_{i+1}$ has
    in common with any path in $\VVV^2 \cup \VVV^3 \cup \VVV^4$. Then
    $L'\cup L''\cup L'''$ contains a path $L_{i+1}$ from the endpoint of
    $S_{h_3+i+1}$ to $v''$. Hence, $\III_{i+1}, J_{i+1}$ and $L_{i+1}$
    satisfy the conditions above. See Figure~\ref{fig:lem-6-14} b) for an illustration.
\item
    Finally, suppose $j=l$. Suppose first that at least half of the paths in $\III_i$
    occur on $Q^i_l$ before $v$. We set $\III'$ to be these
    paths. Let $J_{i+1}$ be the highest path in $\III'$ and set
    $\III_{i+1} := \III'\setminus \{ J_{i+1}\}$.
    We then construct the path $L_{i+1}$ as in the first case above.

    Otherwise, if more than half of the paths in $\III_i$ occur lower than
    $v$ then let $\III'$ be
    the paths in $\III_{i}$ below $v$  and choose the
    highest path $J_{i+1}\in \III'$. We set $\III_{i+1} :=
    \III'\setminus \{ J_{i+1} \}$ and proceed as in the second case
    above to construct the path $L_{i+1}$.
  \end{itemize}
  This completes the construction of $\III_i, J_i$ and $L_i$.
  We require
  \begin{equation}
    \label{eq:23}
    k' \geq 6\hat{k},
  \end{equation}
  where $\hat{k}$ is the number defined in
  Lemma~\ref{tech00} (called $t'$ in the statement of the lemma).
  We now choose a subgrid $\UUU$ in $\III_{k'}$ and $\VVV_3$ of
  order $6\hat{k}$.
  Let $H^u_1, \dots, H^u_{k'}$ be the horizontal paths in $\UUU$
  ordered from top to bottom. Let $a_i$ and $b_i$ be the start and end
  vertex of $H^u_i$, respectively.

  Then, we use the paths $J_1, \dots, J_{k'}, L_1, \dots, L_{k'}$, the
  subgrid of $\HHH$ restricted to the parts in $\VVV^4$ and the
  subgrid of $\HHH$ restricted to the parts in $\VVV^2$ to construct a
  half-integral linkage $\LLL'$ from $B := \{b_{\frac23k'+1}, \dots, b_{k'}\}$ to
  $A := \{ a_1, \dots, a_{\frac13k'}\}$. By Lemma~\ref{lem:half-integral}
  there also exists an integral linkage $\LLL$ of order
  $\frac12|\LLL'| = 3\hat{k}$ from $B$ to $A$. Thus, we can apply
  Lemma~\ref{tech00} to obtain a cylindrical grid of order $k$ as a butterfly minor, as requested.
\end{proof}

By the previous lemma, if
\begin{eqnarray}
\label{eq:17c}
r_5 & \geq & (q'^l_1)^{o_l}\cdot ( \hat r'_5 + r^s) \\
q'^l_1 & \geq & \hat q^l_1 - o_l,
\end{eqnarray}
for some suitable numbers $\hat r'_5$ and $\hat q_1^l$ to be determined below,
then either we get a cylindrical grid of order $k$, in which case we
are done, or
we can assume that there is a set
$\hat{R}'_5 \subseteq \hat{\RRR}_5$ of order $\hat r'_5$ and a set
$\hat\QQQ^l_1\subseteq \QQQ'^l_1$ of order $\hat q^l_1$, such that for every $R\in \hat{R}'_5$, the
subpath $S(R)$ intersects every path in $\hat\QQQ^l_1$. Furthermore, the
pair $(\hat\QQQ_1^l, M(\hat{R}'_5))$ still forms a segmentation.

Now, if 
\begin{eqnarray}
  \label{eq:17new}
  \hat q^l_1 & \geq & \hat r'_5!\cdot q^l_1,
\end{eqnarray}
for some suitable number $q^l_1$ determined below, then we can apply
Lemma~\ref{lem:strong-segmentation} Part (1) to get a set $\QQQ_1^l\subseteq
\hat\QQQ_1^l$ of order $q_1^l$ such that $(\QQQ_1^l, M(\hat{R}'_5))$
forms a strong segmentation.
\color{black}

We are now ready to complete the proof. 
In the rest, the key idea is to apply the same construction to $(\QQQ^l_1\setminus \OOO,
S(\hat{R}'_5))$ (which again has linkedness $q^*$).
Again, if we obtain the split case, then we can finish the proof fairly quickly. So we are left with the segmentation $\SSS$ as outcome of Corollary~\ref{cor:split-strong-segmentation} applied 
to $(\QQQ^l_1\setminus \OOO, S(\hat{R}'_5))$.

Now we look at the intersection of $\SSS$ and $(\QQQ_1^l, M(\hat{R}'_5))$. 
If some part of $\SSS$ is ``separated'' from $(\QQQ_1^l, M(\hat{R}'_5))$, 
the situation is described as in Figure \ref{fig:sec6-5} (a), and  
we obtain a cylindrical grid of order $k$ as a butterfly minor. 
So $\SSS$ is indeed ``included'' in $(\QQQ_1^l, M(\hat{R}'_5))$, 
as in Figure \ref{fig:sec6-5} (b). 
We shall see how to handle this case below but before, we 
we first apply the same construction to $(\QQQ^l_1\setminus \OOO,
S(\hat{R}'_5))$ (which again has linkedness $q^*$). 

\medskip

We require that $r'_5$ is large enough so that if in
Corollary~\ref{cor:split-strong-segmentation} \color{black}
we set $p$ to $q_1^l$,
$y$ to $q_s$, $x$ to $q_5$, $c$ to $q^*$ an $q$ to $q_7$, for suitable
values of $q_s, q_5, r_7$ to be determined below, then $r'_5$ is
larger than the number $q'$ specified in the lemma.
We can then apply 
Corollary~\ref{cor:split-strong-segmentation} to $(\QQQ_1^l \setminus \OOO, S(\hat{R}'_5))$
  and
either we
get %
\begin{enumerate}
\item a $(q_s, r_7)$-split $(\QQQ_s, \RRR^S_7)$ obtained from a
  single path $Q\in
  \QQQ_1^l\setminus \OOO$ which is split into $q_s$ subpaths, i.e. $Q
  = Q_1\cdot e_1\cdot Q_2 \dots e_{q_s-1} \cdot Q_{q_s}$, or
\item we obtain a strong $(q_5, r_7)$-segmentation $(\QQQ_5, \RRR^S_7)$ \marginpar{$(\QQQ_5, \RRR^S_7)$} 
  defined by a subset
  $\RRR^S_7\subseteq S(\hat{R}'_5)$ of order $r_7$ and a set $\QQQ_5$ of
  order $q_5$ of
  subpaths of paths in $\QQQ_1^l\setminus \OOO$ satisfying the extra
  conditions of 
Corollary~\ref{cor:split-strong-segmentation}.
\end{enumerate}

In the first case, we can get a cylindrical grid of order $t$ as a butterfly minor as
before. Indeed, as before, whenever we get a spilt of large order, we are done. For completeness, we give a proof here. 

\begin{lemma}
  If applying Corollary~\ref{cor:split-strong-segmentation} to $(\QQQ_1^l \setminus \OOO, S(\hat{R}'_5))$
  yields a $(q_s, r_7)$-split $(\QQQ_s, \RRR^S_7)$ obtained from a
  single path $Q\in
  \QQQ_1^l\setminus \OOO$ which is split into $q_s$ subpaths, then $G$ contains a cylindrical grid of order $k$ as a butterfly minor.
\end{lemma}
\begin{proof}
  Let $\hat{\RRR}_7 \subseteq \{ R \in \RRR^* \sth S(R) \in
\RRR^S_7 \}$. Hence, $\hat{\RRR}_7$ is a linkage
of order $r_7$
from the bottom of the original fence $\FFF$ to its top and
$\hat{\RRR}_7$ and $\QQQ_s$ form a pseudo-fence $\FFF'_p$.

We can now define paths back
from the end vertices of $\hat{\RRR}_7$ to their start vertices as follows.
Let $A_1 \subseteq A$ be the set of end vertices of the paths in $\hat{\RRR}_7$ and
let $B_1\subseteq B$ be set of the start vertices of the paths in $\hat{\RRR}_7$.
Choose a set $\QQQ_d \subseteq \OOO$ of order $r_7$, which is possible as we require
\begin{eqnarray}
t &\geq& r_7.  %
\end{eqnarray}

Let
$A_1' \subseteq A'$ and $B_1' \subseteq B'$ be the set of start and
end vertices of the paths in $\QQQ_d$. Then
there is a linkage $\LLL_u$ of order $r_7$ from  $A_1$ to $A'_1$ in $\FFF_1$ and a linkage
$\LLL_d$ of order $r_7$ from $B_1'$ to $B_1$ in $\FFF_3$. Hence, $\LLL_u\cup \QQQ_d\cup \LLL_d$
form a linkage $\LLL$ of order $r_7$ from $A_1$ to $B_1$. Let
$B_2$ be the start vertices of the paths in $\RRR^S_7$
and $A_2$ be their end vertices.

Every path $R \in
\hat{\RRR}_7$ can be split into three disjoint subpaths, $D(R), S(R), U(R)$,
where $D(R)$ is the initial
component of $R-S(R)$ and $U(R)$ is the subpath following
$S(R)$. Then, $\LLL \cup \bigcup\{ U(R), D(R) \sth R\in \hat{\RRR}_7\}$ form a
half-integral linkage from $A_2$ to $B_2$ of order $r_7$ and hence, by
Lemma~\ref{lem:half-integral}, there is an integral linkage $\LLL'$ of
order $\frac12r_7$ from $A_2$ to $B_2$.  Note that $\RRR^S_7$ and
$\LLL''$ are vertex disjoint. Hence, if
\begin{eqnarray}
\frac12r_7&\geq& f'(k) \\
q_s & \geq & f''(k),
\end{eqnarray}
where $f', f''$ are the functions implicitly defined in
Lemma~\ref{lem:ll2-general}, i.e.~setting $t'$ to $f'(k)$ and $t''$ to $f''(k)$ in the
statement of the lemma yields a cylindrical grid of order $k$,
we can apply Lemma~\ref{lem:ll2-general} to
$(\QQQ_s, \RRR^S_7)$ and $\LLL''$ to obtain a cylindrical grid of
order $k$ as a butterfly minor.
\end{proof}

So we can now assume that instead we get a strong $(q_5, r_7)$-segmentation
$\SSS_2 := (\QQQ_5, \RRR^S_7)$\marginpar{$\SSS_2$}.
Let $\RRR^M_7 \subseteq \RRR_5$\marginpar{$\RRR^M_7$} 
be the paths $R$ in $\RRR_5$ which have a
continuation in $\RRR^S_7$, i.e.~$\RRR^M_7 := \{ M(R) \sth R \in \hat{\RRR}_5$ and $S(R)\in \RRR^S_7\}$. Let $\SSS_1' := (\QQQ^l_1, \RRR^M_7)$\marginpar{$\SSS_1'$} be the
restriction of $\SSS_1$ to these paths $\RRR^M_7$. (Recall that
$\SSS_1$ is the first segmentation obtained above.)  Note that every path
$R$ in $\RRR^M_7$ ends in a vertex $v$ such that the successor of $v$
on $R'$ is the start vertex of a path in $\RRR^S_7$, where $R'$ is the path such that $M(R') = R$.

We now show that if we look at $\SSS'_1$ and $\SSS_2$, they either give 
the situation in Figure \ref{fig:sec6-5} (a) or 
that in Figure \ref{fig:sec6-5} (b).
Then we shall show that this situation does not allow the ``down'' jumps 
as in Figure~\ref{fig:last-claim} (the red arrows). For this purpose, 
having a pseudo-fence in the following lemma is critical. 
 More precisely we prove the following lemma.

  \begin{figure}[t]
    \centering
    \includegraphics{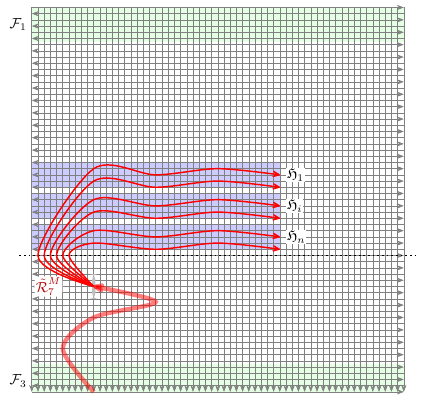}
    \caption{Horizontal strips in the segmentation $\SSS''_1$.}
    \label{fig:sec6-4}
  \end{figure}
\begin{lemma}
  For all $q_6, r_8$ there are minimal values for $r_7$ and $q_5$ such that
  $\QQQ_5 \cup \RRR^M_7 \cup \RRR^S_7$ contains a cylindrical grid of order
  $k$ as a butterfly minor or a pseudo-fence $(\QQQ'_6, \RRR_8)$ for
  some $\RRR_8 \subseteq \RRR^M_7\cup \RRR^S_7$
  of
  order $r_8$ and some $\QQQ'_6\subseteq \QQQ_5$ of order $q_6$.
\end{lemma}
\begin{proof}
  We first consider the pair $(\QQQ_5, \RRR^M_7)$. Note that the paths
  in $\QQQ_5$ are obtained from paths in $\QQQ_1^l$ but possibly by
  splitting paths in $\QQQ_1^l$. Recall that  $(\QQQ_1^l, \RRR^M_7)$
  is a strong segmentation $\SSS'_1$. It follows that $(\QQQ_5, \RRR^M_7)$ is still a
  strong segmentation but it is not necessarily true that every path in
  $\RRR^M_7$ hits every path in $\QQQ_5$. However, to obtain $\QQQ_5$
  from $\QQQ_1^l$, a path in $\QQQ_1^l$ can be split at most $q_s-1$
  times. Hence, if
  \begin{equation}   \label{eq:10}
    r_7 \geq (h_2\cdot h)\cdot q_s^{q_5}
  \end{equation}
  then there is a set $\QQQ_6 \subseteq \QQQ_5$ of order $q_6 \geq
  \frac{q_5}{q_s}$ and a set $\tilde{\RRR}^M_7\subseteq \RRR^M_7$ of
  order $\tilde{r_7}\geq h_2\cdot h$ such that $\SSS_1'' := (\QQQ_6,
  \tilde{\RRR}^M_7)$ is a strong segmentation and every $R\in \tilde{\RRR}^M_7$
  intersects every $Q\in \QQQ_6$.
	
  Let $(R_1, \dots, R_{\tilde{r}_7})$ be an ordering of
  $\tilde{\RRR}^M_7$ in the order in which
  the paths appear on the paths in $\QQQ_6$ from top to bottom.
  We split $\SSS_1''$ into horizontal strips as follows. For all
  $1\leq i \leq h$ let $\HH_i^M := (\QQQ_6, \HHH_i^M)$ where $\HHH_i^M :=
  \{ R_{(i-1)\cdot h_2+1}, \dots, R_{i\cdot h_2}\}$.

  As $r_7 \geq h_2\cdot h$,  every $\HH_i^M$ is itself a strong segmentation using $h_2$ paths of
  $\tilde{\RRR}^M_7$ and the corresponding subpaths of $\QQQ_6$.
  See Figure~\ref{fig:sec6-4} for an illustration.

  For every $\HH_i^M$ let $\HHH_i^S\subseteq \RRR^S_7$ be the paths in
  $\RRR^S_7$ whose start vertex is the successor of the end vertex of a path in
  $\HHH_i^M$. We define $\HH^S_i := (\QQQ_6, \HHH^S_i)$. Again, this is a
  strong segmentation. Furthermore, every horizontal path $R\in \HHH_i^M$ can
  be continued by a path in $\HHH^S_i$.

  By construction of $\QQQ''$, for
  every $1\leq i \leq h$ and for every $Q\in \QQQ_6$, at most $q^*$
  paths in $\HHH^S_i$ can contain a vertex $v\in V(Q)$ such that $v$ appears on
  $Q$ after the last vertex $Q$ has in common with any
  path in $\HHH_i^M$.

  Hence, we can take a subset $\HHH'^S_i\subseteq \HHH^S_i$ of order
  $h_3 := h_2-q_6\cdot q^*$ such that no
  path in $\HHH'^S_i$ contains a vertex $v\in V(Q)$, for any $Q\in
  \QQQ_6$, which appears after the last vertex $Q$ shares with
  $\HHH_i^M$. We now claim that the horizontal strips must intersect
  nicely as illustrated in Figure~\ref{fig:sec6-5} b).

  \begin{Claim}
    $\QQQ_5 \cup \RRR^M_7\cup \RRR^S_7$ contains a cylindrical
    grid of order $k$ as a butterfly minor or there is a subset $\QQQ_5'\subseteq \QQQ_6$ of
    order $q_5'$ and for every $1\leq i \leq
    h$ a subset $\hat\HHH^S_i \subseteq
    \HHH'^S_i$ of order $h_4$, for some suitable numbers $q_5'$ and $h_4$
    to be determined below, such that every $R\in  \hat\HHH^S_i$
    intersects every  $Q\in \QQQ'_5$ in the subpath of $Q$ between the
    top path $R_{(i-1)\cdot h_2+1}$ and the lowest path $R_{i\cdot
      h_2}$ in $\HHH_i^M$.
  \end{Claim}
  \begin{ClaimProof}
    For every $R\in \HHH'^S_i$ let $\pi_i(R)$ be the set of paths $Q\in
    \QQQ_6$ such that $R$ intersects $Q$ only in vertices which occur
    on $Q$ before the first vertex $Q$ has in common with
    $\HHH_i^M$. Now suppose there are at least $\tilde h\cdot \binom{q_6}{q_7}$ paths
    $R\in \HHH'^S_i$ with $|\pi_i(R)| \geq q_7$, for some numbers
    $\tilde{h}$ and $q_7$
    to be determined below. By the pigeon hole
    principle, there is a set $\tilde\HHH^S_i\subseteq
    \HHH'^S_i$of order $\tilde h$ such that $\pi_i(R) = \pi_i(R')$ for
    all $R\in \tilde\HHH^S_i$ and $|\pi_i(R)|\geq q_7$. We
    claim that in this case we obtain a cylindrical grid of order
    $k$. The construction is illustrated in Figure~\ref{fig:sec6-5} a).

    \begin{figure}[t]
      \centering
      \hspace*{-0.5cm}
      \begin{tabular}{c@{\hspace*{1cm}}c}
        \includegraphics[height=5cm]{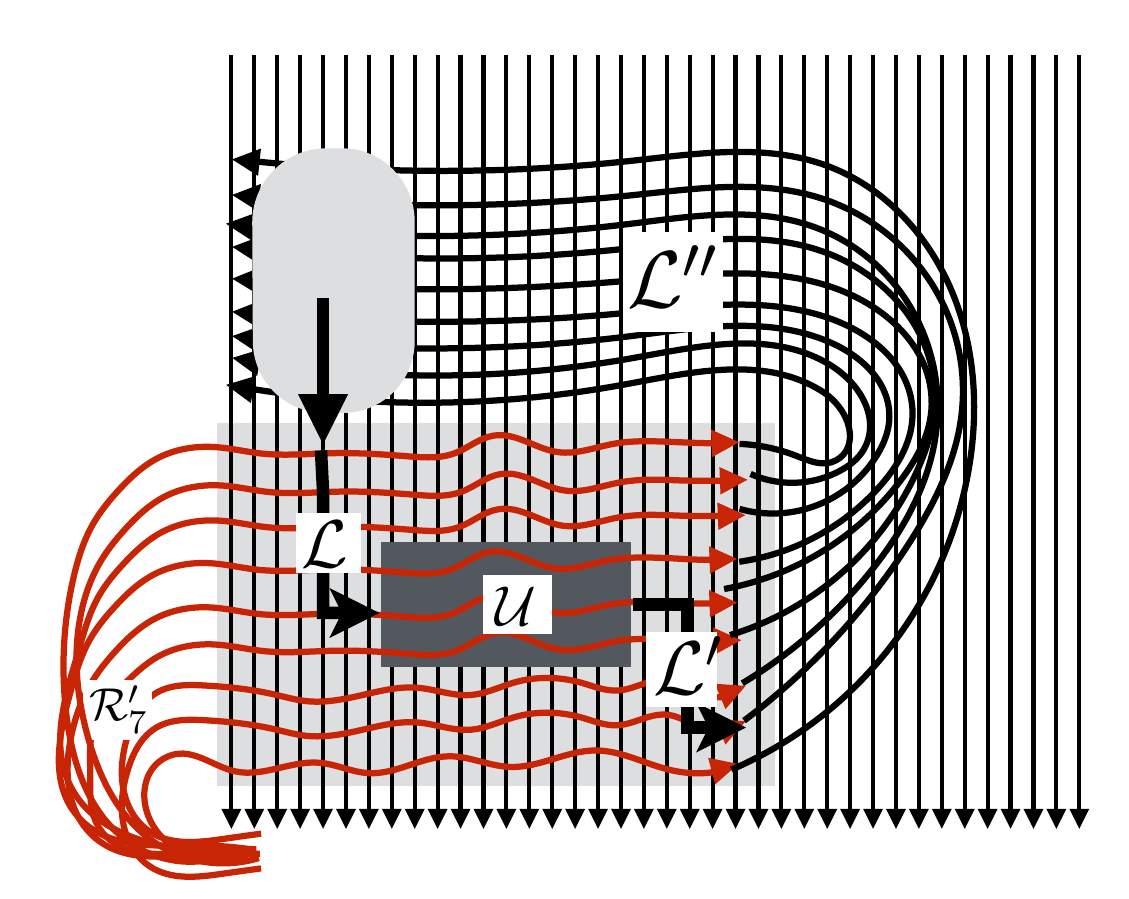} &
		 \includegraphics[height=8cm]{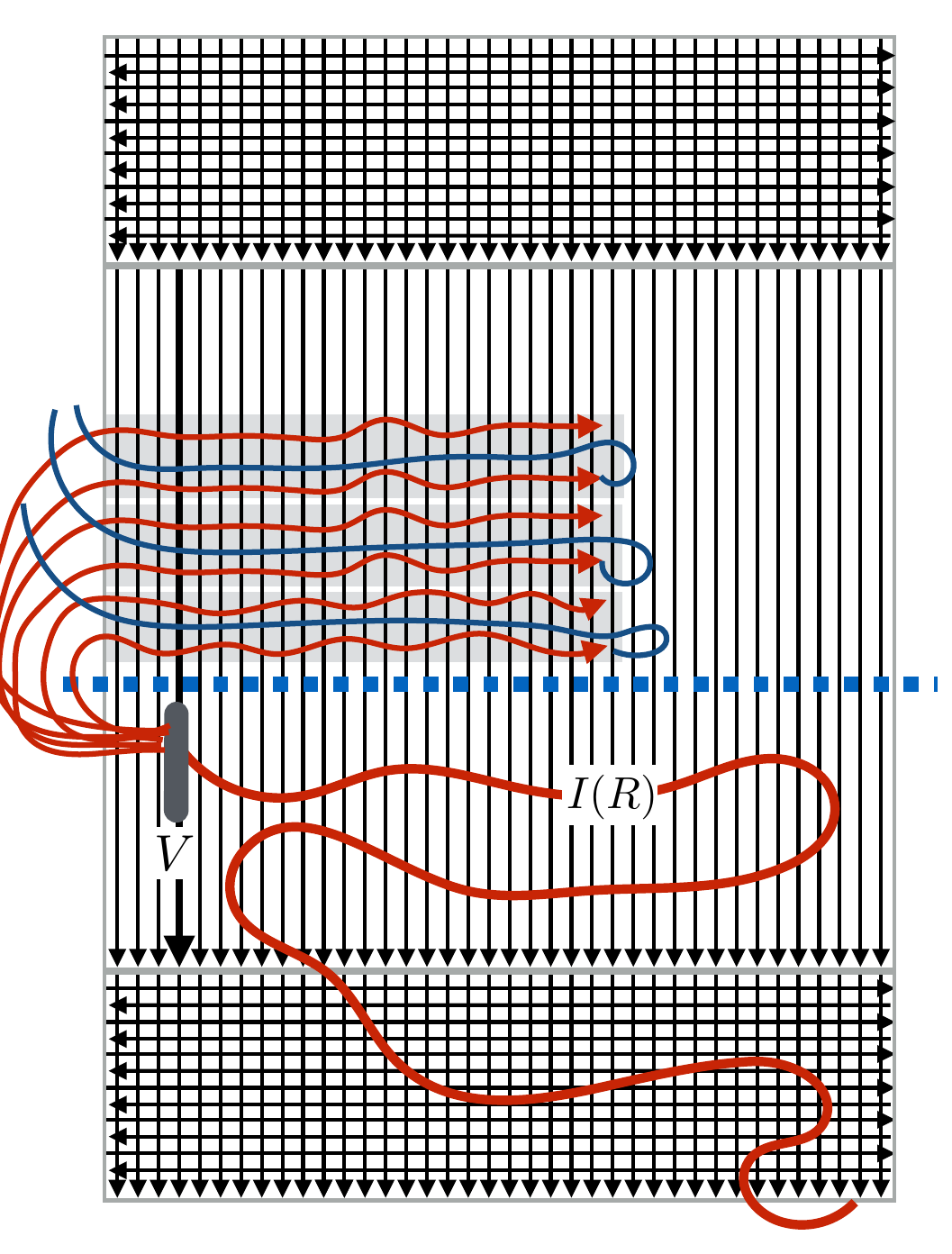}\\
        a) & b)
      \end{tabular}
      \caption{a) Creating a cylindrical grid from two disjoint horizontal
        strips and b) two segmentations $\SSS'_1$ and $\SSS_2$ forming a pseudo-fence.}
      \label{fig:sec6-5}
    \end{figure}

    Let $\VVV := \pi_i(R)$ for some (and hence all) $R\in
    \tilde\HHH^S_i$. Let $\tilde\HHH_i^M\subseteq \HHH_i^M$ be the set of
    paths in $\HHH_i^M$ ending in the predecessor of a start vertex of a path in
    $\tilde\HHH^S_i$. Finally, let $\tilde\QQQ^M$ be the set of minimal subpaths
    of paths $Q\in\VVV$  containing every vertex $Q$
    has in common with $\tilde\HHH_i^M$. By construction, every $Q\in
    \tilde\QQQ^M$ is disjoint from every $R\in \tilde\HHH_i^S$. For
    every $Q\in \VVV$ we can therefore take the subpath $i(Q)$ from
    the beginning of $Q$ to the predecessor of the first vertex $Q$
    has in common with $\tilde\HHH_i^M$. Let $\tilde\QQQ^S := \{ i(Q) \sth Q\in \VVV\}$. Then $\tilde\QQQ^S$ and $\tilde\HHH_i^S$ form a
    strong $(q_7, \tilde{h})$-segmentation.
    By Lemma~\ref{lem:segmentation-grid} using Property (2),
    if
    \begin{eqnarray}
      \label{eq:16}
      \tilde h &\geq & f_r(h_6)\cdot \binom{f_p(3h_6)}{3h_6}\cdot (3h_6)!\cdot
                       4h_6\\
      q_7 & \geq & f_p(3h_6)\cdot \tilde h!,
    \end{eqnarray}
   for some suitable number $h_6$ to be determined below,
    then
    $(\tilde\QQQ^M,
    \tilde\HHH_i^M)$ contains an acyclic $(h_6, h_6)$-grid $\GGG := (\VVV_\GGG,
    \HHH_\GGG)$  such that the
    paths $\HHH_\GGG$ are obtained from subpaths of $\tilde\QQQ^M$
    and $\tilde\HHH_i^M$ preserving the end vertices of the paths in
    $\tilde\HHH_i^M$. Here $f_r, f_p$ are the functions defined in
    Lemma~\ref{lem:new-good-tuples}. Let $\VVV_\GGG := (V_1, \dots, V_{q_7})$ be ordered
    in the order in which they appear on the paths in $\HHH_\GGG$ and
    let
    $\HHH_\GGG := (H_1, \dots, H_{h_6})$ be ordered in the order in which
    they appear on $\VVV_\GGG$. 

    The following argument is illustrated in  Figure~\ref{fig:sec6-5} a). To this end let $\UUU$ be the subgrid of $\GGG$ formed by $(\VVV_\UUU,
    \HHH_\UUU)$ where $\VVV_\UUU$ is the set of minimal subpaths of
    $\VVV_\GGG$ to include every vertex of $H_{\frac13 h_6}, \dots,
    H_{\frac23h_6}$ and $\HHH_\UUU$ are the minimal subpaths of
    $H_{\frac13 h_6}, \dots,
    H_{\frac23h_6}$ including every vertex they have in common with
    $\VVV_\UUU$. Note that the top and the bottom of this grid are the
    endpoints of the paths in $\HHH_\UUU$, i.e.~the grid is
    ``tilted''.

    Then
    in $\GGG$ there is a linkage $\LLL$ of order $\frac13h_6$ from the
    start vertices of $V_1, \dots, V_{\frac13h_6}$
    to the top of $\UUU$, i.e.~the start vertices of $\HHH_\UUU$,
    and a linkage $\LLL'$ of order $\frac13h_6$ from the bottom of
    $\UUU$, i.e.~the end vertices of $\HHH_\UUU$, to the end vertices of $H_{\frac23h_6}, \dots,
    H_{h_6}$. Furthermore, in $\tilde\HHH^S_i\cup\tilde\QQQ^S$ which forms a
    strong $(q_7, \tilde{h})$-segmentation, as mentioned above, there is
    a linkage $\LLL''$ from the end vertices of $H_{\frac23h_6}, \dots,
    H_{h_6}$ to the start vertices of $V_1, \dots, V_{\frac13h_6}$. As
    $\LLL, \LLL', \LLL''$ are pairwise disjoint except for the
    end vertices they have in common, they form a linkage $\LLL'''$ from
    the bottom of $\UUU$ to the top which is disjoint from
    $\UUU$. We require that
    \begin{equation}
      \label{eq:5}
      h_6 \geq 3\hat{k},
    \end{equation}
    where $\hat{k}$ is the integer defined in Lemma~\ref{tech00} (called
    $t'$ in the statement of the lemma).
    We can now apply Lemma~\ref{tech00} to obtain a cylindrical grid of
    order $k$ as a butterfly minor.
  \end{ClaimProof}

  By the previous claim, in every $\HHH_i^M$ there is a path $R\in
  \HHH_i^M$ and a path $R' \in \HHH^S_i$ such that the endpoint of $R$ is
  the start vertex of $R'$ and a set $\gamma(R') \subseteq \QQQ_6$ of order
  $q_7$ such that $R'$ hits every path $Q\in \QQQ'_i$ within
  $\HHH_i^M$. For all $1\leq i \leq h$ we choose such a path $R_i$ and
  $R'_i$. Note that $\SSS_2$ is a strong segmentation of $\QQQ''$, hence no
  path $R'_i$ can intersect any $Q\in \QQQ'_i$ at a vertex $v$ which
  occurs on $Q$ before a vertex $w\in V(Q)\cap V(R'_j)$ for some
  $j<i$.

  We require
  \begin{equation}
    h\geq {q''\choose q_7}\cdot r_8.\label{eq:12}
  \end{equation}
  Thus, we can choose a set $\RRR_8$ of paths
  $R_i$ and $R'_i$ such that $\gamma(R'_i) = \gamma(R'_j)$ for all
  $R'_i, R'_j\in \RRR_8$. Let $\QQQ'_6:= \gamma(R'_i)$ for some (and
  hence all) $R'_i\in \RRR_8$. Hence, $\RRR_8$ and
  $\QQQ'_6$ form a pseudo-fence as required.
\end{proof}

Finally, we are ready to finish the proof. 
Suppose now the previous lemma does not result in a cylindrical grid
of order $k$. Hence, we now have a pseudo-fence $(\QQQ'_6, \RRR_8)$\marginpar{$(\QQQ'_6, \RRR_8)$} for
some $\RRR_8 \subseteq \RRR^M_7\cup \RRR^S_7$
  of
  order $r_8$ and some $\QQQ'_6\subseteq \QQQ_5$ of order $q_6$.
The current situation is illustrated in Figure~\ref{fig:sec6-5}
b).
  
We shall show that the current situation does not allow the ``down'' jumps 
as in Figure~\ref{fig:last-claim} (the red arrows). For this purpose, 
having a pseudo-fence $(\QQQ'_6, \RRR_8)$ is critical. 

Recall Definition~\ref{def:v(R)} of the vertex $v(Q)$. 
Let $V\in \QQQ''$\marginpar{$V$} be the path such that every $R\in \RRR_3$ contains a
good vertex $v(R)$ on $V$. We define $\QQQ_7 := \QQQ'_6\cup \{ V \}$\marginpar{$\QQQ_7 := \QQQ'_6\cup \{ V \}$}. 
Now, $\RRR_8$ and $\QQQ_7$ are no longer a pseudo-fence, but they are a
pseudo-fence in restriction to $\QQQ'_6$ and
furthermore, every path $R \in \RRR_8$ also intersects
$V$.

Recall that $\RRR_8$ is a set of paths $R_i\in \RRR^M_7$ and $R'_i\in
\RRR^S_7$ such that $R'_i$ is the continuation of $R_i$, i.e.~there is
a path $\hat{R}_i\in \RRR^*$ and $R_i, R'_i$ are subpaths of $R$ such
that the start vertex of $R'_i$ is the successor on $\hat{R}_i$ of the
end vertex of $R_i$. Let $(R'_1, \dots, R'_{r_8})$ be an ordering of
the paths $R'_i\in \RRR_8\cap \RRR^S_7$ in the order in which they
occur on the paths in $\QQQ'_6$.
We require
\begin{equation}
r_8 \geq (h'_9)^2,
\end{equation}
for some value
of $h'_9$ to be determined below.
As in the proof of the previous
lemma we
define horizontal strips $\HHH_i := \{ R_{(i-1)h'_9+1}\cup R'_{(i-1)h'_9+1}, \dots,
  R_{ih'_9}\cup R'_{ih'_9}\}$, for all $1\leq i
  \leq h'_9$, and let $\VVV_i := \{m_i(Q) \sth Q\in
  \QQQ_6'\}$ where $m_i(Q)$\marginpar{$\VVV_i, m_i(Q)$} is the minimal subpath of $Q$ containing
  every vertex  of $V(\HHH_i)$.
Recall from above that every path $R\in \RRR^*$ is split into three
distinct parts, $I(R)$, $M(R)$ and $S(R)$. The subpaths $M(R)$ and
$S(R)$ are part of the construction of $\RRR_8$,
where the $M(R)$ play
the role of the $R_i$ above and the $S(R)$ play
the role of
$R'_i$.  We will now use the initial subpaths $I(R)$. Recall further
that the endpoint of each $I(R)$ for $R \in
\RRR_3$ is on the path $V$.

We require that
   \begin{equation}
   h'_9\geq 2q^*+2\label{eq:11}
 \end{equation}
to make sure that the following lemma holds. 
\begin{lemma}
  There is a $1\leq i \leq h'_9$ such that $I := \{ I(R) \sth M(R)\in
  \HHH_i\}$ is disjoint from $\HHH_i\cup\VVV_i$.
\end{lemma}
\begin{proof}
   Towards a contradiction, suppose the claim was false. For every
   $1\leq i \leq h'_9$ choose a path $M(\hat{R}_i)\in \HHH_i$ such that $I(\hat{R}_i)$
   intersects $\HHH_i\cup\VVV_i$. As $I(\hat{R}_i)$
   ends in $V$, in fact ends in
   $l(V)$, and furthermore, every path $R\in
   \HHH_i$ intersects $u(V)$,
   this implies that there is a path $P_i$ from $u(V)$ to $l(V)$ in
   $\HHH_i\cup\VVV_i\cup I(\hat{R}_i)$. Note that for $i\not=j$ the paths
   $P_i$ and $P_j$ may not be disjoint, as, e.g., $I(R_i)$ may intersect
   $\HHH_i$ and $\HHH_j$.

   However, the set $\{ P_i \sth 1\leq i \leq h'_9\}$ forms a
   half-integral linkage from $u(V)$ to $l(V)$ and therefore, by
   Lemma~\ref{lem:half-integral}, there also is an integral linkage of
   order $\frac12h'_9$. See Figure~\ref{fig:last-claim} for an
   illustration.
   \begin{figure}[t]
     \centering
     \includegraphics[height=7cm]{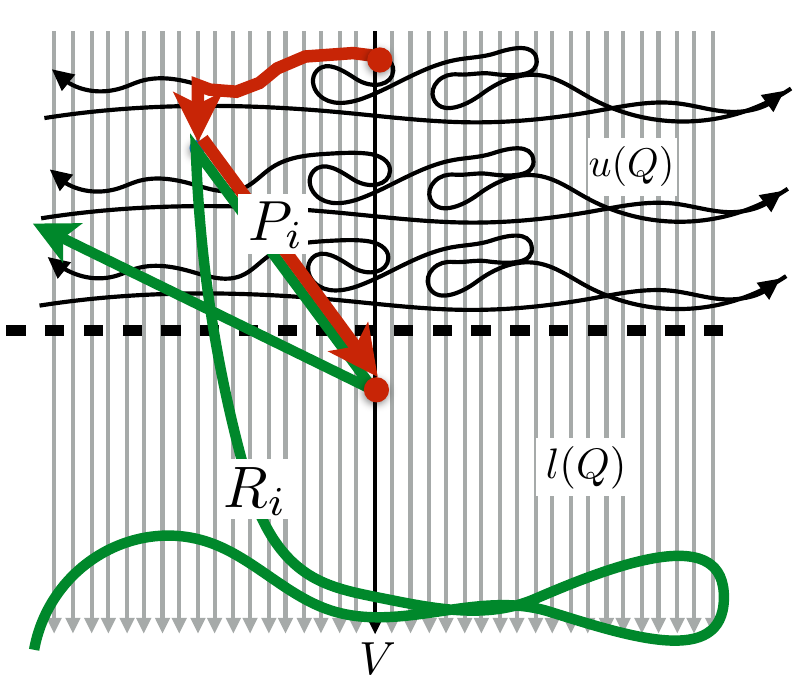}
     \caption{Illustration for the last lemma.}
     \label{fig:last-claim}
   \end{figure}

As $h'_9\geq 2q^*+2$,
 this contradicts the
   fact that in $\QQQ''$ at most $q^*$
   paths can go from some $u(Q)$ to $l(Q)$, see Lemma \ref{lem:reroute-R*}.
\end{proof}

Let $i\leq h'_9$ be such that $I := \{ I(\hat{R}) \sth M(\hat{R})\in \HHH_i\}$ is
disjoint from $\HHH_i\cup\VVV_i$. Let $\HHH \subseteq  \{ M(R) \sth M(R)\in
\HHH_i\}$ be a set of order $h_{10}$, for some $h_{10}$ to be
determined below. Note that $\HHH$ is a subset of $\HHH_i$ with the paths $S(R)$ removed,
which are no longer needed.
We require that
\begin{eqnarray}
h'_9 & \geq & h_{10}\\
h_{10} & \geq & f_r(s)\binom{f_p(3s)}{3s} \cdot (3s)! \cdot 4s\\
q_6 & \geq & f_p(3s)\cdot h_{10}!
\end{eqnarray}
where $s = 18s'$ and $s'$ is such that if in Lemma~\ref{tech00} we set
$t'$ to $s'$ and thus $t$ to $6s$ then the lemma implies a
cylindrical grid of order $k$ as a butterfly minor.
 Here, $f_r, f_p$ are the functions defined in Lemma~\ref{lem:new-good-tuples}.

By Lemma
\ref{lem:segmentation-grid}, $(\VVV_i, \HHH)$ contains an $(s,s)$-grid
$\UUU := (\VVV_i', \hat\HHH)$ which can be chosen
so that the start vertices of the paths
$M(R)$ are preserved. Let $(H_1, \dots, H_{s})$ be an ordering of
$\hat\HHH$ in the order in which they occur on the
paths in $\VVV'_i$ and
let $(V_1, \dots, V_{s})$ be an ordering of $\VVV'_i$
in the order in which
the paths occur on the paths in $\HHH$.

We now take the subgrid
$\UUU'$ induced by $(H_{\frac13s+1}, \dots, H_{\frac23s})$
and $(V_{\frac13s+1}, \dots, V_{\frac23s})$.
More precisely, for every $H\in \{H_{\frac13s+1}, \dots,
H_{\frac23s}\}$ let $\rho(H)$ be the minimal subpath of $H$
containing all of $H\cap \bigcup\{V_{\frac13s+1}, \dots,
V_{\frac23s}\}$ and for all $V'\in
\{V_{\frac13s+1}, \dots,
V_{\frac23s}\}$ let $\rho(V')$ be the minimal
subpath of $V'$
containing all of $V'\cap \bigcup \{H_{\frac13s+1}, \dots,
H_{\frac23s}\}$. Then $\UUU'$ is the grid induced by $\{
\rho(\{H_{\frac13s+1}), \dots, \rho(H_{\frac23s})\}$ and $\{
\rho(V_{\frac13s+1}), \dots, \rho(V_{\frac23s})\}$.
Let $(t_1, \dots, t_{\frac13s})$ be the start vertices and $(b_1,
\dots, b_{\frac13s})$ be the end vertices of the
paths $(
\rho(\{H_{\frac13s+1}), \dots, \rho(H_{\frac23s}))$. Let $T :=
\{ t_1, \dots, t_{\frac19s}\}$ and $B := \{ b_{\frac13s-\frac19s},
\dots, b_{\frac13s})$. 

We can now construct a linkage from $B$ to $T$ as
follows. Let $\III'$ be the set of paths $I(R)$
with end vertex in
$T$. By construction, every $I(R)$ intersects every $l(Q)$ for $Q\in
\{ V_{\frac23s}, \dots, V_{s}\}$, but does not intersect any vertex in  $\UUU'$. Hence, $\{ V_{\frac23s}, \dots,
V_{s}\} \cup \{ H_{\frac23s}, \dots,
H_{s}\}\cup \III'$ contains a half-integral linkage from
$B$ to $T$, and therefore by Lemma~\ref{lem:half-integral}, also an
integral linkage $\LLL$ from $B$ to $T$ of order $\frac16s$.

By our choice of $s$, Lemma~\ref{tech00} implies that $\UUU'$
together with $\LLL$ contains a cylindrical grid of order $k$ as a butterfly minor.
This completes the proof of  Theorem~\ref{thm:main-cylindrical} and    hence the proof of
Theorem~\ref{thm:main-bramble} and therefore Theorem~\ref{thm:main}.

\section{Conclusion}

In this paper we proved the directed grid conjecture by Reed and Johnson, Robertson, Seymour and Thomas. We view this result as a first but
significant step towards a more general structure theory for directed graphs based on
directed tree-width, similar to the grid theorem \cite{GMV} for
undirected graphs being the basis of more general structure
theorems \cite{GMXVI}.

Our proof indeed yields the following algorithmic result, which is perhaps of independent interest.
\begin{quote}
There is a function $f\sth \N\rightarrow \N$
  such that given any directed graph and any fixed constant $k$, in polynomial time, we can obtain either
  \begin{enumerate}
  \item
  a cylindrical grid of order $k$ as a butterfly minor or
  \item
  a directed tree decomposition of width at most $f(k)$.
  \end{enumerate}
  \end{quote}

In fact, since our cylindrical grid is obtained from two linkages $\PPP, \QQQ$, together with $\R$, such that all of $\PPP, \QQQ, \R$ are disjoint paths joining two vertices of the well-linked set, our proof implies the following. 
\begin{theorem}
  There is a function $f\sth \N\rightarrow \N$
  such that given a well-linked set $F$ of order $f(k)$ in any directed graph and any constant $k$, 
  there is a cylindrical wall $W$ of order $k$, such that for any $k$ vertices of degree at least three (i.e., vertices of either outdegree at least two or indegree at least two), there are $k$ disjoint paths from $F$ to $W$ and from $W$ to $F$
  
  In addition, such a cylindrical grid can be found in polynomial time if $k$ is fixed.
\end{theorem}

This is indeed the analogue of the main result in \cite{GMV}. 

Since the first version of this paper appears in STOC'15, there is some progress made for the directed graph structure theory and algorithms, building on our directed grid theorem. Firstly, building on our directed grid theorem, the directed version of the flat wall theorem (the weak structure theorem) in \cite{GMXIII} is obtained in \cite{flatdirect}. This is a significant step towards the directed version of the main structure theorem in \cite{GMXVI}, as in \cite{KTW}, the flat wall theorem is the base case for the (new) proof of the main graph minor structure theorem in \cite{GMXVI}.

 Secondly, the tangle tree-decomposition theorem, proved by Robertson and
 Seymour in \cite{GMX} in the graph minor series, turns out to be an
 extremely valuable tool in structural and algorithmic graph theory.  
 In \cite{tangletree}, the authors introduce directed tangles and provide a directed tree-decomposition of digraphs $G$ that distinguishes all maximal directed tangles in $G$. Furthermore, for any integer~$k$, they construct a directed tree-decomposition that distinguishes all directed tangles of order $k$ (and the construction results in a polynomial time algorithm for fixed $k$). 
 Building on our directed grid theorem and this result, we can decide for a given
digraph $G$ and $k$ pairs $(s_1, t_1), \dots, (s_k, t_k)$ of vertices
either there are directed paths $P_1, \dots, P_k$ such that
$P_i$ links $s_i$ to $t_i$ and such that no vertex of $G$ is contained
in more than two paths from $\{P_1, \dots, P_k\}$, or there are no $k$ directed disjoint paths $P_1, \dots, P_k$ such that $P_i$ links $s_i$ to $t_i$, see \cite{tangletree}. This improves the previous result in \cite{KawarabayashiKK14}. 
 
We also believe that this theorem will prove to be very useful for further applications of directed tree-width, for instance
to Erd\H os-P\'osa type results for directed graphs. 
Indeed, building on our direcred grid theorem,  a sufficient condition when Erd\H os-P\'osa type results holds for strongly connected directed graphs is given in \cite{epg}.
Furthermore, it is likely that the duality of directed
tree-width and directed grids will make it possible to
develop algorithm design techniques such as bidimensionality theory or
the irrelevant vertex technique for directed graphs. We are
particularly optimistic that this approach will lead to new results for the directed disjoint paths problem. We leave this for future research.

\bibliographystyle{abbrv}

\end{document}